\documentclass[english,12pt]{article}
\usepackage[T1]{fontenc}
\usepackage{verbatim}
\usepackage{amsmath}
\usepackage{amsthm}
\usepackage{amssymb}
\usepackage{algorithm}
\usepackage{algpseudocode}
\usepackage{appendix}
\usepackage{natbib}

\usepackage{amsmath,amssymb,amsthm,bm}

\usepackage{graphicx,subfig,psfrag,epsf}
\usepackage{multirow}
\usepackage{enumerate,array,booktabs}
\usepackage[hidelinks]{hyperref}
\usepackage{enumitem}
\usepackage{caption}
\usepackage{rotating}
\usepackage {graphicx}
\usepackage{amsmath,color,xcolor}
\usepackage{amsfonts}
\usepackage{dsfont}
\usepackage{subfig}
\usepackage{multirow}
\usepackage{pdflscape,geometry}
\usepackage{alltt}
\usepackage[T1]{fontenc}
\usepackage{bm}
\usepackage{pifont}
\usepackage{booktabs}
\IfFileExists{upquote.sty}{\usepackage{upquote}}{}
\usepackage{tikz}
\usetikzlibrary{arrows.meta, positioning, shapes.geometric,decorations.markings}

\usepackage{setspace}
\makeatletter
\def\singlespace{\def\baselinestretch{1}\@normalsize}

 \makeatletter
 \g@addto@macro\normalsize{%
 	\setlength{\abovedisplayskip}{3.0pt plus 2.0pt minus 5.0pt}%
 	\setlength{\belowdisplayskip}{5.0pt plus 2.0pt minus 5.0pt}%
 }
 \makeatother
 
 \makeatletter
\theoremstyle{plain}
\newtheorem{thm}{\protect\theoremname}
\theoremstyle{plain}
\newtheorem{lem}{\protect\lemmaname}
\theoremstyle{definition}
\newtheorem{defn}{\protect\definitionname}
\newtheorem{remark}{Remark}

\DeclareMathOperator{\E}{\mathbb{E}}
\DeclareMathOperator{\I}{\mathbb{I}}

\makeatother

\usepackage{babel}
\providecommand{\definitionname}{Definition}
\providecommand{\lemmaname}{Lemma}
\providecommand{\theoremname}{Theorem}

\newcommand{\blind}{1} 
\begin{document}
\global\long\def\bU{\mathbf{U}}%
\global\long\def\bX{\bm{X}}%

\global\long\def\ntr{n_{\textrm{train}}}%

\global\long\def\nte{n_{\textrm{test}}}%

\if1\blind
{ 
	\title{\large
		\bf Stable Causal Discovery via Directed Acyclic Graph Aggregation}  
	\author{Yunan Wu, Yue Wang, Chunlin Li and Chenglong Ye}
\date{}

\maketitle

\begin{singlespace}
	\begin{footnotetext}[1]
		{
			Yunan Wu is an Assistant Professor at the Yau Mathematical Sciences Center, Tsinghua University. Email: wuyunan@mail.tsinghua.edu.cn. 
			Yue Wang is an Assistant Professor in the Department of Biostatistics \& Informatics and the Center for Innovative Design and Analysis, University of Colorado Anschutz Medical Campus. Email: yue.2.wang@cuanschutz.edu. 
			Chunlin Li is an Assistant Professor in the Department of Statistics, University of Virginia. Email: chunlin@virginia.edu. 
			Chenglong Ye is an Assistant Professor in the Dr. Bing Zhang Department of Statistics, University of Kentucky. Email: chenglong.ye@uky.edu. 
		}	
	\end{footnotetext}
\end{singlespace}

} \fi

\if0\blind
{
\bigskip
\bigskip
\bigskip
\begin{center}
	{\large\bf Stable Causal Discovery via Directed Acyclic Graph Aggregation}
\end{center}
\medskip
} \fi

 
\begin{abstract}
Directed Acyclic Graphs (DAGs) are central to uncovering causal structure in complex systems, yet learning a single DAG from data is often challenging: model uncertainty, finite samples, and a combinatorially large search space frequently yield unstable estimates. We propose DAGgr, a model averaging framework that aggregates multiple candidate DAGs into a single stable representation. Candidate graphs are weighted by their out-of-sample predictive likelihood across repeated data splits, and a thresholding rule on the resulting edge-importance scores guarantees that the aggregated graph is itself acyclic. We establish a finite-sample risk bound, prove that the procedure preserves acyclicity, and show that edge selection is consistent under mild conditions on the weights. Simulations across random, hub, and chain structures, together with an analysis of the \citet{sachs2005causal} protein-signaling network, show that DAGgr matches or exceeds the best individual candidate while consistently outperforming bootstrap-aggregation baselines across structural recovery metrics.
\end{abstract}
\noindent%
\section{Introduction}
\label{sec:intro}
Directed Acyclic Graphs (DAGs) are a foundational tool for representing
dependency and causal structure across the sciences, with established
applications in genomics and gene-regulatory network inference
\citep{sachs2005causal, friedman2000expression}, neuroscience
\citep{ramsey2017fges}, epidemiology \citep{greenland1999causal},
the social sciences \citep{spirtes2000causation}, and machine learning
\citep{koller2009pgm, pearl2009causality}. By encoding variables as
nodes and direct influences as directed edges, DAGs support tasks
ranging from confounder identification to counterfactual reasoning and
policy evaluation \citep{pearl2009causality, hernan2020causal}.

Despite this broad utility, learning a DAG from observational data
remains notoriously difficult. The number of possible DAGs grows
super-exponentially in the number of variables $p$, rendering
exhaustive search infeasible beyond roughly 25 nodes
\citep{chickering2002optimal, koivisto2004}. The literature has
converged on three main families of methods. \emph{Score-based}
methods optimize a global criterion such as BIC or Bayesian marginal
likelihood under acyclicity constraints
\citep{chickering2002optimal, andrews2023boss, huang2018generalized}.
\emph{Constraint-based} methods, including PC and FCI, recover
structure through repeated conditional independence tests
\citep{spirtes2000causation, kalisch2007estimating, colombo2014pcstable}.
\emph{Hybrid} methods such as Max-Min Hill-Climbing
\citep{tsamardinos2006max} combine the two. More recently,
\emph{continuous-optimization} approaches have reformulated structure
learning as a smooth program over weighted adjacency matrices: NOTEARS
\citep{zheng2018dags} introduced a differentiable acyclicity
constraint via the trace of a matrix exponential, which has since been
extended to nonlinear and nonparametric settings
\citep{zheng2020learning, lachapelle2019gradient} and refined through
alternative acyclicity characterizations such as the log-determinant
formulation in DAGMA \citep{bello2022dagma} and identifiability-aware variants \citep{pmlr-v258-berrevoets25a}. Practical issues, such as finite samples, optimization
landscapes with poor local minima
\citep{wei2020nofears, ng2024closer}, and mild misspecification
routinely yield unstable structural estimates. Small
perturbations to the data or the estimation procedure can produce
markedly different graphs, and many DAGs may fit the data nearly
equally well \citep{wang2014learning, kitson2023comparison}. 

Recent benchmarking studies have further documented that the relative performance of individual algorithms varies widely and unpredictably
across datasets \citep{kitson2023comparison, theLandscape2024}, and
that simulated benchmarks can be deceptively easy to game
\citep{reisach2021beware, reisach2023scale}. Reporting a single ``best'' DAG therefore risks elevating sampling artifacts or algorithmic idiosyncrasies to the status of causal claims, undermining
the interpretability and credibility of downstream inference.

Model averaging offers a principled response to this instability.
Rather than committing to one structure, averaging integrates evidence
across multiple plausible candidates---obtained from different
algorithms, tuning parameters, or subsamples---to reduce variance and
produce more balanced estimates of the data-generating process. The idea has a long pedigree in statistics, with stability selection \citep{meinshausen2010stability} providing finite-sample
false-discovery control through subsampling, and Bayesian model
averaging offering a coherent probabilistic treatment for graphical
models \citep{friedman2003being, madigan1996bayesian}. The
frequentist counterpart for undirected Gaussian graphical models has also been developed recently \citep{10.1111/biom.13758}. Direct extensions of these ideas to DAG learning, however, are non-trivial:
naively averaging adjacency matrices can violate acyclicity and destroy the structural properties that make DAGs causally interpretable in the first place.

Existing aggregation strategies for DAGs fall into three broad
categories. The first relies on \textbf{bootstrap aggregation}: \citet{wang2014learning} introduced DAGBag, which
runs a score-based learner on bootstrap resamples and aggregates the
resulting graphs by minimizing structural-Hamming distance over the
ensemble. \citet{chowdhury2022dagbagm} extended this to mixed continuous and
binary nodes (DAGBagM), and \citet{debeire2024bagged} adapted bagging
to time-series causal discovery (Bagged-PCMCI+) using majority voting
on edge confidence. The second category uses \textbf{edge-frequency
thresholding}: edges are scored by the fraction of resamples in which they
appear, and high-confidence edges are retained
\citep{friedman1999bootstrap, imoto2002bayesian}. This idea predates
formal acyclicity-preserving aggregation and provides no guarantee
that the union of high-frequency edges forms a DAG. The
Model-Averaging Hill-Climbing (MAHC) algorithm of
\citet{constantinou2019model} addresses this by integrating averaging
into the search itself. Bayesian model averaging
\citep{madigan1996bayesian, friedman2003being} provides the most
principled probabilistic treatment, but its computational cost grows
sharply with $p$ and it does not naturally accommodate non-Bayesian
candidate learners such as NOTEARS or DAGMA. Stability selection
\citep{meinshausen2010stability} offers a related frequentist tool
with finite-sample false-discovery control via subsampling; our edge
importance score and consistency results are in this spirit, but
adapted to the DAG setting and combined with predictive-likelihood
weighting. \citet{10.1111/biom.13758} have recently developed
frequentist model averaging for undirected Gaussian graphical models;
to our knowledge an analogous treatment for DAGs with
acyclicity-preservation guarantees and edge-selection consistency
has not been given. The
third and most recent category formulates aggregation itself as an
\textbf{optimization or knowledge-elicitation problem}.
\citet{aslani2023ensemble} cast aggregation as a graph-distance
optimization that maximizes the marginal contribution of each edge,
while \citet{tench2026dynamic} use dynamically requested expert
knowledge---including LLMs as imperfect experts---to mediate
disagreements between candidate algorithms. \citet{caelen2025higher}
move beyond pairwise edge frequencies by aggregating over higher-order
edge structures. 

These approaches have improved stability, but share three limitations. First, most rely on edge-frequency or majority-voting heuristics that ignore differences in candidate quality, treating poorly-fitting and strongly-fitting candidates symmetrically. Second, the structural metrics they aggregate over---typically Hamming distance or raw edge frequencies---do not exploit out-of-sample predictive performance, which is a direct measure of which candidates the data favor. Third, the theoretical link between the aggregation rule and the consistency of edge recovery is typically weak or absent, leaving users without guidance on when averaging will help or how to choose threshold parameters. Approaches that do offer principled weighting (Bayesian model averaging) face computational and acyclicity-preservation challenges in continuous, high-dimensional settings.

In this paper, we propose a DAG model averaging framework that
addresses these gaps. Candidate DAGs are weighted by out-of-sample
predictive likelihood over repeated data splits. The framework is agnostic to the
underlying learner and integrates naturally with score-based,
constraint-based, hybrid, and continuous-optimization methods.

\paragraph{Contributions.} Our main contributions are as follows.
\begin{enumerate}
    \item \emph{Acyclicity-preserving aggregation.} Our
framework keeps the candidate set heterogeneous and weights each
complete candidate DAG by its held-out likelihood, which lets
poorly-fitting candidates be down-weighted automatically rather than
treated symmetrically with strong ones. The predictive-likelihood-based weighting combined with an importance-score threshold is sufficient to guarantee acyclicity (Lemma \ref{lem1}). 

    \item \emph{Theoretical guarantees.} We establish a risk bound (Theorem \ref{thmbound}) for
    the aggregated estimator. We also show  weak and strong consistency for edge selection (Theorems \ref{(a)-Assume-that} and \ref{thm2}), even when candidate
learners are heterogeneous and span multiple paradigms.

    \item \emph{Practical implementation.} The procedure is
    computationally tractable, agnostic to the underlying DAG
    learners, and integrates seamlessly with score-based,
    constraint-based, hybrid, and continuous-optimization methods,
    with explicit guidelines for choosing the number of candidates
    $K$, the number of splits $L$, and the threshold $c$.

    \item \emph{Empirical validation.} Across simulations with
    random, hub, and chain structures, and on the benchmark
    \citet{sachs2005causal} protein-signaling dataset, our method
    matches or exceeds the best individual estimators across nine to
    ten candidate algorithms, and outperforms standard ensemble
    baselines on FNR, MCC, and structural Hamming distance.
\end{enumerate}

The remainder of the paper is organized as follows.
Section~\ref{sec:setup} introduces the model setup and notation.
Section~\ref{sec:method} presents the proposed aggregation procedure
and practical guidelines for its tuning parameters.
Section~\ref{sec:theory} develops the theoretical properties of the
aggregated estimator, including a risk bound and acyclicity
preservation. Section~\ref{sec:edges} introduces the edge importance
score and establishes consistency results for edge selection.
Sections~\ref{sec:simulation} and~\ref{sec:realdata} report
simulation studies and an application to the
\citet{sachs2005causal} dataset. Section~\ref{sec:conclusion}
concludes with a discussion of extensions and open directions.

\section{Model Setup}\label{sec:setup}
Let $\bX=(X_{1},\ldots,X_{p})^{\top}$ be a $p$-dimensional random vector. We consider the structural equation model
\begin{equation}
    X_{j}=\sum_{k\neq j}U_{kj}X_{k}+\epsilon_{j},\quad j\in\{1,\ldots,p\},\label{eq:model}
\end{equation}
where the noise terms satisfy $\epsilon_{j}\overset{i.i.d.}{\sim}N(0,\sigma^{2})$ and $\mathbf{U}=\{U_{kj}\}_{p\times p}$ is the weighted adjacency matrix. As shown in \citet{yuan2019constrained}, model \eqref{eq:model} defines a DAG when $\mathbf{U}$ satisfies the acyclicity condition
\begin{equation}
\sum_{t=1}^{L}\I(U_{j_{t}j_{t+1}}\neq0)\le L-1 \label{eq:dagcondition}
\end{equation}
for any $j_{1},\ldots,j_{L}\in[p]$ and $L\in\{2,\ldots,p\}$, where $\mathbb{I}(\cdot)$ is the indicator function.

The DAG model in \eqref{eq:model} is identifiable on the parameter space
\begin{equation}
\Theta=\{(\mathbf{U},\sigma):\sigma>0,\ \mathbf{U}\text{ satisfies }\eqref{eq:dagcondition}\}.
\end{equation}
By the local Markov property \citep{edwards2012introduction}, the joint distribution of $(X_{1},\ldots,X_{p})$ factorizes as a product of conditional distributions of each $X_{j}$ given its parents. Given an i.i.d.\ sample $\{\mathbf{x}_{i}\in\mathbb{R}^{p}\}_{i=1}^{n}$, the log-likelihood is proportional to
\begin{equation}\label{eq:likelihood}
-\sum_{i=1}^{n}\sum_{j=1}^{p}\left[\frac{1}{2\sigma^{2}}\left(x_{ij}-\sum_{k\neq j}U_{jk}x_{ik}\right)^{2}+\frac{1}{2}\log\sigma^{2}\right],
\end{equation}
where $\mathbf{x}_{i}=(x_{i1},\ldots,x_{ip})^{\top}$. For convenience, we denote the per-observation log-likelihood by
\[
l_{\mathbf{U},\sigma^{2}}(\mathbf{x}_{i}):=-\sum_{j=1}^{p}\left[\frac{1}{2\sigma^{2}}\left(x_{ij}-\sum_{k\neq j}U_{jk}x_{ik}\right)^{2}+\frac{1}{2}\log\sigma^{2}\right].
\]
\paragraph{Notation.}
Throughout the paper, bold uppercase upright letters (e.g., $\mathbf{A}$, $\mathbf{X}$) denote matrices, bold lowercase upright letters (e.g., $\mathbf{x}$, $\mathbf{v}$) denote vectors, calligraphic letters (e.g., $\mathcal{V}$, $\mathcal{E}$) denote sets, and bold italic uppercase letters (e.g., $\boldsymbol{X}$, $\boldsymbol{Y}$) denote random vectors or random matrices.
We represent a DAG as a pair $(\mathcal{V},\mathcal{E})$ with vertex set $\mathcal{V}=\{1,\ldots,p\}$ and edge set $\mathcal{E}\subseteq\{(i,j):1\le i,j\le p\}$, where $(i,j)\in \mathcal{E}$ encodes the directed edge $i\to j$. For any integer $p\ge1$, write $[p]:=\{1,\ldots,p\}$. We use $\nabla$ for the symmetric difference of two sets and $|\cdot|$ for cardinality.

\section{Directed Acyclic Graph Aggregation (DAGgr)}\label{sec:method}

Given data $\mathbf{X}:=\{\mathbf{x}_{1},\ldots,\mathbf{x}_{n}\}\subset\mathbb{R}^{p}$ and a collection of candidate DAGs $\{(\mathcal{V}_{k},\mathcal{E}_{k})\}_{k=1}^{K}$ with common vertex set $\mathcal{V}_{k}\equiv[p]$, obtained either from different procedures or from a single procedure run with different tuning parameters, our goal is to construct a weighted aggregate that improves stability, interpretability, and predictive performance over any single candidate. The central technical difficulty is that naive aggregation can violate acyclicity. If a candidate method produces only an edge set $(\mathcal{V},\mathcal{E})$ without an estimated weighted adjacency matrix $\widehat{\mathbf{U}}$, we obtain one by maximizing \eqref{eq:likelihood} subject to the structural constraint $\I(\widehat{U}_{ij}\neq0)=\I((i,j)\in \mathcal{E})$. Given the resulting estimates $\{\widehat{\mathbf{U}}^{(k)}\}_{k=1}^{K}$, we form the aggregated adjacency matrix as the convex combination
\[
\widetilde{\mathbf{U}}=\sum_{k=1}^{K}w_{k}\widehat{\mathbf{U}}^{(k)},
\]
where the weights satisfy $w_{k}\ge0$ and $\sum_{k=1}^{K}w_{k}=1$. The aggregated DAG $(\,[p],\widetilde{\mathcal{E}}\,)$ is then read off from $\widetilde{\mathbf{U}}$ via $(i,j)\in\widetilde{\mathcal{E}}\iff\widetilde{U}_{ij}\neq0$.


The remaining challenge is to choose the weights so that acyclicity is preserved and the strengths of individual candidates are exploited. Uniform weights or weights based on a single training-set criterion typically fail to deliver these properties, especially when candidates differ in quality across regimes or when only certain edges are stably identified.

We propose a data-driven weighting scheme based on repeated data splitting, summarized in Algorithm~\ref{alg:ARM-weighting-method-1}. At each iteration we fit each candidate DAG on a training subsample and evaluate its predictive performance on the held-out subsample; the resulting validation likelihoods determine the weights, which are then averaged across splits.

\begin{algorithm}[!h]
{\bf Input:} Candidate DAGs $\{(\mathcal{V}_{k},\mathcal{E}_{k})\}_{k=1}^{K}$, with optional estimated adjacency matrices $\{\widehat{\mathbf{U}}^{(k)}\}_{k=1}^{K}$ and noise standard deviations $\{\widehat{\sigma}_{k}\}_{k=1}^{K}$.
\begin{itemize}
\item {\bf Step 1: Data splitting.} Randomly split the data into equal halves $\mathcal{S}_{1}$ and $\mathcal{S}_{2}$. After reindexing, write $\mathcal{S}_{1}=\{\mathbf{x}_{i}\}_{i=1}^{n/2}$ and $\mathcal{S}_{2}=\{\mathbf{x}_{i}\}_{i=n/2+1}^{n}$.
\item {\bf Step 2: Model fitting.} For each $k\in\{1,\ldots,K\}$, fit the $k$-th candidate DAG on $\mathcal{S}_{1}$ and obtain $\widehat{\mathbf{U}}^{(k)}$ and $\widehat{\sigma}^{(k)}$.
\item {\bf Step 3: Likelihood-based weights.} Compute
\[
w^{k}=\frac{\pi_k\exp(\lambda\sum_{i\in\mathcal{S}_{2}}l_{\widehat{\mathbf{U}}^{(k)},\widehat{\sigma}^{(k)}}(\mathbf{x}_{i}))}{\sum_{k'=1}^{K}\pi_k\exp(\lambda\sum_{i\in\mathcal{S}_{2}}l_{\widehat{\mathbf{U}}^{(k')},\widehat{\sigma}^{(k')}}(\mathbf{x}_{i}))}.
\]
\item {\bf Step 4: Averaging across splits.} Repeat Steps~1--3 a total of $L$ times and set the final weights to $w_{k}=\frac{1}{L}\sum_{l=1}^{L}w_{l}^{k}$.
\end{itemize}
\noindent\textbf{Output:} The weight vector $\mathbf{w}=(w_{1},\ldots,w_{K})$ and the aggregated adjacency matrix $\widetilde{\mathbf{U}}=\sum_{k=1}^{K}w_{k}\widehat{\mathbf{U}}^{(k)}$.
\caption{Directed Acyclic Graph Aggregation (DAGgr)\label{alg:ARM-weighting-method-1}}
\end{algorithm}

The parameter $\lambda>0$ is a learning rate that controls how sharply the weights concentrate on the best-fitting candidates: larger $\lambda$ produces a more decisive weight vector, while $\lambda\to0$ recovers the prior $\pi_{k}$. The prior $\{\pi_{k}\}_{k=1}^{K}$ encodes any domain knowledge about candidate quality and defaults to the uniform prior $\pi_{k}=1/K$. Theoretical guidance for choosing $\lambda$ is given in Theorem~\ref{thmbound} and its proof in the appendix; in our simulations and real-data analysis, we use the default $\lambda=1$ with uniform priors.
\subsection{Practical Guidelines for Choosing \texorpdfstring{$K$, $L$, $\lambda$, and $c$}{K, L, lambda, and c}}
The proposed framework involves four tuning parameters: the number of candidate DAGs $K$ and the number of data splits $L$, both introduced in Algorithm~\ref{alg:ARM-weighting-method-1}; the learning rate $\lambda$, which controls weight concentration in Step~3 of the same algorithm; and a threshold $c\in(0,1)$ applied to each edge's importance score---a weighted measure of how consistently the edge appears across candidates. Edges with importance score above $c$ are retained, so $c$ controls the sparsity of the final graph; we formally introduce the importance score and threshold in Section~\ref{sec:edges}. Together, $K$, $L$, $\lambda$, and $c$ control the trade-off between stability, sparsity, predictive performance, and computational cost.

\paragraph{Choosing $K$ (number of candidate DAGs).}
A practical range is $K\in[5,15]$, drawing from a mix of paradigms (score-based, constraint-based, hybrid, continuous-optimization) and tuning configurations. Excessive similarity among candidates dilutes the benefit of averaging, while including many weak candidates raises the computational burden and can destabilize the weight estimates.

\paragraph{Choosing $L$ (number of data splits).}
The split count $L$ controls the variance of the weight estimates. A default of $L=20$ is often sufficient; for small or noisy datasets, $L\in[30,50]$ is preferable. A simple stopping rule is to continue increasing $L$ until consecutive weight vectors differ negligibly.

\paragraph{Choosing $\lambda$ (learning rate).}
The learning rate $\lambda>0$ controls how strongly the weights favor high-likelihood candidates. The default $\lambda=1$ corresponds to standard Bayesian-style model averaging and works well in our experiments. Larger $\lambda$ produces sparser, more decisive weight vectors but can over-concentrate on a single candidate; smaller $\lambda$ smooths the weights toward the prior $\pi_{k}$. In settings where any individual candidate may be poorly calibrated, $\lambda<1$ (e.g., $\lambda\in[0.1,1]$) is a useful safeguard. Theoretical guidance is given by the conditions of Theorem~\ref{thmbound}, which require $\lambda$ small relative to $1/(p A^{2})$ for the oracle inequality to hold.

\paragraph{Choosing $c$ (edge-importance threshold).}
The threshold $c$ governs sparsity and the tolerance for false discoveries. A data-driven approach is to scan a grid (e.g., $c\in\{0.5,0.55,\ldots,0.95\}$) and select the value that optimizes validation performance subject to a target sparsity level. Larger $c$ yields sparser graphs with higher precision; smaller $c$ favors recall. We discuss the role of $c$ in greater depth in Section~\ref{sec:edges}.

\paragraph{Recommended defaults.}
As a starting point we recommend $K=10$, $L=30$, $\lambda=1$ with uniform prior $\pi_{k}=1/K$, and selecting $c$ from a validation grid, preferring the smallest threshold consistent with the desired false-discovery behavior. For transparency, reporting sensitivity across several values of $c$ (and, if relevant, $\lambda$) is also advisable.

\section{Theoretical Properties}\label{sec:theory}
This section establishes the theoretical foundations of DAGgr. We first prove a finite-sample risk bound for the aggregated estimator, then show that a simple thresholding rule guarantees that the aggregated graph is acyclic.
\begin{thm}\label{thmbound}
Let $\widetilde{\mathbf{U}}=\sum_{k=1}^{K}w_{k}\widehat{\mathbf{U}}^{(k)}$ denote the DAGgr estimator constructed via Algorithm~\ref{alg:ARM-weighting-method-1}, and let $\{\pi_{k}\}_{k=1}^{K}$ be any prior weights with $\pi_{k}>0$ and $\sum_{k}\pi_{k}=1$. Assume:
\begin{itemize}
    \item[(C1)] (\emph{Bounded coefficients.}) There exists a constant $A<\infty$ such that $\max_{i,j}|U_{ij}|\le A$ and $\max_{k,i,j}|\widehat{U}_{ij}^{(k)}|\le A$.
    \item[(C2)] (\emph{Tuning.}) The parameter $\lambda>0$ is chosen sufficiently small that the sub-exponential moment conditions in the appendix hold; namely,
\[
6\lambda p A^{2}\le d_{2}/d_{1}
\qquad\text{and}\qquad
36\lambda^{2}pA^{2}\!\left[\exp(72e^{2}d_{1}^{2}\lambda^{2}pA^{2})+\tfrac{d_{4}}{2}\exp(36 d_{3}\lambda^{2}pA^{2})\right]\le\lambda,
\]
where $d_{1}=\sup_{k\ge 1}k^{-1}( \E|X_{i}|^{2k})^{1/k}$ is the sub-exponential constant of $|X_{i}|^{2}$ for $X_i$'s in model \eqref{eq:model} and $d_{2}=1/(4e)$, $d_{3}=8e^{4}d_{1}^{2}$, $d_{4}=16\sqrt{2}\,d_{1}^{2}$ are the associated universal constants.
\end{itemize}
Then
\[
\E\|\bX-\widetilde{\mathbf{U}}\bX\|^{2}\le\inf_{k}\left\{\E\|\bX-\widehat{\mathbf{U}}^{(k)}\bX\|^{2}+\frac{2\log(1/\pi_{k})}{\lambda (n-n_{\textrm{train}}) }\right\}.
\]
\end{thm}

Theorem~\ref{thmbound} shows that the predictive risk of the aggregated DAG is no worse than that of the single best candidate, up to an additive term that vanishes as $n\to\infty$. This adaptivity property is what makes model averaging (in this case, DAGgr) a robust alternative to estimator selection when the practitioner does not know in advance which method is best suited to the data at hand.
\begin{remark}[Oracle interpretation]
The right-hand side of the bound is the so-called \emph{oracle risk}: the predictive risk one would attain if the identity of the best candidate were known a priori, plus a complexity-style penalty $2\log(1/\pi_k)/[\lambda(n-n_{\textrm{train}})]$. Because the bound holds for every $k$ simultaneously, taking the infimum shows that DAGgr automatically matches the best candidate up to a vanishing penalty---without ever needing to identify which candidate is best. With a uniform prior $\pi_k=1/K$, the penalty is $2\log K/[\lambda(n-n_{\textrm{train}})]$, so the cost of model averaging scales only logarithmically in the size of the candidate set.
\end{remark}

\begin{remark}[On the rate $\lambda$]
Condition (C3) requires $\lambda$ to be sufficiently small to ensure the sub-exponential moment generating functions remain finite. The constraints scale as $\lambda \lesssim 1/(pA^{2})$, so larger candidate matrices or larger graphs require more conservative weighting. The price of a smaller $\lambda$, however, is a slower vanishing rate for the additive penalty. The interplay $\lambda(n-n_{\textrm{train}})$ in the denominator clarifies the role of the data split: more data devoted to validation ($n-n_{\textrm{train}}$ large) sharpens the oracle bound but reduces the training sample $n_{\textrm{train}}$ used to estimate each $\widehat{\mathbf{U}}^{(k)}$; in our experiments the standard equal split $n_{\textrm{train}}=n/2$ proved a robust default.
\end{remark}

\subsection{Acyclicity-Preserving Aggregation}
While the procedure above is intuitive, a key concern is that averaging a collection of DAGs can in principle introduce cycles: each $\widehat{\mathbf{U}}^{(k)}$ is acyclic, but their weighted combination need not be. We now establish a thresholded variant of the aggregation rule that provably preserves acyclicity.

The intuition behind the construction is that any cycle in $\widetilde{\mathbf{U}}$ must be assembled from edges that fail to form cycles in any single candidate. By requiring each surviving edge to be supported by a sufficiently large weighted fraction of candidates, we ensure that no such assembled cycle can arise. Cycles tend to be unstable artifacts that do not appear consistently across diverse estimators or random splits, so this rule is also intuitively appealing.
\begin{lem}\label{lem1}
Let $\{\mathbf{U}^{(1)},\ldots,\mathbf{U}^{(K)}\}$ be a collection of DAGs with nonnegative weights $\{w_{k}\}_{k=1}^{K}$ satisfying $\sum_{k}w_{k}=1$. For a constant $c\in[1-1/p,1]$, define $\mathbf{U}$ entrywise by
\[
\mathbf{U}_{ij}=\begin{cases}
\sum_{k=1}^{K}w_{k}\mathbf{U}_{ij}^{(k)}, & \textrm{if }\sum_{k=1}^{K}w_{k}\I\!\left(\mathbf{U}_{ij}^{(k)}\neq0\right)>c,\\[4pt]
0, & \textrm{if }\sum_{k=1}^{K}w_{k}\I\!\left(\mathbf{U}_{ij}^{(k)}\neq0\right)\le c.
\end{cases}
\]
Then $\mathbf{U}$ is the adjacency matrix of a DAG.
\end{lem}
\begin{proof}
We show that $\sum_{t=1}^{L}\I(\mathbf{U}_{j_{t}j_{t+1}}\neq0)\le L-1$ for every cycle $(j_{1},\ldots,j_{L},j_{L+1}=j_{1})$ with $L\in\{2,\ldots,p\}$. Suppose, for contradiction, that there exist $L_{0}$ and indices $(j_{1},\ldots,j_{L_{0}},j_{L_{0}+1}=j_{1})$ such that $\sum_{t=1}^{L_{0}}\I(\mathbf{U}_{j_{t}j_{t+1}}\neq0)=L_{0}$, i.e., every edge in the cycle is retained. Then for each $t=1,\ldots,L_{0}$,
\[
\sum_{k=1}^{K}w_{k}\I\!\left(\mathbf{U}_{j_{t}j_{t+1}}^{(k)}\neq0\right)>c.
\]
Summing over $t$ and exchanging the order of summation,
\[
L_{0}c<\sum_{t=1}^{L_{0}}\sum_{k=1}^{K}w_{k}\I\!\left(\mathbf{U}_{j_{t}j_{t+1}}^{(k)}\neq0\right)=\sum_{k=1}^{K}w_{k}\sum_{t=1}^{L_{0}}\I\!\left(\mathbf{U}_{j_{t}j_{t+1}}^{(k)}\neq0\right)\le\sum_{k=1}^{K}w_{k}(L_{0}-1)=L_{0}-1,
\]
where the second inequality uses the acyclicity of each $\mathbf{U}^{(k)}$. Rearranging yields $c<1-1/L_{0}$, which contradicts $c\ge1-1/p\ge 1-1/L_{0}$. Hence no such cycle exists, and $\mathbf{U}$ is acyclic.
\end{proof}
\begin{remark}
The proof formalizes the intuition stated above: any putative cycle in $\widetilde{\mathbf{U}}$ would require every edge along it to be supported, in weighted fraction exceeding $c$, by candidate DAGs. The constraint $c\ge 1-1/p$ ensures that this requirement is incompatible with the acyclicity of each individual candidate.
\end{remark}
\begin{remark}[Geometric interpretation and tightness]
The condition $c\ge 1-1/p$ has a natural interpretation: in any cycle of length $L\le p$, the average ``vote'' across the $L$ edges cannot exceed $(L-1)/L=1-1/L$ since no individual candidate contains a cycle. The threshold $1-1/p$ is the worst-case bound, valid for cycles of maximum possible length $p$; for graphs known to have no long cycles, a looser threshold may be used. In our edge-importance framework (Section~\ref{sec:edges}), the threshold $c$ plays an analogous role and the same argument shows $c\ge 1-1/p$ suffices for acyclicity. Crucially, the result is \emph{purely combinatorial}: it requires no distributional assumptions on $\mathbf{x}_i$ and no statistical properties of the weights $w_k$, only that they form a valid convex combination.
\end{remark}

\section{Edge Importance and Selection}\label{sec:edges}
The aggregated matrix $\widetilde{\mathbf{U}}$ blends edges from multiple candidate DAGs, but these edges differ in how strongly they are supported by the ensemble. Some are consistently identified by many high-performing candidates; others appear only in a few, possibly weaker, estimates. To distinguish reliable edges from noisy ones, we introduce a stability-based importance measure that drives a principled edge-selection rule.
\begin{defn}[Edge importance]
The importance score of edge $(i,j)$ is
\[
s_{ij}:=\sum_{k=1}^{K}w_{k}\,\mathbb{I}\!\left(\mathbf{U}_{ij}^{(k)}\neq0\right),\quad i,j\in[p].
\]
\end{defn}

The score $s_{ij}\in[0,1]$ measures the weighted frequency with which edge $(i,j)$ appears in the candidate set: it equals one if and only if every candidate (weighted by $w_{k}$) contains the edge, and zero if and only if no candidate does. Edges supported by high-performing candidates carry large weight in this sum, so consistently chosen edges receive scores close to one, while spurious ones receive scores close to zero.

This naturally yields a thresholding rule: fix $c\in(0,1)$ and retain all edges $(i,j)$ with $s_{ij}>c$. We show below that, under mild consistency conditions on the weights, this rule recovers the true edge set with high probability as $n\to\infty$.

Let $\mathcal{E}^{*}:=\{(i,j):\mathbf{U}_{ij}\neq0\}$ denote the edge set of the true underlying DAG, and let $\mathcal{E}^{(k)}:=\{(i,j):\mathbf{U}_{ij}^{(k)}\neq0\}$ denote the edges selected by the $k$-th candidate.

\begin{defn}
The weight vector $\mathbf{w}$ is \emph{weakly consistent} if
\[
\frac{\sum_{k=1}^{K}w_{k}\,|\mathcal{E}^{(k)}\nabla\mathcal{E}^{*}|}{|\mathcal{E}^{*}|}\overset{p}{\longrightarrow}0\quad\text{as }n\to\infty,
\]
and \emph{strongly consistent} (or simply \emph{consistent}) if
\[
\sum_{k=1}^{K}w_{k}\,|\mathcal{E}^{(k)}\nabla\mathcal{E}^{*}|\overset{p}{\longrightarrow}0\quad\text{as }n\to\infty.
\]
\end{defn}

Weak consistency requires that the weighted edge-discrepancy, normalized by the size of the true edge set, vanishes in probability; strong consistency requires the unnormalized discrepancy to vanish. Under these conditions, true edges concentrate at importance scores near one and spurious edges at scores near zero, so the thresholding rule $s_{ij}>c$ recovers $\mathcal{E}^{*}$ asymptotically.
\begin{thm}\label{(a)-Assume-that}
(a) If $\mathbf{w}$ is weakly consistent, then
\[
\frac{\sum_{(i,j)\in\mathcal{E}^{*}}s_{ij}}{|\mathcal{E}^{*}|}\overset{p}{\longrightarrow}1\quad\text{and}\quad\frac{\sum_{(i,j)\notin\mathcal{E}^{*}}s_{ij}}{|\mathcal{E}^{*}|}\overset{p}{\longrightarrow}0\quad\text{as }n\to\infty.
\]

(b) If $\mathbf{w}$ is consistent, then
\[
\min_{(i,j)\in\mathcal{E}^{*}}s_{ij}\overset{p}{\longrightarrow}1\quad\text{and}\quad\max_{(i,j)\notin\mathcal{E}^{*}}s_{ij}\overset{p}{\longrightarrow}0\quad\text{as }n\to\infty.
\]
\end{thm}
\begin{remark}[Interpretation]
Theorem~\ref{(a)-Assume-that} bridges the abstract notion of weight consistency with the operational quantity used for edge selection: importance scores. Under weak consistency---the weaker of the two assumptions---true edges receive total importance close to their cardinality and false edges receive total importance vanishing in relative terms. Strong consistency strengthens this to a uniform separation: every true edge has score close to one and every false edge has score close to zero, ensuring that \emph{any} threshold $c\in(0,1)$ recovers the true graph with high probability. The dichotomy mirrors the classical distinction between averaged-risk consistency and uniform consistency in estimation theory.
\end{remark}

\begin{remark}[When does weight consistency hold?]
Strong consistency of $\mathbf{w}$ holds whenever at least one candidate is structurally consistent---i.e., $|\mathcal{E}^{(k^{*})}\nabla\mathcal{E}^{*}|\overset{p}{\to}0$ for some $k^{*}$---and the validation-likelihood weights asymptotically concentrate on such a candidate. Many established DAG learners are structurally consistent under appropriate conditions \citep{chickering2002optimal,yuan2019constrained,zheng2018dags}, so the assumption is mild whenever the candidate set is sufficiently diverse. Notably, DAGgr does \emph{not} require the practitioner to know which candidate is consistent---only that one exists.
\end{remark}
Define $\mathcal{E}_{c}:=\{(i,j):s_{ij}>c\}$ as the set of edges retained at threshold $c\in(0,1)$.
\begin{thm}\label{thm2}
For any $c\in(0,1)$:
\begin{itemize}
    \item[(a)] If $\mathbf{w}$ is weakly consistent, then
\[
\frac{|\{(i,j)\in\mathcal{E}^{*}:s_{ij}\le c\}|}{|\mathcal{E}^{*}|}\overset{p}{\longrightarrow}0\quad\text{and}\quad\frac{|\{(i,j)\notin\mathcal{E}^{*}:s_{ij}>c\}|}{|\mathcal{E}^{*}|}\overset{p}{\longrightarrow}0
\]
as $n\to\infty$.
    \item[(b)] If $\mathbf{w}$ is consistent, then
\[
|\{(i,j)\in\mathcal{E}^{*}:s_{ij}\le c\}|\overset{p}{\longrightarrow}0\quad\text{and}\quad|\{(i,j)\notin\mathcal{E}^{*}:s_{ij}>c\}|\overset{p}{\longrightarrow}0
\]
as $n\to\infty$.
\end{itemize}
\end{thm}

\begin{remark}[False discoveries and missed edges]
The two conclusions correspond, respectively, to controlling false negatives ($\{(i,j)\in\mathcal{E}^{*}:s_{ij}\le c\}$, true edges erroneously discarded) and false positives ($\{(i,j)\notin\mathcal{E}^{*}:s_{ij}>c\}$, spurious edges erroneously retained). Under weak consistency, both error sets shrink at rate $|\mathcal{E}^{*}|^{-1}$; under strong consistency they shrink in absolute terms. The conclusion holds for every $c\in(0,1)$, so the choice of threshold does not affect asymptotic correctness (only the finite-sample trade-off between precision and recall), as discussed below.
\end{remark}

The symmetric difference decomposes as $\mathcal{E}_{c}\nabla\mathcal{E}^{*}=\{(i,j)\in\mathcal{E}^{*}:s_{ij}\le c\}\cup\{(i,j)\notin\mathcal{E}^{*}:s_{ij}>c\}$, so the threshold $c$ directly trades off false negatives against false positives. A larger $c$ (e.g., $c=0.9$) yields a sparser graph with high precision; a smaller $c$ (e.g., $c=0.5$) prioritizes recall at the cost of admitting more false positives. In practice, $c$ can be tuned by domain knowledge, model-complexity considerations, or validation-based selection.

To translate the importance scores into an acyclic graph, we combine thresholding with a deterministic pruning step. Starting from the averaged matrix $\widetilde{\mathbf{U}}$, we first zero out all entries with $s_{ij}\le c$. The remaining graph may still contain cycles, so we iteratively remove edges in increasing order of importance: at each iteration, we identify the surviving edges with the smallest importance score and delete them in increasing order of $|\widetilde{U}_{ij}|$ until either acyclicity is restored or no such edges remain. The full procedure is given in Algorithm~\ref{algo:prune}.
\begin{algorithm}[H]
\caption{DAG pruning procedure.}
\label{algo:prune}
\begin{algorithmic}[1]
\Require Averaged adjacency matrix $\widetilde{\mathbf{U}}$, importance scores $\mathcal{S}=(s_{ij})$, threshold $c$.
\State Initialize $\widehat{\mathbf{U}}\gets\widetilde{\mathbf{U}}$.
\State Set $\widehat{U}_{ij}\gets0$ for all $(i,j)$ with $s_{ij}\le c$.
\While{$\widehat{\mathbf{U}}$ is not a DAG}
    \State $\mathcal{E}\gets\{(i,j):\widehat{U}_{ij}\neq0\}$.
    \State $\mathcal{E}_{\min}\gets\{(i,j)\in\mathcal{E}:s_{ij}=\min_{(k,\ell)\in\mathcal{E}}s_{k\ell}\}$.
    \State Sort $\mathcal{E}_{\min}$ in increasing order of $|\widetilde{U}_{ij}|$.
    \For{each $(i,j)\in\mathcal{E}_{\min}$ in this order}
        \State $\widehat{U}_{ij}\gets0$.
        \If{$\widehat{\mathbf{U}}$ is a DAG}
            \State \textbf{break}.
        \EndIf
    \EndFor
\EndWhile
\State \Return the pruned adjacency matrix $\widehat{\mathbf{U}}$.
\end{algorithmic}
\end{algorithm}

The special case $c=0$ corresponds to applying Algorithm~\ref{algo:prune} without any prior thresholding; we refer to this as the \emph{step-size pruned} procedure.



\subsection{Directed and Undirected Edges}\label{sec: undirect}
Candidate DAGs frequently disagree on edge orientations. Inspired by the use of partially directed graphs (PDAGs) in bootstrap-based Bayesian network induction \citep{friedman1999bootstrap}---where edges within an equivalence class are classified as directed when all members agree on orientation and as undirected otherwise---we extend this idea to the weighted-aggregation setting.
\begin{defn}
Fix a confidence threshold $\tau\in(0,1]$ and a symmetry tolerance $\delta\in[0,1)$. We call $(i,j)$ a \emph{directed edge} if $s_{ij}\geq\tau$ and $s_{ij}-s_{ji}>\delta$.
\end{defn}
\begin{defn}
Fix a confidence threshold $\tau\in(0,1]$ and a symmetry tolerance $\delta\in[0,1)$. We call $(i,j)$ an \emph{undirected edge} if $s_{ij}+s_{ji}\geq\tau$ and $|s_{ij}-s_{ji}|\leq\delta$.
\end{defn}
\begin{remark}
An undirected edge indicates that the connection between $i$ and $j$ is well supported overall but that its orientation is underdetermined by the candidate set.
\end{remark}
Whereas \citet{friedman1999bootstrap} treat each bootstrap structure uniformly, our weights $w_{k}$ allow candidate graphs to contribute unequally, reflecting posterior probability, goodness-of-fit, or any other measure of structural plausibility. The resulting graph is therefore a weighted analogue of a PDAG: directed edges encode orientations on which the weighted ensemble strongly agrees, while undirected edges flag connections whose direction the data cannot resolve.

\section{Simulations}\label{sec:simulation}
We conduct extensive simulation studies to evaluate DAGgr under a range of adjacency-matrix structures. Data are generated from model~\eqref{eq:model} with $p=100$ nodes and $n=250$ observations. Let $\bm{E}=(\bm{\epsilon}_{1},\ldots,\bm{\epsilon}_{n})^{\top}$, where each $\bm{\epsilon}_{i}=(\epsilon_{i1},\ldots,\epsilon_{ip})^{\top}$ is drawn independently from a $p$-dimensional Gaussian distribution with mean zero and identity covariance. The data matrix is then
\[
\bm{X}=\bm{E}(\mathbf{I}_{p}-\mathbf{U})^{-1}.
\]

We consider three adjacency-matrix structures: (i) a \emph{random} structure in which nonzero entries are independently selected according to a Bernoulli distribution; (ii) a \emph{hub} structure in which a single node is a parent of all others; and (iii) a \emph{chain} structure in which each node has at most one parent and at most one child. Conditional on being nonzero, each entry of $\bm{U}$ is drawn independently from $\text{Unif}([-1.5,-0.5]\cup[0.5,1.5])$, ensuring a sufficiently strong signal. Representative DAGs from these three structures are shown in Figure~\ref{fig:simulation}.

\begin{figure}[!h]
    \centering
    \subfloat[Random adjacency matrix]{ \label{fig:random}
		 \includegraphics[width=0.31\textwidth]{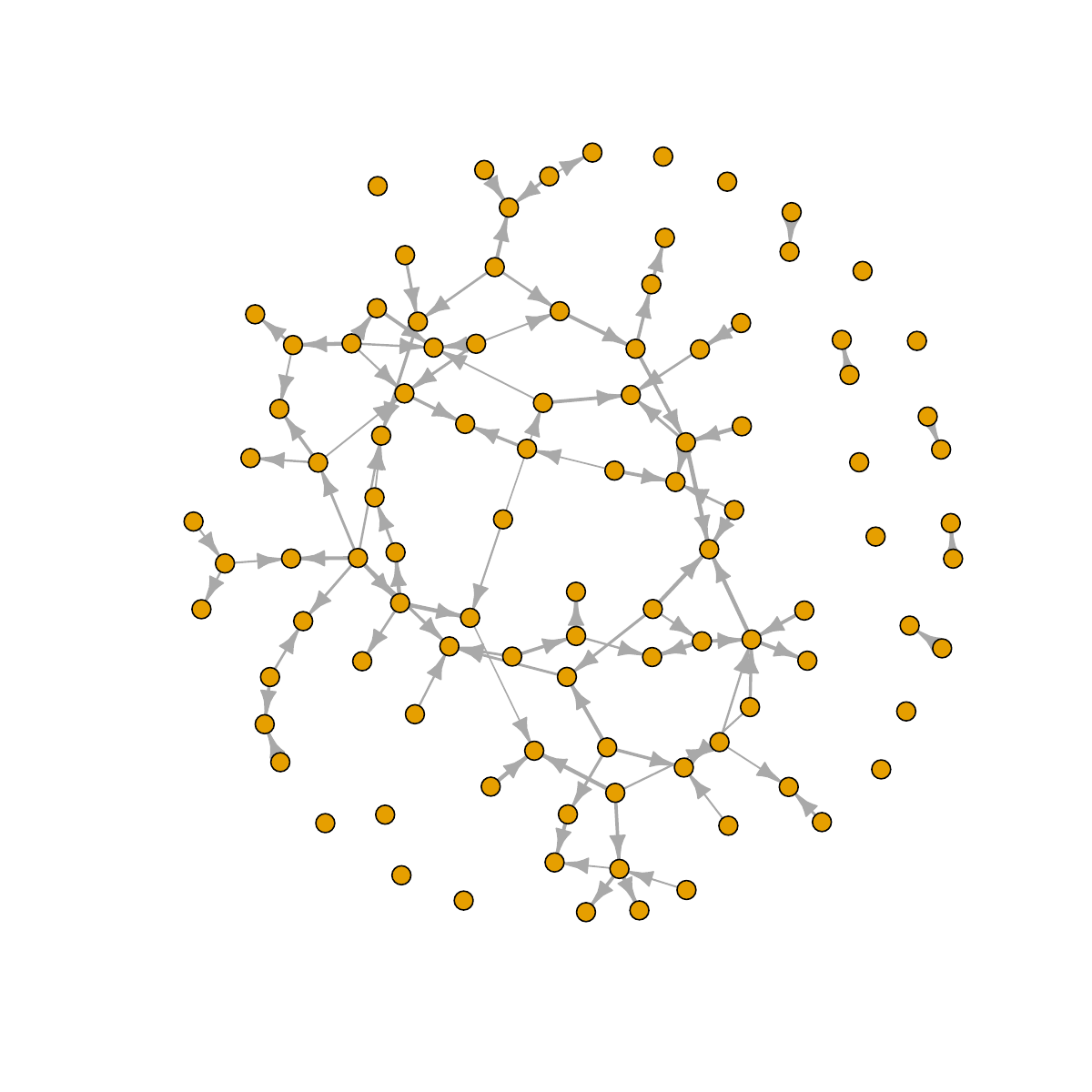}}\ \ 
	\subfloat[Hub adjacency matrix]{\label{fig:hub}
		\includegraphics[width=0.31\textwidth]{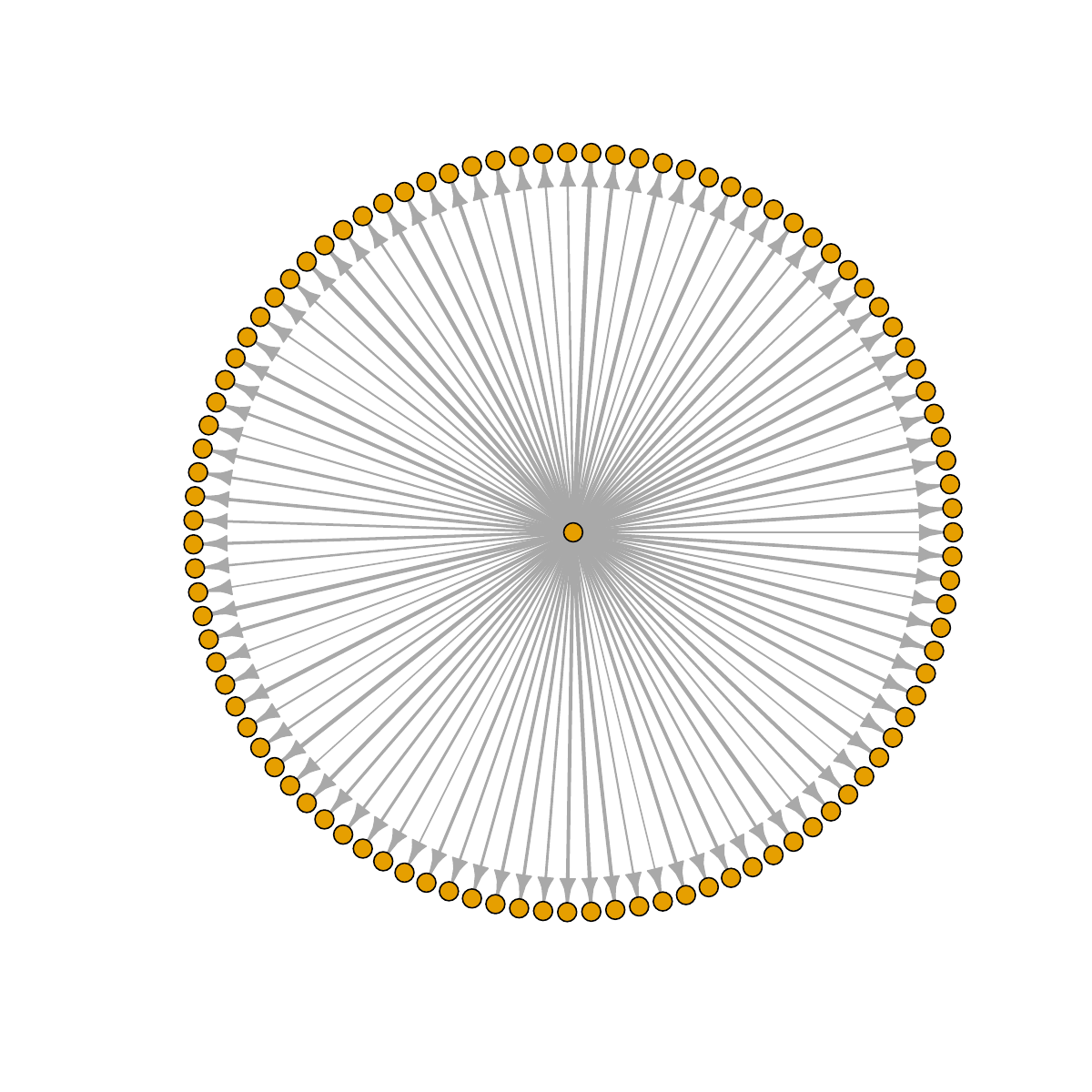}}\ \ 
	\subfloat[Chain adjacency matrix]{\label{fig:chain}
		\includegraphics[width=0.31\textwidth]{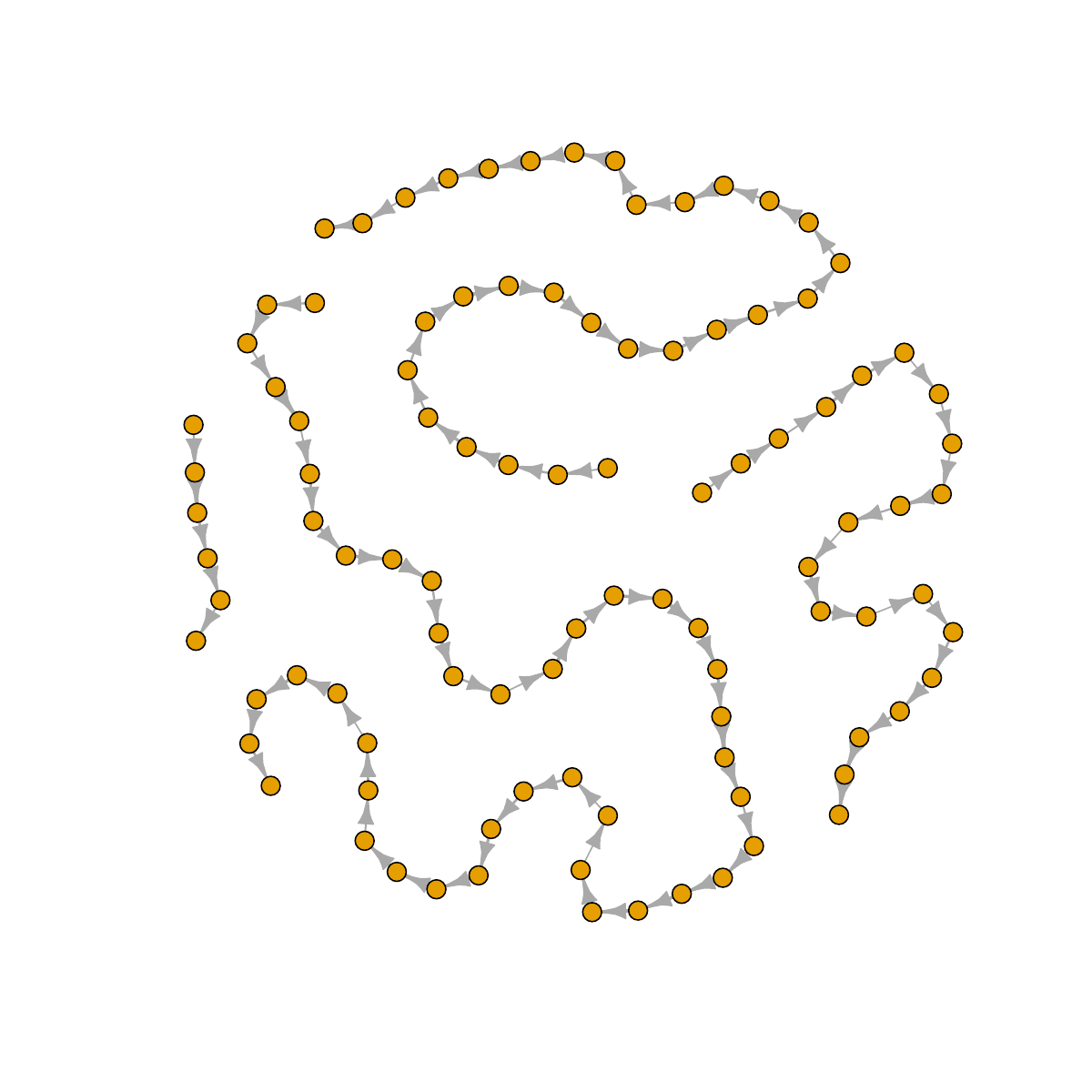}}
	\caption{Representative DAG structures used in the simulation study.}
	\label{fig:simulation}
\end{figure}

For the simulation study, we use nine candidate DAG estimators spanning the major paradigms reviewed in Section~\ref{sec:intro}: the score-based Greedy Equivalence Search (GES) \citep{chickering2002optimal} and Hill-Climbing algorithms; the constraint-based PC algorithm \citep{spirtes2000causation, kalisch2007estimating}; the hybrid Max-Min Hill-Climbing \citep{tsamardinos2006max}; two likelihood-based methods, constrained maximum likelihood estimation \citep{yuan2012maximum} and its variant \texttt{clrdag} \citep{li2020likelihood}; two continuous-optimization methods, NOTEARS \citep{zheng2018dags} (linear version) and DAGMA \citep{bello2022dagma}; and the recently proposed NOTIME \citep{pmlr-v258-berrevoets25a}. For each replication, we run Algorithm~\ref{alg:ARM-weighting-method-1} (the DAGgr weighting procedure) to obtain $\widetilde{\mathbf{U}}$, followed by Algorithm~\ref{algo:prune} (the DAGgr pruning procedure) at three thresholds $c=$  $0.5$, $0.8$, and $1-1/p$.

For benchmarking, we also compare against the bootstrap-aggregation method DAGBag \citep{wang2014learning}, implemented via the R package \texttt{dagbag} \citep{dagbagR}. Other aggregation-based methods are not considered here due to the lack of publicly available software implementations.

Performance is evaluated using four metrics, reported as means with standard errors (in parentheses) over $200$ replications: the average squared error $\|\widehat{\mathbf{U}}-\mathbf{U}\|_{F}^{2}$ in Frobenius norm; the false negative rate (FNR) in edge recovery; the Matthews correlation coefficient (MCC); and the structural Hamming distance (SHD). Because DAGBag outputs only a binary aggregated adjacency matrix, the average squared error is not reported for that method. We additionally report the average weights assigned to each candidate by DAGgr. Detailed results for all three adjacency structures appear in Tables~\ref{table:random}--\ref{table:chain}.

Tables~\ref{table:random}--\ref{table:chain} show that DAGgr performs consistently strongly across all three graph structures and metrics. It is competitive with, and often indistinguishable from, the best individual candidate, while exhibiting greater robustness across structural regimes than any single estimator.

The mechanism behind this behavior is transparent: candidates that capture spurious edges or that fit the data poorly receive lower validation likelihood and therefore smaller weights, limiting their influence on the aggregate. Conversely, well-performing candidates dominate the weighted combination, steering it toward a structure that better reflects the true causal relationships.

Together, these empirical findings reinforce the practical relevance of the theory developed in Sections~\ref{sec:theory} and~\ref{sec:edges}: by combining diverse estimators within a predictive-likelihood-weighted scheme and enforcing acyclicity, DAGgr yields stable and reliable causal inference in complex, uncertain settings.

\begin{table}[!h]
\centering
\caption{Performance comparison under the random adjacency matrix structure}
\resizebox{\textwidth}{!}{\begin{tabular}{cccccc}
\hline  &  Average weights &Average squared error & FNR& MCC & SHD\\\hline
DAGgr-raw &  - &\textbf{0.9955} (0.0427) & \textbf{0.35\%} (0.05\%) & 0.97 (0.00) & 0.8 (0.1)\\
DAGgr-pruned (0.5) &  - &1.0006 (0.0432) & \textbf{0.47\%} (0.05\%) & \textbf{0.98} (0.00) & \textbf{0.7} (0.1)\\
DAGgr-pruned (0.8) &  - &1.0003 (0.0432) & \textbf{0.47\%} (0.05\%) & \textbf{0.98} (0.00) & \textbf{0.7} (0.1)\\
DAGgr-pruned (1-1/p) &  - &\textbf{0.9999} (0.0426) & 0.54\% (0.06\%) & \textbf{0.98} (0.00) & \textbf{0.7} (0.1)\\ 
DAGBag &  - &- & 0.62\% (0.05\%) & 0.96 (0.00) & 2.0 (0.1)\\
\hline\hline
GES &  0\% &23.5622 (0.3427) & 13.62\% (0.19\%) & 0.54 (0.00) & 129.0 (0.7)\\
PC-algorithm &  0\% &128.8640 (0.7159) & 54.82\% (0.41\%) & 0.40 (0.00) & 69.0 (0.5)\\
Hill climbing &  0\% &\textbf{6.6274} (0.1991) & \textbf{2.11\%} (0.10\%) & 0.61 (0.00) & 117.0 (0.7)\\
Max-min hill climbing &  0\% &10.5572 (0.2049) & 9.97\% (0.17\%) & 0.82 (0.00) & 23.7 (0.3)\\
Constrained-MLE &  0\% &32.5609 (0.4812) & 14.24\% (0.19\%) & 0.32 (0.00) & 505.3 (1.5)\\
clrdag &  0\% &13.4153 (0.1272) & 13.41\% (0.15\%) & \textbf{0.85} (0.00) & \textbf{21.3} (0.2)\\
NOTEARS-linear &  2.17\% &\textbf{2.0380} (0.0325) & \textbf{1.14\%} (0.07\%) & \textbf{0.98} (0.00) & \textbf{1.1} (0.1)\\
DAGMA-linear &  97.83\% &\textbf{1.0005} (0.0444) & \textbf{0.47\%} (0.05\%) & \textbf{0.98} (0.00) & \textbf{0.7} (0.1)\\
NOTIME &  0\% &83.1511 (0.5445) & 56.29\% (0.35\%) & 0.35 (0.00) & 106.2 (0.6)\\\hline
\end{tabular} }\label{table:random}
\end{table}

\begin{table}[!h]
	\centering
	\caption{Performance comparison under the hub adjacency matrix structure}
\resizebox{\textwidth}{!}{\begin{tabular}{cccccc}
\hline  
&  Average weights &Average squared error & FNR & MCC & SHD\\
\hline
DAGgr-raw &  - &\textbf{1.2289} (0.0652) & \textbf{0.04\%} (0.02\%) & 0.80 (0.00) & 29.1 (1.0)\\ 
DAGgr-pruned (0.5) &  - &\textbf{1.3667} (0.1220) & \textbf{0.81\%} (0.18\%) & 0.85 (0.00) & 2.4 (0.5)\\ 
DAGgr-pruned (0.8) &  - &1.8086 (0.1470) & 2.11\% (0.23\%) & 0.84 (0.00) & \textbf{2.3} (0.3)\\ 
DAGgr-pruned (1-1/p) &  - &1.8698 (0.1463) & 2.30\% (0.23\%) & \textbf{0.98} (0.00) & \textbf{2.3} (0.2)\\ 
DAGBag &  - &- & 23.49\% (1.71\%) & 0.69 (0.02) & 31.5 (1.7)\\
\hline\hline
GES &  0\% &44.9597 (2.7757) & 34.65\% (2.34\%) & 0.39 (0.01) & 156.8 (2.5)\\ 
PC-algorithm &  0\% &146.3620 (0.5386) & 98.41\% (0.10\%) & 0.00 (0.00) & 210.7 (0.6)\\ 
Hill climbing &  0\% &63.6739 (3.0546) & 51.45\% (2.63\%) & 0.30 (0.02) & 170.5 (2.9)\\ 
Max-min hill climbing &  0\% &111.3782 (0.2152) & 96.19\% (0.26\%) & 0.02 (0.00) & 188.1 (0.6)\\ 
Constrained-MLE &  17.18\% &\textbf{2.0259} (0.0218) & \textbf{0.02\%} (0.01\%) & 0.80 (0.00) & 34.8 (0.5)\\ 
clrdag &  36.88\% &\textbf{2.2002} (0.1899) & \textbf{2.28\%} (0.23\%) & \textbf{0.98} (0.00) & \textbf{2.3} (0.2)\\ 
NOTEARS-linear &  0\% &6.0077 (0.1025) & 4.46\% (0.15\%) & \textbf{0.95} (0.00) & \textbf{4.4} (0.1)\\ 
DAGMA-linear &  45.94\% &\textbf{1.0938} (0.0183) & \textbf{0.21\%} (0.03\%) & \textbf{0.99} (0.00) & \textbf{0.2} (0.0)\\ 
NOTIME &  0\% &109.8656 (0.0009) & 100.00\% (0.00\%) & 0.00 (0.00) & 99.0 (0.0)\\\hline 
	\end{tabular} }\label{table:hub}
\end{table}

\begin{table}[!h]
	\centering
	\caption{Performance comparison under the chain adjacency matrix structure}
\resizebox{\textwidth}{!}{\begin{tabular}{cccccc}
\hline  &  Average weights &Average squared error & FNR& MCC & SHD\\\hline 
DAGgr-raw &  - &0.1900 (0.0241) & \textbf{0.04\%} (0.01\%) & 0.77 (0.00) & 3.4 (0.8)\\
DAGgr-pruned (0.5)&  - &\textbf{0.1882} (0.0241) & \textbf{0.04\%} (0.01\%) & \textbf{0.98} (0.00) & \textbf{1.1} (0.1)\\
DAGgr-pruned (0.8)&  - &\textbf{0.1882} (0.0241) & \textbf{0.04\%} (0.01\%) & \textbf{0.98} (0.00) & \textbf{1.1} (0.1)\\
DAGgr-pruned (1-1/p) &  - &0.1984 (0.0266) & 0.05\% (0.02\%) & \textbf{0.98} (0.00) & 1.2 (0.1)\\
DAGBag &  - &- & 10.27\% (0.27\%) & 0.88 (0.00) & 10.3 (0.3)\\
\hline\hline
GES &  0\% &84.3782 (2.4629) & 47.69\% (1.29\%) & 0.38 (0.01) & 99.0 (1.2)\\
PC-algorithm &  0\% &61.8079 (2.3484) & 26.00\% (1.61\%) & 0.72 (0.02) & \textbf{26.3} (1.6)\\
Hill climbing &  0.7\% &\textbf{1.4332} (0.0791) & \textbf{0.27\%} (0.04\%) & \textbf{0.73} (0.00) & 55.0 (0.6)\\
Max-min hill climbing &  99.3\% &\textbf{0.1869} (0.0240) & \textbf{0.04\%} (0.01\%) & \textbf{0.98} (0.00) & \textbf{1.1} (0.1)\\
Constrained-MLE &  0\% &67.7678 (0.3708) & 38.67\% (0.17\%) & 0.18 (0.00) & 748.8 (1.6)\\
clrdag &  0\% &61.1686 (0.1008) & 42.72\% (0.06\%) & 0.44 (0.00) & 93.4 (0.3)\\
NOTEARS-linear &  0\% &\textbf{14.5058} (0.2740) & \textbf{12.92\%} (0.24\%) & \textbf{0.82} (0.00) & \textbf{18.5} (0.4)\\
DAGMA-linear &  0\% &25.0647 (0.2766) & 22.61\% (0.21\%) & 0.68 (0.00) & 39.1 (0.4)\\
NOTIME &  0\% &79.7030 (0.5253) & 48.13\% (0.27\%) & 0.45 (0.00) & 69.0 (0.4)\\
\hline
	\end{tabular} }\label{table:chain}
\end{table}

Additional simulation results are reported in the supplementary material. These include a weak-signal regime in which nonzero entries of $\mathbf{U}$ are drawn from $\text{Unif}([-0.5,-0.1]\cup[0.1,0.5])$, and a lower-dimensional regime with $p=25$, $n=100$. The supplementary tables also report false-positive rates, false-discovery rates, and KL divergence. Overall, these extended results confirm the robustness of DAGgr.

\section{Real Data Analysis}\label{sec:realdata}

We apply DAGgr to the dataset of \citet{sachs2005causal}, which contains simultaneous measurements of $11$ phosphorylated proteins and phospholipids collected from thousands of individual immune cells under both general and targeted molecular interventions. The interventions activate signaling pathways and, by revealing the effects of stimulatory and inhibitory perturbations, enable causal inference of edge directions. The objective is to recover the causal network among the eleven proteins from observational and interventional data combined. Because the underlying signaling pathways are well established in \citet{sachs2005causal}, we assess DAGgr by comparing its estimated network with the corresponding ``gold-standard'' network shown in Figure~\ref{fig:real_golden}.

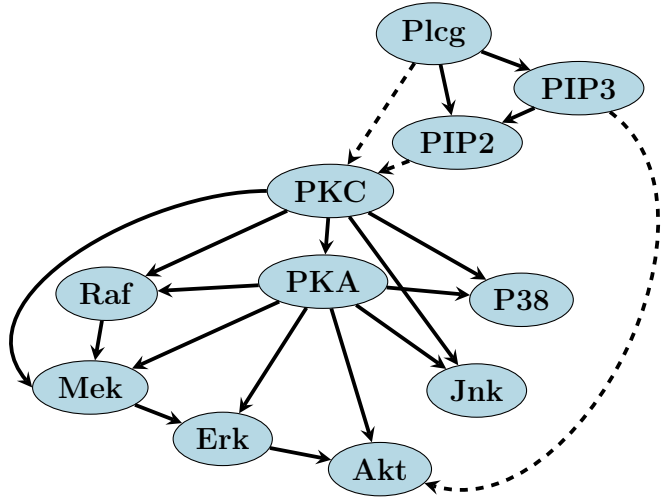
\begin{figure}[!h] 
\centering
\begin{tikzpicture}[  scale=0.8,transform shape,
    every node/.style={
        draw,
        ellipse,
        fill={rgb,255:red,182; green,216; blue,228},
        minimum width=1.3cm,
        minimum height=0.8cm,
        font=\large
    },
    arrow/.style={->,line width=0.5mm,>=stealth }
]

\node (Akt)  at (7.69, 0.77)  {\textbf{Akt}};
\node (Erk)  at (5.09, 1.27) {\textbf{Erk}};
\node (Jnk)  at (9.29, 2.07) {\textbf{Jnk}};
\node (Mek)  at (2.9, 2.12) {\textbf{Mek}};
\node (P38)  at (10.03,3.57) {\textbf{P38}};
\node (PIP2) at (8.97,6.17) {\textbf{PIP2}};
\node (PIP3) at (10.98,7.07) {\textbf{PIP3}};
\node (PKA)  at (6.75, 3.87) {\textbf{PKA}};
\node (PKC)  at (6.87, 5.37) {\textbf{PKC}};
\node (Plcg) at (8.58,7.97) {\textbf{Plcg}};
\node (Raf)  at (3.17, 3.67) {\textbf{Raf}};

\draw[arrow] (PIP3) -- (PIP2);
\draw[arrow] (Plcg) -- (PIP2);
\draw[arrow] (Plcg) -- (PIP3);
\draw[arrow,dashed] (Plcg) -- (PKC);
\draw[arrow,dashed] (PIP2) -- (PKC);

\draw[arrow] (PKC) -- (PKA);
\draw[arrow] (PKC) -- (Raf);
\draw[arrow, bend right=90] (PKC.west) to[out=-40] (Mek.west);
\draw[arrow] (PKC) -- (Jnk);
\draw[arrow] (PKC) -- (P38);

\draw[arrow] (PKA) -- (Raf);
\draw[arrow] (PKA) -- (Mek);
\draw[arrow] (PKA) -- (Erk);
\draw[arrow] (PKA) -- (Akt);
\draw[arrow] (PKA) -- (Jnk);
\draw[arrow] (PKA) -- (P38);

\draw[arrow] (Raf) -- (Mek);  
\draw[arrow] (Mek) -- (Erk);
\draw[arrow] (Erk) -- (Akt); 

\draw[arrow,dashed, bend left=80] (PIP3) to  (Akt);

\end{tikzpicture}
\caption{The golden standard causal network in \citet{sachs2005causal}}\label{fig:real_golden}
\end{figure}

It is worth noting that the empirical distributions of molecular concentrations deviate substantially from normality. Most variables exhibit pronounced right-skewness: concentrations are strictly positive, with a large proportion of small values clustered near zero, leading to highly asymmetric distributions that violate the Gaussian assumptions underlying Gaussian Bayesian networks. To mitigate this issue, we apply a logarithmic transformation to all concentration measurements, which yields distributions that are closer to normal.

To implement DAGgr on this dataset, we use ten candidate methods spanning several classes of causal-discovery approaches. The first five---GES \citep{chickering2002optimal}, the PC algorithm \citep{spirtes2000causation, kalisch2007estimating}, constrained maximum likelihood estimation \citep{yuan2012maximum}, its variant \texttt{clrdag} \citep{li2020likelihood}, and the linear version of NOTEARS \citep{zheng2018dags}---are carried over from the simulation study. To broaden the candidate set, we add five flexible nonlinear methods: NOTEARS-MLP \citep{zheng2020learning}, which extends NOTEARS with multilayer perceptrons; Causal Additive Models (CAM) \citep{CAM}, designed for nonlinear additive dependencies; R2Sort \citep{reisach2023scale}, a scalable ordering-based approach; NHTS \citep{hiremath2024hybrid}, a hybrid constraint- and score-based learner; and RESIT \citep{RESIT}, a functional causal model tailored for nonlinear, non-Gaussian settings. Together these ten methods provide a diverse and representative collection of candidates, all of which have been previously applied to the \citet{sachs2005causal} data.

Although the log-transformed data are approximately Gaussian, mild heteroscedasticity persists. For settings with node-specific noise variances, we use the heteroscedastic extension of Algorithm~\ref{alg:ARM-weighting-method-1} described in Appendix \ref{app:hetero-algo}, which replaces the homoscedastic likelihood in Step~3 with a per-node Gaussian likelihood and introduces a dispersion weight $\gamma\in[0,1]$ that controls how strongly the validation score penalizes the estimated log-variances $\log\widehat{\sigma}_{k,j}^{2}$; we use $\gamma=1$ (the full heteroscedastic likelihood) throughout the simulations and revisit the choice of $\gamma$ in the real-data analysis. We consider two pruned specifications: a baseline with $\lambda=1$, $\gamma=0$ and an alternative with $\lambda=1$, $\gamma=0.03$. Both are obtained via the step-size pruned procedure of Algorithm~\ref{algo:prune}. Consistent with the simulation study, we compare against the bootstrap-aggregation method DAGBag \citep{wang2014learning}. Table~\ref{table:real_results} reports the mean squared error $\frac{1}{np}\|\mathbf{X}(\mathbf{I}_{p}-\widehat{\mathbf{U}})\|_{F}^{2}$ together with true positives (TP), false positives (FP), false negatives (FN), MCC, F1 score, and SHD; DAGBag is excluded from MSE since it returns only a binary aggregated adjacency matrix. We additionally report the DAGgr weights assigned to each candidate under the two pruned specifications.

\begin{table}[h]
\centering
\caption{Performance comparison on the \citet{sachs2005causal} dataset}
\resizebox{\textwidth}{!}{
\begin{tabular}{cccccccccc}
\hline
 \multirow{2}{*}{Method} & \multicolumn{2}{c}{Weights} & \multirow{2}{*}{MSE} &  \multirow{2}{*}{TP} & \multirow{2}{*}{FP} & \multirow{2}{*}{FN} & \multirow{2}{*}{MCC} & \multirow{2}{*}{F1 score} & \multirow{2}{*}{SHD}\\
 &  $\gamma=0$ &  $\gamma=0.03$  & & & & & & &\\
\hline  
DAGgr step-size pruned($\gamma=0$) &  - &- & 7.76 & 7 & 3 & 10 & \textbf{0.48} & 0.52 & \textbf{11}\\
DAGgr step-size pruned($\gamma=0.03$) &- &-  & \textbf{6.26} & \textbf{9} & 7 & \textbf{8} & 0.47 & \textbf{0.55} & 13\\
DAGBag & -& - & - &5 & \textbf{2} & 12 & 0.40 & 0.42 & 12\\
\hline\hline
GES &  0.34\%  &0\% & 8.70 & 4 & 3 & 13 & 0.30 & 0.33 & 13\\ 
PC-algorithm & 0.26\% &0\%  & 9.23 & 5 & 3 & 12 & 0.36 & 0.40 & 13\\ 
Constrained-MLE &24.76\% &47.08\%  & \textbf{3.78} & \textbf{9} & 7 & \textbf{8} & \textbf{0.47} & \textbf{0.55} & 13\\ 
clrdag & 9.85\% &20.09\%  & \textbf{3.26} & 3 & 11 & 14 & 0.06 & 0.19 & 21\\ 
NOTEARS-linear &  5.77\% &0\%  & 9.49 & 4 & \textbf{2} & 13 & 0.34 & 0.35 & 13\\ 
NOTEARS-MLP &  30.19\% &0.15\%  & 9.58 & 5 & \textbf{1} & 12 & \textbf{0.45} & \textbf{0.43} & \textbf{12}\\ 
CAM & 14.04\%  &1.38\%   & 9.56 & \textbf{6} & 5 & \textbf{11} & 0.36 & \textbf{0.43} & \textbf{12}\\ 
R2Sort &5.25\% &0\%  & 9.32 & 4 & 3 & 13 & 0.30 & 0.33 & 13\\ 
NHTS &  9.15\% &10.75\%  & 9.53 & \textbf{6} & 5 & \textbf{11} & 0.36 & 0.43 & 13\\ 
RESIT & 0.39\%   &20.55\%  & 9.46 & \textbf{6} & 10 & \textbf{11} & 0.25 & 0.36 & 18\\ 
\hline
\end{tabular} }\label{table:real_results}
\end{table}


The results show that the proposed pruned models achieve competitive and stable performance across multiple evaluation metrics. The $\gamma=0.03$ specification improves edge recovery and reduces MSE relative to $\gamma=0$, achieving the highest F1 score ($0.55$) and an MCC ($0.47$) comparable to the best-performing individual method, Constrained-MLE. The weights also display the adaptive nature of the framework: under $\gamma=0.03$, more weight is assigned to methods better suited to complex dependency structures. Overall, DAGgr offers a favorable trade-off between estimation accuracy and structural recovery, consistently matching or exceeding strong individual competitors while maintaining stability.

Figure~\ref{fig:real_results} visualizes the two pruned models. Black edges denote correctly identified causal relationships consistent with the gold-standard network; red edges indicate missed true links; green edges correspond to spurious edges absent from the gold standard; and orange edges flag incorrectly oriented (reversed) causal directions.

\begin{figure}[!h] 
\centering
	\subfloat[Pruned model with $\gamma=0$]{ \label{fig:gamma0}
    \begin{tikzpicture}[  scale=0.65,transform shape,
    every node/.style={
        draw,
        ellipse,
        fill={rgb,255:red,182; green,216; blue,228},
        minimum width=1.3cm,
        minimum height=0.8cm,
        font=\large
    },
    arrow/.style={->,line width=0.5mm,>=stealth } 
]
\node (Akt)  at (7.69, 0.77)  {\textbf{Akt}};
\node (Erk)  at (5.09, 1.27) {\textbf{Erk}};
\node (Jnk)  at (9.29, 2.07) {\textbf{Jnk}};
\node (Mek)  at (2.9, 2.12) {\textbf{Mek}};
\node (P38)  at (10.03,3.57) {\textbf{P38}};
\node (PIP2) at (8.97,6.17) {\textbf{PIP2}};
\node (PIP3) at (10.98,7.07) {\textbf{PIP3}};
\node (PKA)  at (6.75, 3.87) {\textbf{PKA}};
\node (PKC)  at (6.87, 5.37) {\textbf{PKC}};
\node (Plcg) at (8.58,7.97) {\textbf{Plcg}};
\node (Raf)  at (3.17, 3.67) {\textbf{Raf}};

\draw[arrow,orange] (PIP2) -- (PIP3);
\draw[arrow,orange] (Jnk) -- (PKC);
\draw[arrow, red,  line width=0.5mm,
      postaction={
        decorate,
        decoration={
          markings,
          mark= at position 0.5 with {
            \draw[-] (-0.15,-0.15) -- (0.15,0.15);
            \draw[-] (-0.15,0.15) -- (0.15,-0.15);
          }
        }
      }]  (Plcg) -- (PIP2);
\draw[arrow] (Plcg) -- (PIP3);
\draw[arrow] (Erk) -- (Akt);  
\draw[arrow] (PKA) -- (Erk);
\draw[arrow] (PKA) -- (Akt);
\draw[arrow] (Raf) -- (Mek);  
\draw[arrow] (PKC) -- (P38);
\draw[arrow] (PKA) -- (P38);

\draw[arrow,green] (Jnk) -- (P38); 

\draw[arrow, red,  line width=0.5mm,
      postaction={
        decorate,
        decoration={
          markings,
          mark= at position 0.5 with {
            \draw[-] (-0.15,-0.15) -- (0.15,0.15);
            \draw[-] (-0.15,0.15) -- (0.15,-0.15);
          }
        }
      }
]  (PKC) -- (PKA);
\draw[arrow, red,  line width=0.5mm,
      postaction={
        decorate,
        decoration={
          markings,
          mark= at position 0.5 with {
            \draw[-] (-0.15,-0.15) -- (0.15,0.15);
            \draw[-] (-0.15,0.15) -- (0.15,-0.15);
          }
        }
      }
]  (PKC) -- (Raf);
\draw[arrow, red,  line width=0.5mm,bend right=90,
      postaction={
        decorate,
        decoration={
          markings,
          mark= at position 0.5 with {
            \draw[-] (-0.15,-0.15) -- (0.15,0.15);
            \draw[-] (-0.15,0.15) -- (0.15,-0.15);
          }
        }
      }
] (PKC.west) to[out=-40] (Mek.west);

\draw[arrow, red,  line width=0.5mm,
      postaction={
        decorate,
        decoration={
          markings,
          mark= at position 0.5 with {
            \draw[-] (-0.15,-0.15) -- (0.15,0.15);
            \draw[-] (-0.15,0.15) -- (0.15,-0.15);
          }
        }
      }] (PKA) -- (Raf);
\draw[arrow, red,  line width=0.5mm,
      postaction={
        decorate,
        decoration={
          markings,
          mark= at position 0.5 with {
            \draw[-] (-0.15,-0.15) -- (0.15,0.15);
            \draw[-] (-0.15,0.15) -- (0.15,-0.15);
          }
        }
      }] (PKA) -- (Mek);
\draw[arrow, red,  line width=0.5mm,
      postaction={
        decorate,
        decoration={
          markings,
          mark= at position 0.5 with {
            \draw[-] (-0.15,-0.15) -- (0.15,0.15);
            \draw[-] (-0.15,0.15) -- (0.15,-0.15);
          }
        }
      }
]  (PKA) -- (Jnk);

\draw[arrow, red,  line width=0.5mm,
      postaction={
        decorate,
        decoration={
          markings,
          mark= at position 0.5 with {
            \draw[-] (-0.15,-0.15) -- (0.15,0.15);
            \draw[-] (-0.15,0.15) -- (0.15,-0.15);
          }
        }
      }
] (Mek) -- (Erk);

\end{tikzpicture} 
} 
	\subfloat[Pruned model with $\gamma=0.03$]{\label{fig:gamma3}
    \begin{tikzpicture}[ scale=0.65,transform shape,
    every node/.style={
        draw,
        ellipse,
        fill={rgb,255:red,182; green,216; blue,228},
        minimum width=1.3cm,
        minimum height=0.8cm,
        font=\large
    },
    arrow/.style={->,line width=0.5mm,>=stealth } 
]
\node (Akt)  at (7.69, 0.77)  {\textbf{Akt}};
\node (Erk)  at (5.09, 1.27) {\textbf{Erk}};
\node (Jnk)  at (9.29, 2.07) {\textbf{Jnk}};
\node (Mek)  at (2.9, 2.12) {\textbf{Mek}};
\node (P38)  at (10.03,3.57) {\textbf{P38}};
\node (PIP2) at (8.97,6.17) {\textbf{PIP2}};
\node (PIP3) at (10.98,7.07) {\textbf{PIP3}};
\node (PKA)  at (6.75, 3.87) {\textbf{PKA}};
\node (PKC)  at (6.87, 5.37) {\textbf{PKC}};
\node (Plcg) at (8.58,7.97) {\textbf{Plcg}};
\node (Raf)  at (3.17, 3.67) {\textbf{Raf}};

\draw[arrow,orange] (PIP2) -- (PIP3);
\draw[arrow, red,  line width=0.5mm,
      postaction={
        decorate,
        decoration={
          markings,
          mark= at position 0.5 with {
            \draw[-] (-0.15,-0.15) -- (0.15,0.15);
            \draw[-] (-0.15,0.15) -- (0.15,-0.15);
          }
        }
      }
]  (Plcg) -- (PIP2);
\draw[arrow] (Plcg) -- (PIP3);

\draw[arrow, red,  line width=0.5mm,
      postaction={
        decorate,
        decoration={
          markings,
          mark= at position 0.5 with {
            \draw[-] (-0.15,-0.15) -- (0.15,0.15);
            \draw[-] (-0.15,0.15) -- (0.15,-0.15);
          }
        }
      }
]  (PKC) -- (PKA);
\draw[arrow, red,  line width=0.5mm,
      postaction={
        decorate,
        decoration={
          markings,
          mark= at position 0.5 with {
            \draw[-] (-0.15,-0.15) -- (0.15,0.15);
            \draw[-] (-0.15,0.15) -- (0.15,-0.15);
          }
        }
      }
]  (PKC) -- (Raf);
\draw[arrow, red,  line width=0.5mm,bend right=90,
      postaction={
        decorate,
        decoration={
          markings,
          mark= at position 0.5 with {
            \draw[-] (-0.15,-0.15) -- (0.15,0.15);
            \draw[-] (-0.15,0.15) -- (0.15,-0.15);
          }
        }
      }
] (PKC.west) to[out=-40] (Mek.west);
\draw[arrow,orange] (Jnk) -- (PKC);
\draw[arrow] (PKC) -- (P38);

\draw[arrow] (PKA) -- (Raf);
\draw[arrow] (PKA) -- (Mek);
\draw[arrow] (PKA) -- (Erk);
\draw[arrow] (PKA) -- (Akt);
\draw[arrow, red,  line width=0.5mm,
      postaction={
        decorate,
        decoration={
          markings,
          mark= at position 0.5 with {
            \draw[-] (-0.15,-0.15) -- (0.15,0.15);
            \draw[-] (-0.15,0.15) -- (0.15,-0.15);
          }
        }
      }
]  (PKA) -- (Jnk);
\draw[arrow] (PKA) -- (P38);

\draw[arrow] (Raf) -- (Mek);  
\draw[arrow, red,  line width=0.5mm,
      postaction={
        decorate,
        decoration={
          markings,
          mark= at position 0.5 with {
            \draw[-] (-0.15,-0.15) -- (0.15,0.15);
            \draw[-] (-0.15,0.15) -- (0.15,-0.15);
          }
        }
      }
] (Mek) -- (Erk); 
\draw[arrow] (Erk) -- (Akt); 


\draw[arrow,green,bend right=28] (Mek) to (Akt);
\draw[arrow,green] (Mek) -- (P38);
\draw[arrow,green] (Jnk) -- (P38);
\draw[arrow,green,bend left=60] (P38) to (Akt.east);
\draw[arrow,green] (Akt) .. controls (14,0) and (14,9) ..   (Plcg) ; 
\end{tikzpicture} } 
\caption{Estimated DAGs obtained from the two pruned specifications}
\label{fig:real_results}
\end{figure}
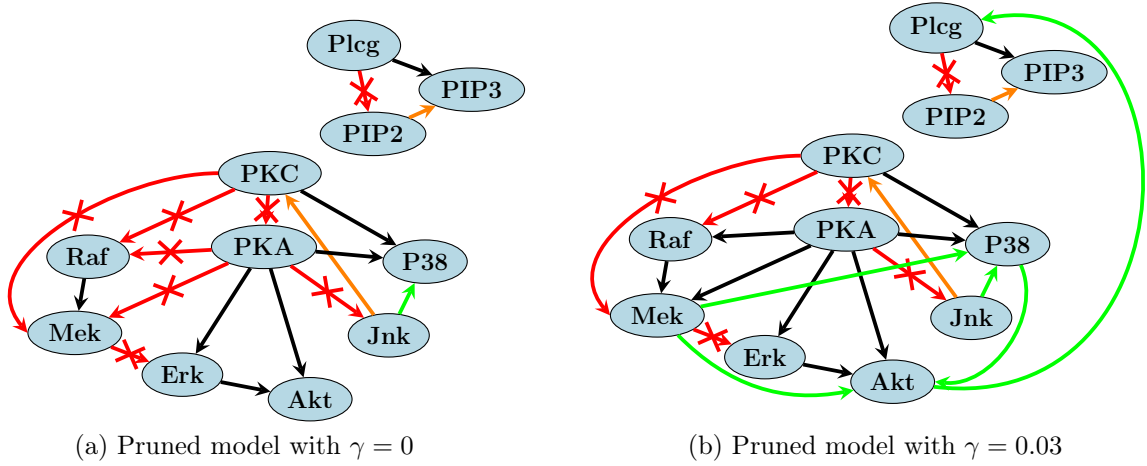

Figure~\ref{fig:real_results} shows that both DAGgr step-size pruned specifications recover a substantial portion of the known signaling structure. The $\gamma=0.03$ model yields a denser and more informative reconstruction, capturing additional true-positive edges---particularly those involving PKA and PKC and their downstream signaling relationships---at the cost of a moderate increase in spurious or incorrectly oriented edges (green and orange links). The $\gamma=0$ specification produces a sparser graph with fewer false positives but more missed true edges, and it achieves the lowest SHD ($11$) among all candidate DAGs. The comparison illustrates the expected trade-off between sparsity and sensitivity: larger $\gamma$ improves recall but slightly reduces structural precision.

To address the disagreement among candidate DAGs on edge directions---which can produce ``suspect'' directed edges in the pruned graph---we also apply the PDAG procedure of Section~\ref{sec: undirect} under both specifications. The resulting PDAGs are shown in Figure~\ref{fig:real_results_agg}, where blue lines denote undirected edges.

\begin{figure}[!h] 
\centering
	\subfloat[PDAG with $\gamma=0$,  $\tau=0.4$ and $\delta=0.2$]{ \label{fig:gamma0_agg}
    \begin{tikzpicture}[  scale=0.65,transform shape,
    every node/.style={
        draw,
        ellipse,
        fill={rgb,255:red,182; green,216; blue,228},
        minimum width=1.3cm,
        minimum height=0.8cm,
        font=\large
    },
    arrow/.style={->,line width=0.5mm,>=stealth } 
]
\node (Akt)  at (7.69, 0.77)  {\textbf{Akt}};
\node (Erk)  at (5.09, 1.27) {\textbf{Erk}};
\node (Jnk)  at (9.29, 2.07) {\textbf{Jnk}};
\node (Mek)  at (2.9, 2.12) {\textbf{Mek}};
\node (P38)  at (10.03,3.57) {\textbf{P38}};
\node (PIP2) at (8.97,6.17) {\textbf{PIP2}};
\node (PIP3) at (10.98,7.07) {\textbf{PIP3}};
\node (PKA)  at (6.75, 3.87) {\textbf{PKA}};
\node (PKC)  at (6.87, 5.37) {\textbf{PKC}};
\node (Plcg) at (8.58,7.97) {\textbf{Plcg}};
\node (Raf)  at (3.17, 3.67) {\textbf{Raf}};

\draw[arrow] (Erk) -- (Akt);   
\draw[arrow] (PKA) -- (Akt);   
\draw[arrow] (PKA) -- (Erk);   
\draw[arrow] (Raf) -- (Mek);    
\draw[arrow] (PKA) -- (P38); 
\draw[arrow] (PKC) -- (P38);

\draw[arrow,green] (Jnk) -- (P38); 
\draw[arrow,orange] (Jnk) -- (PKC);
\draw[arrow,orange] (PIP2) -- (PIP3); 

\draw[blue,line width=0.5mm] (Plcg) -- (PIP3);

\end{tikzpicture} 
} 
	\subfloat[PDAG with $\gamma=0.03$,  $\tau=0.45$ and $\delta=0.2$]{\label{fig:gamma3_agg}
    \begin{tikzpicture}[ scale=0.65,transform shape,
    every node/.style={
        draw,
        ellipse,
        fill={rgb,255:red,182; green,216; blue,228},
        minimum width=1.3cm,
        minimum height=0.8cm,
        font=\large
    },
    arrow/.style={->,line width=0.5mm,>=stealth } 
]
\node (Akt)  at (7.69, 0.77)  {\textbf{Akt}};
\node (Erk)  at (5.09, 1.27) {\textbf{Erk}};
\node (Jnk)  at (9.29, 2.07) {\textbf{Jnk}};
\node (Mek)  at (2.9, 2.12) {\textbf{Mek}};
\node (P38)  at (10.03,3.57) {\textbf{P38}};
\node (PIP2) at (8.97,6.17) {\textbf{PIP2}};
\node (PIP3) at (10.98,7.07) {\textbf{PIP3}};
\node (PKA)  at (6.75, 3.87) {\textbf{PKA}};
\node (PKC)  at (6.87, 5.37) {\textbf{PKC}};
\node (Plcg) at (8.58,7.97) {\textbf{Plcg}};
\node (Raf)  at (3.17, 3.67) {\textbf{Raf}};

\draw[arrow,orange] (PIP2) -- (PIP3);
\draw[arrow] (Plcg) -- (PIP3);

\draw[line width=0.5mm,blue] (Jnk) -- (PKC);
\draw[arrow] (PKC) -- (P38);

\draw[arrow] (PKA) -- (Raf);
\draw[arrow] (PKA) -- (Mek);
\draw[line width=0.5mm,blue] (PKA) -- (Erk);
\draw[arrow] (PKA) -- (Akt);
 
\draw[arrow] (PKA) -- (P38);

\draw[arrow] (Raf) -- (Mek);  
 
\draw[arrow] (Erk) -- (Akt);

\draw[arrow,green,bend right=28] (Mek) to (Akt);
\draw[arrow,green] (Mek) -- (P38);
\draw[line width=0.5mm,blue] (Jnk) -- (P38);
\draw[arrow,green,bend left=60] (P38) to (Akt.east);
\draw[arrow,green] (Akt) .. controls (14,0) and (14,9) ..   (Plcg) ; 
\end{tikzpicture} } 
\caption{Estimated PDAGs obtained from the two pruned specifications}
\label{fig:real_results_agg}
\end{figure}
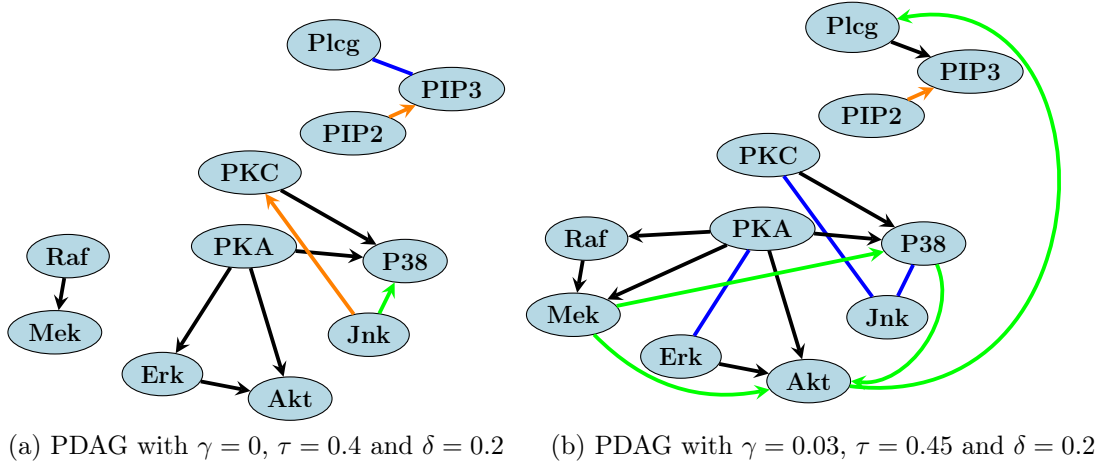

Figure~\ref{fig:gamma0_agg} presents the PDAG constructed from the $\gamma=0$ specification with $\tau=0.4$ and $\delta=0.2$, while Figure~\ref{fig:gamma3_agg} corresponds to the $\gamma=0.03$ specification with $\tau=0.45$ and $\delta=0.2$. The figures suggest that the procedure effectively identifies uncertain edge directions by representing many reversed or potentially spurious edges as undirected, producing a more conservative and stable characterization of the underlying causal structure, particularly in settings where the estimated DAG contains a relatively large number of potential false positives.

\section{Conclusion and Discussion}\label{sec:conclusion}
We developed DAGgr, a model-averaging framework for learning Directed Acyclic Graphs. By weighting candidate DAGs through out-of-sample predictive likelihood across repeated data splits and combining them under an importance-score threshold, DAGgr produces an aggregated graph that is provably acyclic, more stable than any individual candidate, and more interpretable. Our theoretical analysis established a finite-sample risk bound, proved acyclicity preservation, and showed that edge selection is consistent under mild conditions on the weights. Simulations and the application to the \citet{sachs2005causal} dataset confirm that DAGgr reduces false positives, improves true-edge detection, and achieves lower structural Hamming distance than single-model alternatives.

Several extensions warrant further study. Promising directions include extending the framework to partially observed data, integrating prior information about causal directions, and incorporating nonparametric or deep-learning-based candidate estimators. Adaptive thresholding strategies and fast approximation algorithms would further enhance computational scalability. We expect the ideas developed here to encourage broader adoption of model averaging in causal structure learning, ultimately enabling more reliable and actionable causal inference in complex systems. A further extension worth pursuing is relaxing the homoscedastic noise assumption in model~\eqref{eq:model}. The current likelihood treats $\sigma$ as common across nodes, but in practice the noise scale often varies by variable---a feature already explored in our heteroscedastic simulations, where each $X_{j}$ is generated with its own node-specific variance $\sigma_{j}^{2}$. Extending the framework to fit and weight candidates under a node-dependent noise model $\epsilon_{j}\overset{\textrm{ind.}}{\sim}N(0,\sigma_{j}^{2})$ is conceptually straightforward---the per-observation log-likelihood in~\eqref{eq:likelihood} becomes a sum over node-specific terms---and could improve both the calibration of the validation-likelihood weights and the recovery of edges incident to high-noise nodes. A related direction is observation-dependent (data-driven) heteroscedasticity, where $\sigma_{j}$ depends on the values of the parents; this connects DAGgr to recent identifiability results that exploit noise heterogeneity to resolve Markov-equivalence ambiguity.

\section*{\normalsize Generative AI Use Statement}
During the preparation of this work, the authors used Opus 4.7, Anthropic to summarize background literature and to improve the language and readability of the manuscript. After using this tool, the authors reviewed, verified, and edited the content as needed, and take full responsibility for the content of the publication.

 \appendix
 \section{Proofs}
 \subsection{Proof of Theorem 1}
 \begin{proof}
Without confusion, we denote $l_{k}(\mathbf{x}_{i}):=l_{\hat{\bU}^{(k)},\hat{\sigma}_{k}}(\mathbf{x}_{i})$. We also write
\[
L(\mathbf{x},\mathbf{y}):=\|\mathbf{x}-\mathbf{y}\|^{2}=\sum_{j=1}^{p}(x_{j}-y_{j})^{2}
\]
for the squared-error loss between $\mathbf{x},\mathbf{y}\in\mathbb{R}^{p}$. Under the Gaussian model~\eqref{eq:model}, the per-observation log-likelihood satisfies $l_{\mathbf{U},\sigma^{2}}(\mathbf{x}_{i})=-\frac{1}{2\sigma^{2}}L(\mathbf{x}_{i},\mathbf{U}\mathbf{x}_{i})-\frac{p}{2}\log\sigma^{2}$, so $L$ is the natural loss associated with the likelihood.

Denote 
\[
q_{\ntr}^{n}=\sum_{k=1}^{K}\pi_{k}\exp\left\{ \lambda\sum_{i=\ntr+1}^{n}l_{k}(\mathbf{x}_{i})\right\} ,
\]
which can be decomposed as 
\begin{align}
q_{\ntr}^{n}= & \sum_{k=1}^{K}\pi_{k}\exp\left\{ \lambda l_{k}(\mathbf{x}_{\ntr+1})\right\} \times\frac{\sum_{k=1}^{K}\pi_{k}\exp\left\{ \lambda\sum_{i=\ntr+1}^{\ntr+2}l_{k}(\mathbf{x}_{i})\right\} }{\sum_{k=1}^{K}\pi_{k}\exp\left\{ \lambda l_{k}(\mathbf{x}_{\ntr+1})\right\} }\nonumber \\
 & \times\cdots\times\frac{\sum_{k=1}^{K}\pi_{k}\exp\left\{ \lambda\sum_{i=\ntr+1}^{n}l_{k}(\mathbf{x}_{i})\right\} }{\sum_{k=1}^{K}\pi_{k}\exp\left\{ \lambda\sum_{i=\ntr+1}^{n-1}l_{k}(\mathbf{x}_{i})\right\} }\nonumber \\
= & \prod_{i=\ntr+1}^{n}\left(\sum_{k=1}^{K}w_{k,i}\exp\left\{ \lambda l_{k}(\mathbf{x}_{i})\right\} \right).\label{eq:qn_decomp}
\end{align}

For a fixed $i\in\{\ntr+1,...,n\}$, let $J$ be a discrete random
variable with $\ensuremath{P(J=k)=w_{k,i}}$, where $k=1,...,K$.
Let $\nu$ be the discrete measure induced by $J$ on $\mathbb{Z}^{+}$
such that $\text{\ensuremath{\nu(\{k\})=P(J=k).}}$ Denote $h(J)=-L(\mathbf{x}_{i},\widehat{\bU}^{(J)}\mathbf{x}_{i})$
and we have 
\[
\sum_{k=1}^{K}w_{k,i}\exp\left\{ \lambda l_{k}(\mathbf{x}_{i})\right\} =\E_{\nu}\exp(\lambda h(J)).
\]

By Lemma 3.6.1 of \citet{catoni2004statistical}, we have 
\[
\log \E_{\nu}\exp(\lambda h(J))\le\lambda \E_{\nu}h(J)+\frac{\lambda^{2}}{2}\mathrm{Var}_{\nu}(h(J))\exp\left(\lambda\max\left\{ 0,\sup_{\gamma\in[0,\lambda]}\dfrac{M_{\nu_{\gamma}}^{3}(h(J))}{\mathrm{Var}_{\nu_{\gamma}}(h(J))}\right\} \right),
\]
where $\nu_{\gamma}(k)=\frac{w_{k,i}\exp(\gamma h(k))}{\sum_{k=1}^{K}w_{k,i}\exp(\gamma h(k))}$
for $k\ge1$ and $M_{\nu_{\gamma}}^{3}(h(J))=\E_{\nu_{\gamma}}(h(J)-\E_{\nu_{\gamma}}(h(J)))^{3}.$
Thus, we have

\begin{align*}\sup_{\gamma\in[0,\lambda]}\frac{M_{\nu_{\gamma}}^{3}(h(J))}{\mathrm{Var}_{\nu_{\gamma}}(h(J))} & \le\sup_{\gamma\in[0,\lambda]}\sup_{k\ge0}|h(k)-\E_{\nu_{\gamma}}h(J)|\\
 & \le2\sup_{k\ge0}\left|L\left(\mathbf{x}_{i},\widehat{\bU}^{(k)}\mathbf{x}_{i}\right)-L\left(\mathbf{x}_{i},\bU\mathbf{x}_{i}\right)\right|\\
 & \le2\sup_{k\ge0}\{|(\mathbf{x}_{i}-\bU\mathbf{x}_{i})^\top\cdot(\widehat{\bU}^{(k)}\mathbf{x}_{i}-\bU\mathbf{x}_{i})|+2(\widehat{\bU}^{(k)}\mathbf{x}_{i}-\bU\mathbf{x}_{i})^{2}\}\\
 & \le2|\epsilon_{i}|\cdot\sqrt{p}A|\mathbf{x}_{i}|+4pA^{2}|\mathbf{x}_{i}|^{2}\\
 & \le2\sqrt{p}A|\mathbf{x}_{i}|\cdot\sqrt{p}A|\mathbf{x}_{i}|+4pA^{2}|\mathbf{x}_{i}|^{2}\\
 & =6pA^{2}|\mathbf{x}_{i}|^{2}
\end{align*}

and

\begin{align*}\mathrm{Var}_{\nu}(h(J)) & \le \E_{\nu}\left(L(\mathbf{x}_{i},\widehat{\bU}^{(J)}\mathbf{x}_{i})-L(\mathbf{x}_{i},\E_{\nu}\widehat{\bU}^{(J)}\mathbf{x}_{i})\right)^{2}\\
 & \le\sup_{k\ge1}\left(2(\mathbf{x}_{i}-\widehat{\bU}^{(k)}\mathbf{x}_{i})+c_{2}|\widehat{\bU}^{(k)}\mathbf{x}_{i}-\E_{\nu}\widehat{\bU}^{(J)}\mathbf{x}_{i}|\right)^{2}\E_{v}\left(\widehat{\bU}^{(J)}\mathbf{x}_{i}-\E_{\nu}\widehat{\bU}^{(J)}\mathbf{x}_{i}\right)^{2}\\
 & \le\sup_{k\ge1}\left(2(\mathbf{x}_{i}-\bU\mathbf{x}_{i})+4c_{2}\sup_{k\ge1}|\widehat{\bU}^{(k)}\mathbf{x}_{i}-\bU\mathbf{x}_{i}|\right)^{2}\cdot \E_{v}\left(\widehat{\bU}^{(J)}\mathbf{x}_{i}-\E_{\nu}\widehat{\bU}^{(J)}\mathbf{x}_{i}\right)^{2}\\
 & \le\left(2|\epsilon_{i}|+4\sqrt{p}A|\mathbf{x}_{i}|\right)^{2}\E_{v}\left(\widehat{\bU}^{(J)}\mathbf{x}_{i}-\E_{\nu}\widehat{\bU}^{(J)}\mathbf{x}_{i}\right)^{2}.\\
 & \le\left(2\sqrt{p}A|\mathbf{x}_{i}|+4\sqrt{p}A|\mathbf{x}_{i}|\right)^{2}\E_{v}\left(\widehat{\bU}^{(J)}\mathbf{x}_{i}-\E_{\nu}\widehat{\bU}^{(J)}\mathbf{x}_{i}\right)^{2}
\end{align*}

We also have 
\begin{align*} & \E_{v}\left[L\left(\mathbf{x}_{i},\widehat{\bU}^{(J)}\mathbf{x}_{i}\right)-L\left(\mathbf{x}_{i},\E_{\nu}\widehat{\bU}^{(J)}\mathbf{x}_{i}\right)\right]\\
\ge & \E_{v}\left[2\left(\mathbf{x}_{i}-\E_{\nu}\widehat{\bU}^{(J)}\mathbf{x}_{i}\right)\left(\E_{\nu}\widehat{\bU}^{(J)}\mathbf{x}_{i}-\widehat{\bU}^{(J)}\mathbf{x}_{i}\right)\right]+c_{1}\E_{v}\left(\E_{\nu}\widehat{\bU}^{(J)}\mathbf{x}_{i}-\widehat{\bU}^{(J)}\mathbf{x}_{i}\right)^{2}\\
= & c_{1}\E_{v}\left(\E_{\nu}\widehat{\bU}^{(J)}\mathbf{x}_{i}-\widehat{\bU}^{(J)}\mathbf{x}_{i}\right)^{2},
\end{align*}

which leads to 
\[
\E_{v}(\E_{\nu}\widehat{\bU}^{(J)}\mathbf{x}_{i}-\widehat{\bU}^{(J)}\mathbf{x}_{i})^{2}\le\frac{1}{c_{1}}\E_{v}\left[L\left(\mathbf{x}_{i},\widehat{\bU}^{(J)}\mathbf{x}_{i}\right)-L\left(\mathbf{x}_{i},\E_{\nu}\widehat{\bU}^{(J)}\mathbf{x}_{i}\right)\right].
\]

We have  
\begin{align}
 \log \E_{v}\exp\{\lambda h(J)\}
\le & -\lambda \E_{\nu}L\left(\mathbf{x}_{i},\widehat{\bU}^{(J)}\mathbf{x}_{i}\right)\nonumber\\&+36\lambda^{2}pA^{2}(\exp(72e^{2}d_{1}^{2}\lambda^{2}pA^{2})+\frac{d_{4}}{2}\exp(36d_{3}\lambda^{2}pA^{2}))\nonumber\\
 & \cdot\frac{1}{c_{1}}\E_{v}\left[L\left(\mathbf{x}_{i},\widehat{\bU}^{(J)}\mathbf{x}_{i}\right)-L\left(\mathbf{x}_{i},\E_{\nu}\widehat{\bU}^{(J)}\mathbf{x}_{i}\right)\right].\label{eq:}
\end{align}
For any sub-exponential variable $Z$ we have $\E\exp\left(t|Z|\right)\le2\exp(2e^{2}d_{1}^{2}t^{2})$
and $\E\left[Z^{2}\exp\left\{ t|Z|\right\} \right]\le d_{4}\exp(d_{3}t^{2})$
for any $|t|\le d_{2}/d_{1}$, where $d_{1}=\sup_{k\ge1}k^{-1}(\E|Z|^{k})^{1/k}$,
$d_{2}=1/(4e)$, $d_{3}=8e^{4}d_{1}^{2}$, $d_{4}=16\sqrt{2}d_{1}^{2}$.
It follows that 
\[
\E\left[|Z|\exp\left\{ t|Z|\right\} \right]\le \E\left[\frac{1+|Z|^{2}}{2}\exp\left\{ t|Z|\right\} \right]\le\exp(2e^{2}d_{1}^{2}t^{2})+\frac{d_{4}\exp(d_{3}t^{2})}{2}.
\]
Since $|\mathbf{x}_{i}|^{2}$ follows a non-central chi-square distribution
and $\mathbf{x}_{i}$ follows a Gaussian distribution with mean $\mathbf{0}$,
we have $|\mathbf{x}_{i}|^{2}$ is sub-exponential. Thus the above
inequality holds if we replace $Z$ with $|\mathbf{x}_{i}|^{2}$. 

Taking expectation $\E_{\mathbf{x}_{i}|\mathbf{x}_{1},...,\mathbf{x}_{i-1},J}$
(denoted as $\E_{i}$ for convenience) of both sides of the inequality
\eqref{eq:}, for $6\lambda pA^{2}\le d_{2}/d_{1}$ we have 
\begin{align*}
 & \E_{i}\log \E_{v}\exp\{\lambda h(J)\}\\
\le & -\lambda \E_{i}\E_{\nu}L\left(\mathbf{x}_{i},\widehat{\bU}^{(J)}\mathbf{x}_{i}\right)+36\lambda^{2}pA^{2}\E_{i}(\exp(72e^{2}d_{1}^{2}\lambda^{2}pA^{2})+\frac{d_{4}}{2}\exp(36d_{3}\lambda^{2}pA^{2}))\\
 & \cdot\frac{1}{c_{1}}\E_{v}\left[L\left(\mathbf{x}_{i},\widehat{\bU}^{(J)}\mathbf{x}_{i}\right)-L\left(\mathbf{x}_{i},\E_{\nu}\widehat{\bU}^{(J)}\mathbf{x}_{i}\right)\right].
\end{align*}
 By choosing a small enough $\lambda$ such that 
\[
36\lambda^{2}pA^{2}\E_{i}(\exp(72e^{2}d_{1}^{2}\lambda^{2}pA^{2})+\frac{d_{4}}{2}\exp(36d_{3}\lambda^{2}pA^{2}))\text{\ensuremath{\le\lambda}},
\]
 which gives 
\begin{align*}
\E_{i}\log \E_{v}e^{\lambda h(J)}&\le-\lambda \E_{i}\E_{\nu}L\left(\mathbf{x}_{i},\widehat{\bU}^{(J)}\mathbf{x}_{i}\right)+\lambda \E_{i}\E_{v}\left[L\left(\mathbf{x}_{i},\widehat{\bU}^{(J)}\mathbf{x}_{i}\right)-L\left(\mathbf{x}_{i},\E_{\nu}\widehat{\bU}^{(J)}\mathbf{x}_{i}\right)\right]\\
&=-\lambda \E_{i}\left[L\left(\mathbf{x}_{i},\E_{\nu}\widehat{\bU}^{(J)}\mathbf{x}_{i}\right)\right].
\end{align*}
 We have 
\begin{align*}
\E\log(1/q_{\ntr}^{n})= & -\sum_{i=\ntr+1}^{n}\E\log\left(\sum_{k=1}^{K}w_{k,i}\exp\left\{ -\lambda l_{k}(\mathbf{x}_{i})\right\} \right)\\
= & -\sum_{i=\ntr+1}^{n}\E\E_{i}\log\left(\E_{v}\exp\left\{ -\lambda L(\mathbf{x}_{i},\widehat{\bU}^{(J)}\mathbf{x}_{i})\right\} \right)\\
\end{align*}
 We also have for each $k\ge1$, 
\begin{align*}
\E\log(1/q_{\ntr}^{n})\le & \log(1/\pi_{k})+\lambda\sum_{i=\ntr+1}^{n}\E L\left(\mathbf{x}_{i},\widehat{\bU}^{(J)}\mathbf{x}_{i}\right)=\log(1/\pi_{k})+\lambda(n-\ntr)\E L\left(\mathbf{X},\widehat{\bU}^{(J)}(\mathbf{X})\right).
\end{align*}
 Thus, by convexity, we have 
\begin{align*}
\E L\left(\mathbf{x}_{i},\tilde{\bU}\mathbf{x}_{i}\right)\le & \frac{1}{n-\ntr}\sum_{i=\ntr+1}^{n}\E L\left(\mathbf{x}_{i},\E_{v}\widehat{\bU}^{(J)}\mathbf{x}_{i}\right)\le\frac{\E\log(1/q_{\ntr}^{n})}{\lambda(n-\ntr)}\\
\le&\frac{\log(1/\pi_{k})}{\lambda(n-\ntr)}+\E L\left(\mathbf{X},\widehat{\bU}^{(J)}(\mathbf{X})\right),
\end{align*}
 where the desired result follows. 
\end{proof}

\subsection{Proof of Theorem 2}
\begin{proof}
Denote by ${\cal E}^{*}\backslash{\cal E}^{(k)}$ the set of variables
contained in ${\cal E}^{*}$ but not in ${\cal E}^{(k)}$. Since 
\begin{eqnarray*}
\dfrac{\sum_{k=1}^{K}w_{k}|{\cal E}^{*}\backslash{\cal E}^{(k)}|}{|{\cal E}^{*}|} & = & \dfrac{\sum_{k=1}^{K}w_{k}\sum_{(i,j)\in{\cal E}^{*}}\I((i,j)\notin{\cal E}^{(k)})}{|{\cal E}^{*}|}\\
 & = & \dfrac{\sum_{(i,j)\in{\cal E}^{*}}\sum_{k=1}^{K}w_{k}\I((i,j)\notin{\cal E}^{(k)})}{|{\cal E}^{*}|}\\
 & = & \dfrac{\sum_{(i,j)\in{\cal E}^{*}}\sum_{k=1}^{K}w_{k}(1-\I((i,j)\in{\cal E}^{(k)}))}{|{\cal E}^{*}|}\\
 & = & \dfrac{\sum_{(i,j)\in{\cal E}^{*}}(1-s_{ij})}{|{\cal E}^{*}|}.
\end{eqnarray*}
 and by the definition of weak consistency, 
\[
0\leq\dfrac{\sum_{k=1}^{K}w_{k}|{\cal E}^{*}\backslash{\cal E}^{(k)}|}{|{\cal E}^{*}|}\leq\dfrac{\sum_{k=1}^{K}w_{k}|{\cal E}^{(k)}\nabla{\cal E}^{*}|}{|{\cal E}^{*}|}{\displaystyle \ \overset{p}{\to}\ }0.
\]
 Hence, 
\begin{eqnarray*}
\dfrac{\sum_{(i,j)\in{\cal E}^{*}}(1-s_{ij})}{|{\cal E}^{*}|} & {\displaystyle \overset{p}{\to}} & 0.
\end{eqnarray*}
 On the other hand, 
\begin{eqnarray*}
\dfrac{\sum_{(i,j)\notin{\cal E}^{*}}s_{ij}}{|{\cal E}^{*}|} & = & \dfrac{\sum_{(i,j)\notin{\cal E}^{*}}\sum_{k=1}^{K}w_{k}\I((i,j)\in{\cal E}^{(k)})}{|{\cal E}^{*}|}\\
 & = & \dfrac{\sum_{k=1}^{K}w_{k}\sum_{(i,j)\notin{\cal E}^{*}}\I((i,j)\in{\cal E}^{(k)})}{|{\cal E}^{*}|}\\
 & = & \dfrac{\sum_{k=1}^{K}w_{k}|{\cal E}^{(k)}\backslash{\cal E}^{*}|}{|{\cal E}^{*}|}\\
 & \leq & \dfrac{\sum_{k=1}^{K}w_{k}|{\cal E}^{(k)}\nabla{\cal E}^{*}|}{|{\cal E}^{*}|}\overset{p}{\to}0.
\end{eqnarray*}
\end{proof}
\subsection{Proof of Theorem 3}
\begin{proof}
We prove by contradiction. Assume $\frac{|\{(i,j)\in{\cal E}^{*}:s_{ij}\le c\}|}{|{\cal E}^{*}|}$
does not converge to $0$ in probability as $n$ tends to infinity
($|{\cal E}^{*}|$ may or may not depend on $n$), then there exists
a positive constant $\epsilon_{0}$, such that $P\left(\frac{|\{(i,j)\in{\cal E}^{*}:s_{ij}\le c\}|}{|{\cal E}^{*}|}\geq\epsilon_{0}\right)$
does not converge to 0. On the other hand,
\begin{eqnarray*}
\dfrac{\sum_{(i,j)\in{\cal E}^{*}}(1-s_{ij})}{|{\cal E}^{*}|} & = & \dfrac{\sum_{(i,j)\in{\cal E}^{*},s_{ij}\leq c}(1-s_{ij})}{|{\cal E}^{*}|}+\dfrac{\sum_{(i,j)\in{\cal E}^{*},s_{ij}>c}(1-s_{ij})}{|{\cal E}^{*}|}\\
 & \geq & \dfrac{\sum_{(i,j)\in{\cal E}^{*},s_{ij}\leq c}(1-s_{ij})}{|{\cal E}^{*}|}\\
 & \geq & \dfrac{\sum_{(i,j)\in{\cal E}^{*},s_{ij}\leq c}(1-c)}{|{\cal E}^{*}|}\\
 & = & (1-c)\frac{|\{(i,j)\in{\cal E}^{*}:s_{ij}\le c\}|}{|{\cal E}^{*}|}.
\end{eqnarray*}
 So we have $P\left(\dfrac{\sum_{(i,j)\in{\cal E}^{*}}(1-s_{ij})}{|{\cal E}^{*}|}\geq(1-c)\epsilon_{0}\right)\geq P(\frac{|\{(i,j)\in{\cal E}^{*}:s_{ij}\le c\}|}{|{\cal E}^{*}|}\geq\epsilon_{0})$,
which does not converge to 0. But this contradicts with Theorem \ref{(a)-Assume-that}.
Hence, we have $\frac{|\{(i,j)\in{\cal E}^{*}:s_{ij}\le c\}|}{|{\cal E}^{*}|}{\displaystyle \ \overset{p}{\to}\ }0.$
Similarly, we can prove $\frac{|\{(i,j)\notin{\cal E}^{*}:s_{ij}>c\}|}{|{\cal E}^{*}|}{\displaystyle \ \overset{p}{\to}\ }0$.
\end{proof}
\section{Heteroscedastic data-splitting procedure with $\gamma$-tempered dispersion penalty}\label{app:hetero-algo}
\begin{algorithm}[H]
{\bf Input:} Candidate DAGs $\{(\mathcal{V}_{k},\mathcal{E}_{k})\}_{k=1}^{K}$, with optional estimated adjacency matrices $\{\widehat{\mathbf{U}}^{(k)}\}_{k=1}^{K}$ and node-specific noise standard deviations $\{\widehat{\boldsymbol{\sigma}}_{k}=(\widehat{\sigma}_{k,1},\ldots,\widehat{\sigma}_{k,p})\}_{k=1}^{K}$; learning rate $\lambda>0$; dispersion weight $\gamma\in[0,1]$; prior weights $\{\pi_{k}\}_{k=1}^{K}$ with $\pi_{k}>0$ and $\sum_{k}\pi_{k}=1$ (default: $\pi_{k}=1/K$).
\begin{itemize}
\item {\bf Step 1: Data splitting.} Randomly split the data into equal halves $\mathcal{S}_{1}$ and $\mathcal{S}_{2}$.
\item {\bf Step 2: Model fitting.} For each $k\in\{1,\ldots,K\}$, fit the $k$-th candidate on $\mathcal{S}_{1}$ to obtain $\widehat{\mathbf{U}}^{(k)}$ and node-specific noise estimates $\widehat{\sigma}_{k,1},\ldots,\widehat{\sigma}_{k,p}$.
\item {\bf Step 3: $\gamma$-tempered heteroscedastic scoring.} For each $k$, compute the validation score
\[
L_{k}=\sum_{j=1}^{p}\left[\frac{1}{2\widehat{\sigma}_{k,j}^{2}}\sum_{i\in\mathcal{S}_{2}}\!\left(x_{ij}-\sum_{k'\neq j}\widehat{U}^{(k)}_{jk'}\,x_{ik'}\right)^{\!2} \;+\;\frac{\gamma\,|\mathcal{S}_{2}|}{2}\log\widehat{\sigma}_{k,j}^{2}\right],
\]
and the weight
\[
w^{k}=\frac{\pi_{k}\exp\!\left\{-\lambda\,L_{k}\right\}}{\sum_{k'=1}^{K}\pi_{k'}\exp\!\left\{-\lambda\,L_{k'}\right\}}.
\]
\item {\bf Step 4: Averaging across splits.} Repeat Steps~1--3 a total of $L$ times and set $w_{k}=\frac{1}{L}\sum_{l=1}^{L}w_{l}^{k}$.
\end{itemize}
\noindent\textbf{Output:} The weight vector $\mathbf{w}=(w_{1},\ldots,w_{K})$ and the aggregated adjacency matrix $\widetilde{\mathbf{U}}=\sum_{k=1}^{K}w_{k}\widehat{\mathbf{U}}^{(k)}$.
\caption{Heteroscedastic data-splitting procedure with $\gamma$-tempered dispersion penalty.\label{alg:ARM-weighting-method-1-hetero}}
\end{algorithm}

The dispersion weight $\gamma\in[0,1]$ controls how heavily the heteroscedastic log-variance term $\log\widehat{\sigma}_{k,j}^{2}$ contributes to the validation score. When $\gamma=1$, $\mathcal{S}_{k}$ is the exact negative Gaussian log-likelihood (up to additive constants), so the procedure reduces to standard likelihood-based weighting under a heteroscedastic model with node-specific variances $\sigma_{j}^{2}$. When $\gamma=0$, the $\log\widehat{\sigma}_{k,j}^{2}$ term drops out entirely and $\mathcal{S}_{k}$ becomes a pure $\sigma_{j}^{-2}$-weighted sum of squared prediction errors on the validation set, making the score robust to noisy or systematically biased variance estimates. Intermediate values $\gamma\in(0,1)$ interpolate between these regimes. In the homoscedastic case $\sigma_{k,1}=\cdots=\sigma_{k,p}=\widehat{\sigma}_{k}$, the dispersion term reduces to a global constant per candidate and Algorithm~\ref{alg:ARM-weighting-method-1-hetero} recovers the homoscedastic procedure of Section~\ref{sec:method} regardless of $\gamma$.

\bibliographystyle{apalike}
\bibliography{reference}

\newpage
\begin{center}
	{\large\bf Supplementary Material for ``Stable Causal Discovery via Directed Acyclic Graph Aggregation''}
\end{center}

 \setcounter{section}{0}
 \renewcommand{\thesection}{S\arabic{section}}  
 
 \setcounter{table}{0}
 \renewcommand{\thetable}{S\arabic{table}}  
 \setcounter{figure}{0}
 \renewcommand{\thefigure}{S\arabic{figure}}  

\section{Additional Simulation Studies}

Additional simulation results are presented in this section. 

The data generation procedure follows that of the main paper. We consider the same three adjacency-matrix structures studied therein: random, hub, and chain structures. For each adjacency-matrix structure, we examine four simulation settings obtained from combinations of signal strength and dimensionality. Specifically, we consider a \textbf{strong-signal} regime in which nonzero entries of $\bm{U}$ are independently generated from $\text{Unif}([-1.5,-0.5]\cup[0.5,1.5])$, and a \textbf{weak-signal} regime in which nonzero entries are generated from $\text{Unif}([-0.5,-0.1]\cup[0.1,0.5])$.
We further consider a \textbf{lower-dimensional} setting with $p=25$ and $n=100$, and a \textbf{higher-dimensional} setting with (p=100) and (n=250). Representative DAGs for the three structures, across four settings each, are displayed in Figures~\ref{fig:random_all}--\ref{fig:chain_all}.

We employ the same candidate DAG estimators as those used in the main paper. Performance is evaluated using seven metrics, reported as means with standard errors (in parentheses) over 200 replications: the average squared estimation error $|\widehat{\mathbf{U}}-\mathbf{U}|_{F}^{2}$, where $|\cdot|_F$ denotes the Frobenius norm; the false positive rate (FPR), false discovery rate (FDR), and false negative rate (FNR) for edge recovery; the Matthews correlation coefficient (MCC); the Kullback--Leibler divergence (KL); and the structural Hamming distance (SHD). We additionally report the average model averaging weights assigned to each candidate estimator by DAGgr.

\subsection{Additional Results for Random Adjacency Matrix Structures}
\begin{figure}[H]
\centering
\subfloat[Strong-signal (\(p=25\))]{ \label{fig:random_strong_p25}
	\includegraphics[width=0.23\textwidth]{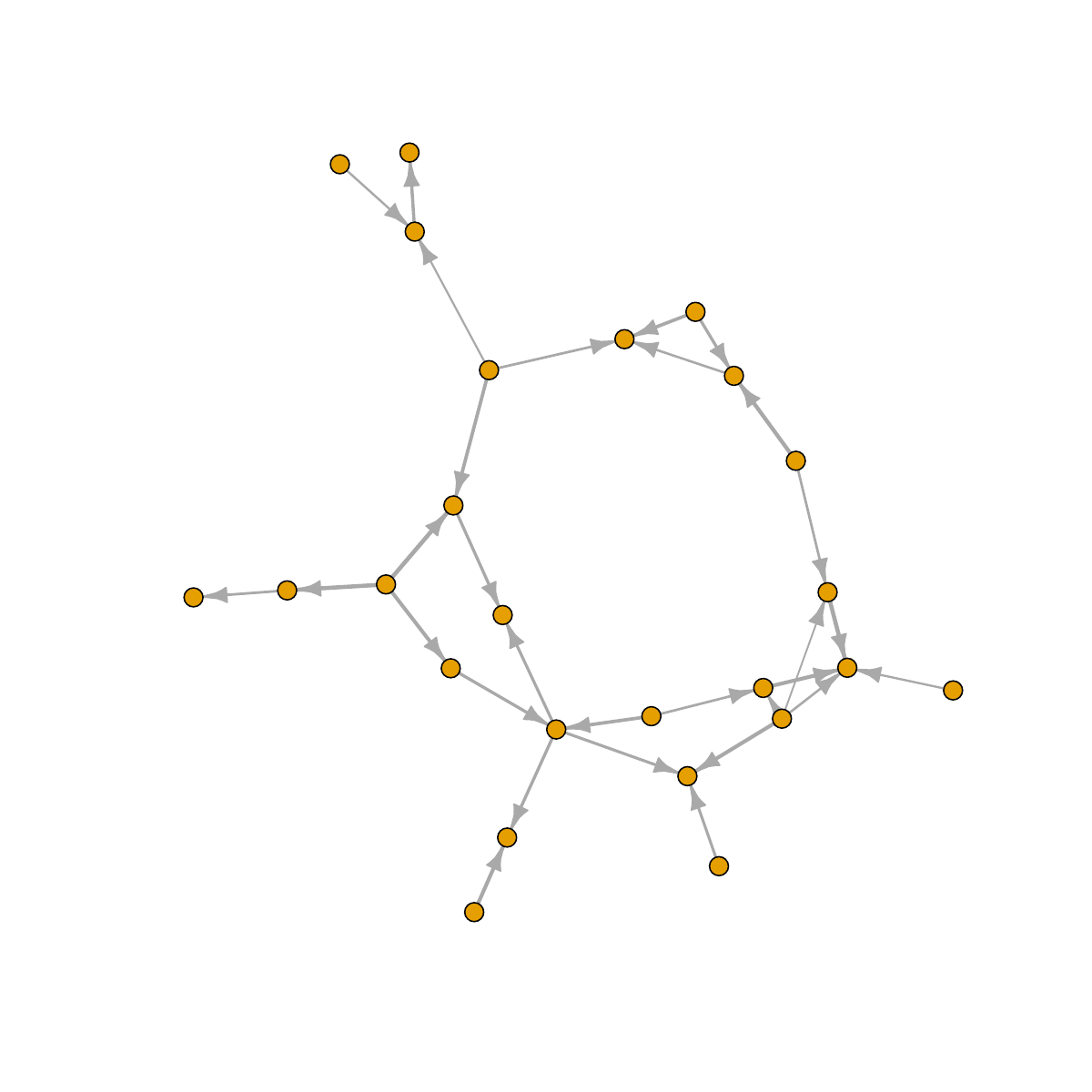}}\ \ 
\subfloat[Weak-signal (\(p=25\))]{ \label{fig:random_weak_p25}
	\includegraphics[width=0.23\textwidth]{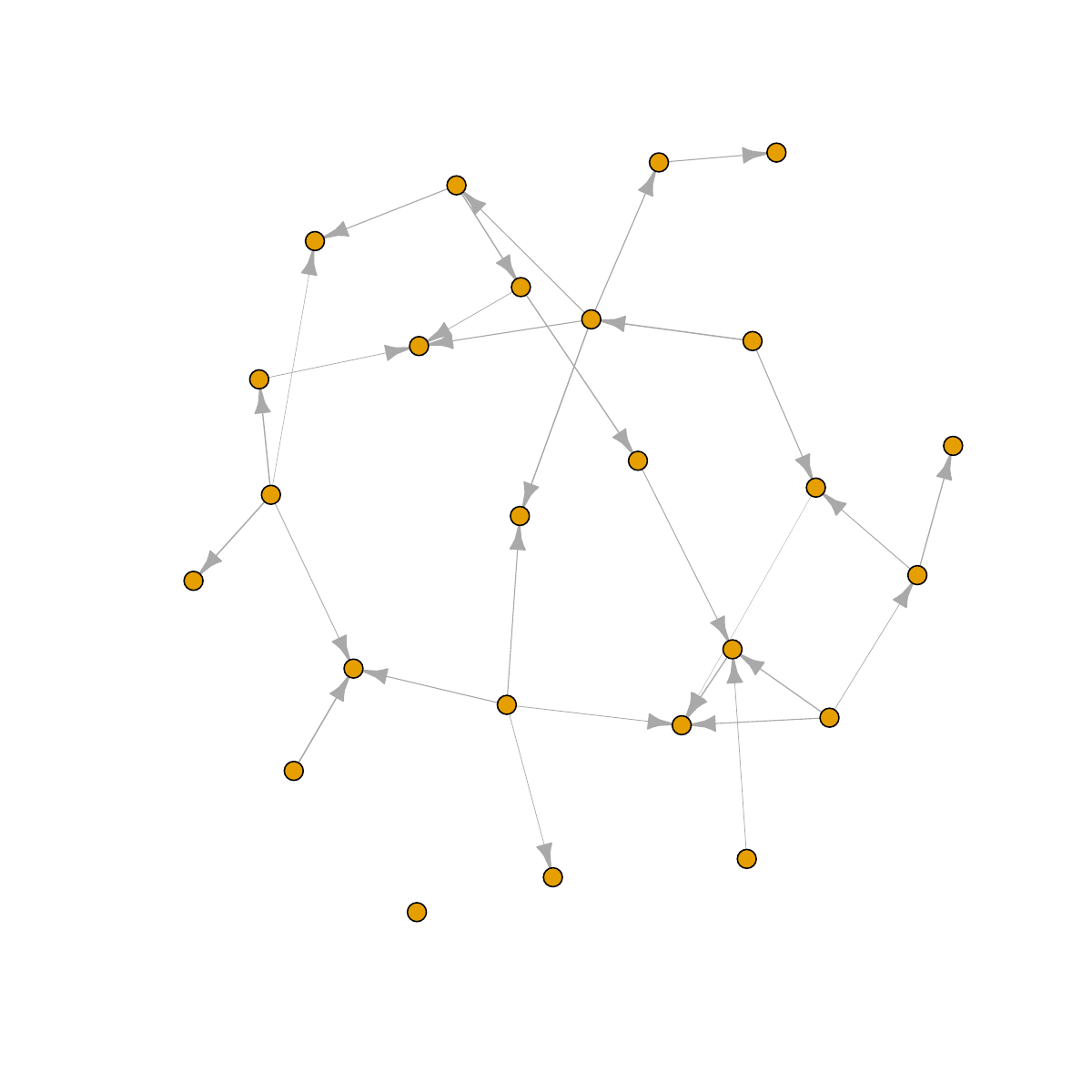}} \ \ 
\subfloat[Strong-signal (\(p=100\))]{ \label{fig:random_strong_p100}
	\includegraphics[width=0.23\textwidth]{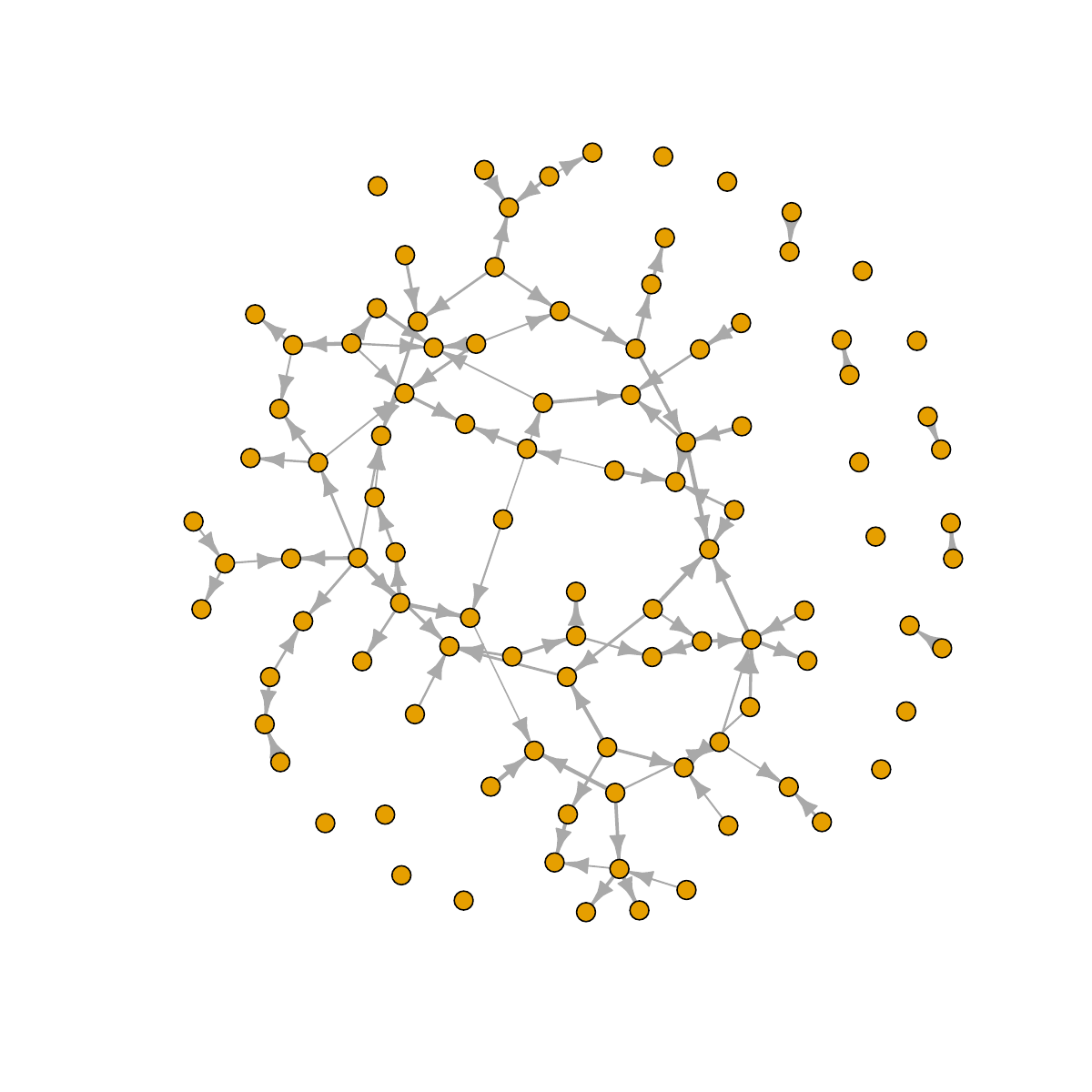}}\ \ 
\subfloat[Weak-signal (\(p=100\))]{ \label{fig:random_weak_p100}
	\includegraphics[width=0.23\textwidth]{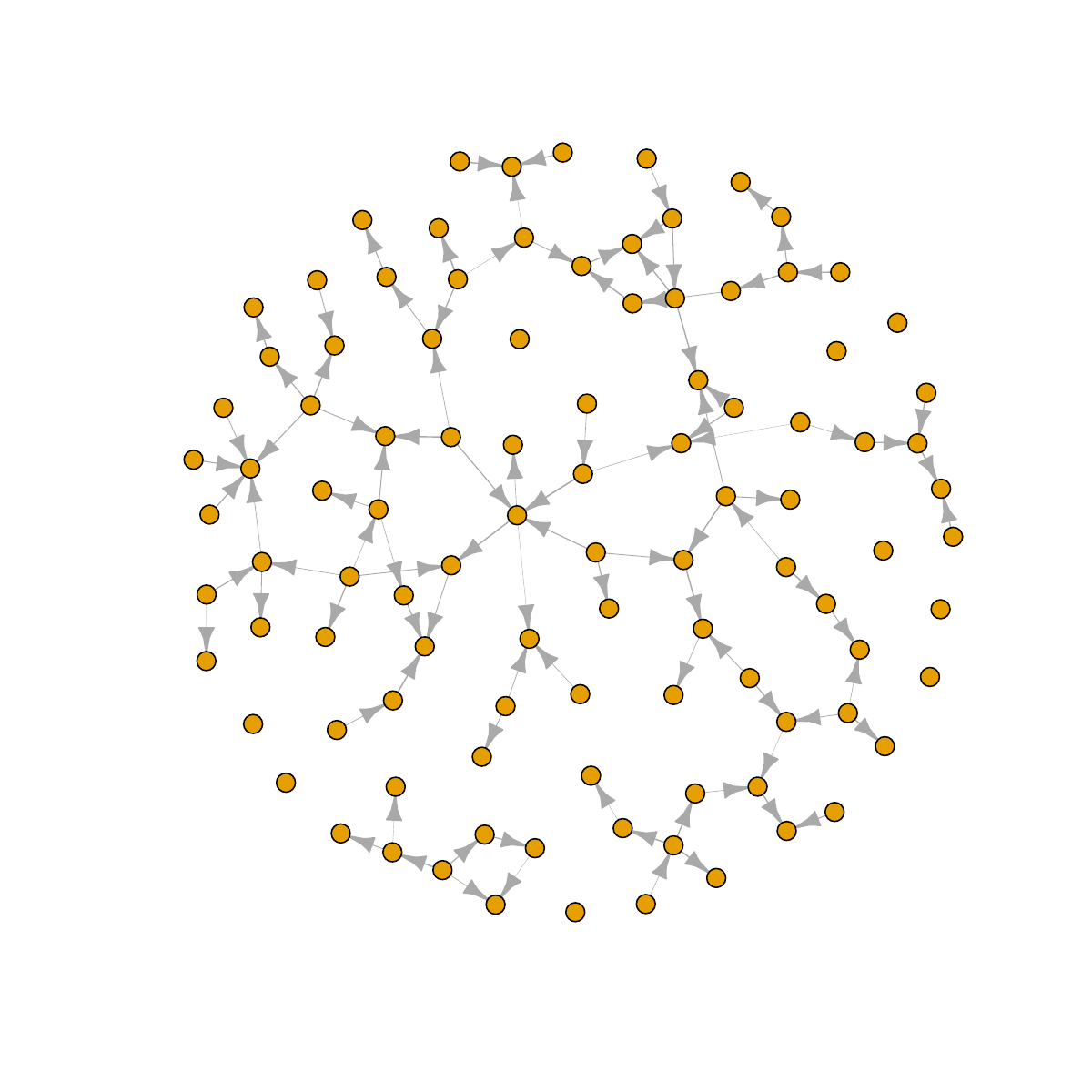}}
\caption{Random adjacency matrix structures used in the Table~\ref{tab:random_strong_p25}-Table~\ref{tab:random_weak_p100}.}
\label{fig:random_all}
\end{figure}

\begin{table}[H]
\centering
\caption{Performance comparison under the random adjacency matrix structure in the strong-signal setting with $p=25$ and $n=100$}
\resizebox{1.05\textwidth}{!}{	\begin{tabular}{ccccccccc}
		\hline  & Average weights  &Average squared error & FPR & FDR & FNR& MCC & KL & SHD\\\hline 
		DAGgr-raw & -  &\textbf{1.0666} (0.0451) & 6.89\% (0.38\%) & 32.19\% (1.78\%) & \textbf{0.02\%} (0.02\%) & 0.46 (0.00) & 33.56 (0.12) & 38.7 (2.1)\\ 
		DAGgr-pruned (0.5) & -  &\textbf{1.1410} (0.0504) & 0.23\% (0.02\%) & 3.93\% (0.25\%) & \textbf{1.98\%} (0.15\%) & 0.92 (0.00) & 33.34 (0.13) & 1.6 (0.1)\\ 
		DAGgr-pruned (0.8) & -  &1.4562 (0.0790) & 0.10\% (0.01\%) & 1.68\% (0.21\%) & 4.30\% (0.29\%) & 0.89 (0.00) & 33.03 (0.14) & 1.7 (0.1)\\ 
		DAGgr-pruned (1-1/p) & -  &2.0858 (0.1287) & \textbf{0.05\%} (0.01\%) & \textbf{0.91\%} (0.17\%) & 7.08\% (0.48\%) & 0.86 (0.00) & \textbf{32.45} (0.17) & 2.4 (0.1)\\ 
		DAGBag & -  &- & 0.08\% (0.01\%) & 1.48\% (0.15\%) & 3.72\% (0.19\%) & \textbf{0.94} (0.00) & - & \textbf{1.3} (0.1)\\
		\hline\hline
		GES & 0.03\%  &12.2422 (0.2846) & 3.35\% (0.06\%) & 35.39\% (0.61\%) & 24.97\% (0.52\%) & 0.54 (0.00) & 31.66 (0.20) & 20.4 (0.3)\\ 
		PC-algorithm & 0\%  &43.3443 (0.4480) & 2.17\% (0.04\%) & 42.64\% (0.82\%) & 57.93\% (0.75\%) & 0.40 (0.01) & \textbf{16.48} (0.33) & 18.3 (0.2)\\ 
		Hill climbing & 2.07\%  &6.7060 (0.2362) & 2.48\% (0.05\%) & 26.19\% (0.54\%) & 11.90\% (0.35\%) & 0.64 (0.00) & 36.14 (0.15) & 15.4 (0.3)\\ 
		Max-min hill climbing & 0\%  &8.6882 (0.1682) & 0.41\% (0.02\%) & 8.39\% (0.32\%) & 24.08\% (0.37\%) & 0.78 (0.00) & 31.95 (0.16) & 8.1 (0.1)\\ 
		Constrained-MLE & 3.24\%  &10.4168 (0.2967) & 11.52\% (0.06\%) & 62.52\% (0.34\%) & 10.35\% (0.31\%) & 0.44 (0.00) & 35.94 (0.12) & 65.8 (0.4)\\ 
		clrdag & 61.01\%  &\textbf{1.5213} (0.0724) & \textbf{0.28\%} (0.02\%) & \textbf{4.69\%} (0.25\%) & \textbf{3.02\%} (0.18\%) & \textbf{0.91} (0.00) & 35.23 (0.09) & \textbf{2.1} (0.1)\\ 
		NOTEARS-linear & 16.71\%  &\textbf{1.8790} (0.1027) & \textbf{0.09\%} (0.01\%) & \textbf{1.69\%} (0.20\%) & \textbf{3.13\%} (0.22\%) & \textbf{0.95} (0.00) & \textbf{29.26} (0.09) & \textbf{1.1} (0.1)\\ 
		DAGMA-linear & 16.93\%  &\textbf{1.2125} (0.0674) & \textbf{0.34\%} (0.02\%) & \textbf{5.27\%} (0.30\%) & \textbf{1.00\%} (0.13\%) & \textbf{0.89} (0.00) & 33.37 (0.10) & \textbf{2.1} (0.1)\\ 
		NOTIME & 0\%  &29.3907 (0.3662) & 3.98\% (0.06\%) & 47.22\% (0.73\%) & 55.77\% (0.67\%) & 0.32 (0.01) & \textbf{23.83} (0.32) & 31.4 (0.4)\\
		\hline
\end{tabular} }\label{tab:random_strong_p25}
\end{table}  

\begin{table}[H]
\centering
\caption{Performance comparison under the random adjacency matrix structure in the weak-signal setting with $p=25$ and $n=100$} 
\resizebox{1.05\textwidth}{!}{	\begin{tabular}{ccccccccc} 
		\hline  & Average weights  &Average squared error & FPR & FDR & FNR& MCC & KL & SHD\\\hline
		DAGgr-raw & -  &\textbf{1.6353} (0.0236) & 5.64\% (0.09\%) & 44.01\% (0.67\%) & \textbf{20.73\%} (0.47\%) & 0.47 (0.00) & 2.84 (0.04) & 33.8 (0.4)\\ 
		DAGgr-pruned (0.5) & -  &\textbf{1.8858} (0.0289) & 0.87\% (0.05\%) & 11.87\% (0.69\%) & \textbf{47.48\%} (0.70\%) & 0.43 (0.01) & 2.20 (0.04) & 17.0 (0.2)\\ 
		DAGgr-pruned (0.8) & -  &2.4253 (0.0320) & 0.19\% (0.03\%) & 2.78\% (0.40\%) & 67.10\% (0.86\%) & 0.28 (0.01) & 1.68 (0.04) & 20.8 (0.2)\\ 
		DAGgr-pruned (1-1/p) & -  &2.7867 (0.0323) & \textbf{0.05\%} (0.01\%) & \textbf{1.48\%} (0.24\%) & 76.88\% (0.79\%) & 0.28 (0.01) & \textbf{1.28} (0.04) & 23.3 (0.2)\\ 
		DAGBag & -  &- & 0.31\% (0.01\%) & 8.14\% (0.38\%) & 52.75\% (0.54\%) & \textbf{0.54} (0.01) & - & \textbf{16.4} (0.2)\\
		\hline\hline
		GES & 8.59\%  &3.4652 (0.0577) & 3.62\% (0.05\%) & 43.90\% (0.66\%) & 48.17\% (0.76\%) & 0.38 (0.01) & 4.42 (0.04) & 27.0 (0.4)\\ 
		PC-algorithm & 27.11\%  &\textbf{2.2203} (0.0424) & 1.07\% (0.03\%) & 20.30\% (0.49\%) & \textbf{42.80\%} (0.61\%) & \textbf{0.56} (0.01) & \textbf{2.46} (0.02) & \textbf{16.5} (0.2)\\ 
		Hill climbing & 14.07\%  &2.5146 (0.0414) & 2.96\% (0.05\%) & 37.53\% (0.59\%) & \textbf{36.87\%} (0.64\%) & \textbf{0.49} (0.01) & 4.57 (0.03) & 23.0 (0.3)\\ 
		Max-min hill climbing & 16.2\%  &\textbf{2.3361} (0.0343) & 1.36\% (0.03\%) & 25.88\% (0.56\%) & \textbf{47.95\%} (0.54\%) & \textbf{0.50} (0.01) & 3.61 (0.03) & \textbf{18.1} (0.2)\\ 
		Constrained-MLE & 4.82\%  &2.5286 (0.0406) & \textbf{0.64\%} (0.02\%) & \textbf{15.91\%} (0.58\%) & 59.70\% (0.63\%) & 0.44 (0.01) & 2.97 (0.03) & 19.0 (0.2)\\ 
		clrdag & 0.14\%  &3.8062 (0.0156) & \textbf{0.04\%} (0.01\%) & \textbf{4.20\%} (0.61\%) & 97.42\% (0.20\%) & 0.06 (0.00) & \textbf{0.39} (0.03) & 29.2 (0.1)\\ 
		NOTEARS-linear & 8.75\%  &2.5386 (0.0312) & \textbf{0.23\%} (0.01\%) & \textbf{8.83\%} (0.50\%) & 70.03\% (0.51\%) & 0.41 (0.01) & \textbf{1.64} (0.02) & 21.1 (0.2)\\ 
		DAGMA-linear & 18.61\%  &\textbf{2.4773} (0.0458) & 0.69\% (0.02\%) & 16.48\% (0.57\%) & 56.53\% (0.65\%) & 0.47 (0.01) & 3.05 (0.03) & \textbf{18.2} (0.2)\\ 
		NOTIME & 1.7\%  &2.7582 (0.0412) & 0.85\% (0.03\%) & 20.21\% (0.62\%) & 60.47\% (0.62\%) & 0.43 (0.01) & 2.96 (0.04) & 19.6 (0.2)\\
		\hline
\end{tabular} }\label{tab:random_weak_p25}
\end{table}

\begin{table}[H]
\centering
\caption{Performance comparison under the random adjacency matrix structure in the strong-signal setting with $p=100$ and $n=250$}
\resizebox{1.05\textwidth}{!}{	\begin{tabular}{ccccccccc}
		\hline  &  Average weights &Average squared error & FPR & FDR & FNR& MCC & KL & SHD\\\hline
		DAGgr-raw &  - &\textbf{0.9955} (0.0427) & 0.01\% (0.00\%) & 0.65\% (0.08\%) & \textbf{0.35\%} (0.05\%) & 0.97 (0.00) & 99.36 (0.11) & 0.8 (0.1)\\
		DAGgr-pruned (0.5) &  - &1.0006 (0.0432) & 0.01\% (0.00\%) & 0.59\% (0.08\%) & \textbf{0.47\%} (0.05\%) & \textbf{0.98} (0.00) & 99.35 (0.11) & \textbf{0.7} (0.1)\\
		DAGgr-pruned (0.8) &  - &1.0003 (0.0432) & 0.01\% (0.00\%) & 0.59\% (0.08\%) & \textbf{0.47\%} (0.05\%) & \textbf{0.98} (0.00) & 99.35 (0.11) & \textbf{0.7} (0.1)\\
		DAGgr-pruned (1-1/p) &  - &\textbf{0.9999} (0.0426) & \textbf{0.00\%} (0.00\%) & \textbf{0.46\%} (0.07\%) & 0.54\% (0.06\%) & \textbf{0.98} (0.00) & \textbf{99.32} (0.12) & \textbf{0.7} (0.1)\\
		DAGBag &  - &- & 0.02\% (0.00\%) & 1.63\% (0.09\%) & 0.62\% (0.05\%) & 0.96 (0.00) & -& 2.0 (0.1)\\
		\hline\hline
		GES &  0\% &23.5622 (0.3427) & 1.31\% (0.01\%) & 54.38\% (0.31\%) & 13.62\% (0.19\%) & 0.54 (0.00) & 96.10 (0.23) & 129.0 (0.7)\\
		PC-algorithm &  0\% &128.8640 (0.7159) & 0.62\% (0.01\%) & 51.13\% (0.43\%) & 54.82\% (0.41\%) & 0.40 (0.00) & \textbf{49.55} (0.52) & 69.0 (0.5)\\
		Hill climbing &  0\% &\textbf{6.6274} (0.1991) & 1.19\% (0.01\%) & 48.13\% (0.30\%) & \textbf{2.11\%} (0.10\%) & 0.61 (0.00) & 107.31 (0.12) & 117.0 (0.7)\\
		Max-min hill climbing &  0\% &10.5572 (0.2049) & 0.15\% (0.00\%) & 13.00\% (0.24\%) & 9.97\% (0.17\%) & 0.82 (0.00) & 97.72 (0.19) & 23.7 (0.3)\\
		Constrained-MLE &  0\% &32.5609 (0.4812) & 5.15\% (0.01\%) & 79.07\% (0.23\%) & 14.24\% (0.19\%) & 0.32 (0.00) & 104.40 (0.18) & 505.3 (1.5)\\
		clrdag &  0\% &13.4153 (0.1272) & \textbf{0.12\%} (0.00\%) & \textbf{11.40\%} (0.15\%) & 13.41\% (0.15\%) & \textbf{0.85} (0.00) & 102.19 (0.14) & \textbf{21.3} (0.2)\\
		NOTEARS-linear &  2.17\% &\textbf{2.0380} (0.0325) & \textbf{0.00\%} (0.00\%) & \textbf{0.41\%} (0.04\%) & \textbf{1.14\%} (0.07\%) & \textbf{0.98} (0.00) & \textbf{87.95} (0.07) & \textbf{1.1} (0.1)\\
		DAGMA-linear &  97.83\% &\textbf{1.0005} (0.0444) & \textbf{0.01\%} (0.00\%) & \textbf{0.59\%} (0.08\%) & \textbf{0.47\%} (0.05\%) & \textbf{0.98} (0.00) & 99.62 (0.09) & \textbf{0.7} (0.1)\\
		NOTIME &  0\% &83.1511 (0.5445) & 0.85\% (0.01\%) & 59.02\% (0.39\%) & 56.29\% (0.35\%) & 0.35 (0.00) & \textbf{66.33} (0.40) & 106.2 (0.6)\\\hline
\end{tabular} }\label{tab:random_strong_p100}
\end{table}

\begin{table}[H]
\centering
\caption{Performance comparison under the random adjacency matrix structure in the weak-signal setting with $p=100$ and $n=250$}
\resizebox{1.05\textwidth}{!}{	\begin{tabular}{ccccccccc}
		\hline  &  Average weights &Average squared error & FPR & FDR & FNR& MCC & KL & SHD\\\hline
		DAGgr-raw &  - &\textbf{4.8229} (0.0883) & 1.52\% (0.02\%) & 50.06\% (0.66\%) & \textbf{15.70\%} (0.37\%) & 0.48 (0.00) & 10.55 (0.08) & 139.0 (1.3)\\
		DAGgr-pruned (0.5) &  - &\textbf{5.4875} (0.0953) & 0.75\% (0.02\%) & 42.05\% (1.06\%) & \textbf{37.00\%} (0.51\%) & 0.47 (0.00) & 8.86 (0.11) & 89.7 (1.4)\\
		DAGgr-pruned (0.8) &  - &5.9417 (0.0788) & 0.29\% (0.02\%) & 16.82\% (1.26\%) & 52.31\% (0.94\%) & 0.36 (0.01) & 7.26 (0.16) & 72.3 (1.0)\\
		DAGgr-pruned (1-1/p) &  - &6.1971 (0.0733) & \textbf{0.14\%} (0.02\%) & \textbf{8.58\%} (0.95\%) & 60.25\% (0.94\%) & 0.31 (0.01) & \textbf{6.14} (0.15) & 69.2 (0.7)\\
		DAGBag &  - &- & 0.10\% (0.00\%) & 11.75\% (0.27\%) & 39.31\% (0.32\%) & \textbf{0.65} (0.00) & -& \textbf{45.4} (0.3)\\\hline\hline
		GES &  0\% &9.6403 (0.0775) & 1.50\% (0.01\%) & 62.23\% (0.35\%) & 37.34\% (0.35\%) & 0.40 (0.00) & 14.29 (0.04) & 158.1 (0.9)\\
		PC-algorithm &  9.17\% &6.7732 (0.2751) & 0.77\% (0.01\%) & 47.84\% (0.78\%) & \textbf{30.08\%} (1.20\%) & \textbf{0.55} (0.01) & 8.90 (0.06) & 89.5 (1.2)\\
		Hill climbing &  0\% &\textbf{6.4061} (0.0609) & 1.36\% (0.01\%) & 55.18\% (0.34\%) & \textbf{24.03\%} (0.30\%) & 0.48 (0.00) & 14.72 (0.04) & 143.9 (0.9)\\
		Max-min hill climbing &  19.73\% &\textbf{4.9592} (0.0650) & 0.68\% (0.01\%) & 42.70\% (0.34\%) & \textbf{28.18\%} (0.32\%) & \textbf{0.57} (0.00) & 12.51 (0.03) & 80.5 (0.6)\\
		Constrained-MLE &  58.72\% &7.8422 (0.0703) & 0.92\% (0.01\%) & 52.07\% (0.39\%) & 39.43\% (0.34\%) & 0.45 (0.00) & 11.73 (0.03) & 103.9 (0.7)\\
		clrdag &  0\% &10.9579 (0.0112) & \textbf{0.00\%} (0.00\%) & \textbf{1.25\%} (0.39\%) & 99.06\% (0.06\%) & 0.05 (0.00) & \textbf{0.38} (0.03) & 97.1 (0.1)\\
		NOTEARS-linear &  0\% &6.4838 (0.0407) & \textbf{0.02\%} (0.00\%) & \textbf{4.41\%} (0.21\%) & 71.93\% (0.23\%) & 0.44 (0.00) & \textbf{4.05} (0.03) & \textbf{70.5} (0.2)\\
		DAGMA-linear &  12.38\% &\textbf{5.1858} (0.0532) & \textbf{0.06\%} (0.00\%) & \textbf{10.07\%} (0.29\%) & 58.11\% (0.26\%) & \textbf{0.55} (0.00) & \textbf{7.55} (0.04) & \textbf{57.0} (0.3)\\
		NOTIME &  0\% &6.5081 (0.0780) & 0.10\% (0.00\%) & 16.61\% (0.41\%) & 61.40\% (0.36\%) & 0.49 (0.00) & 7.78 (0.05) & \textbf{60.3} (0.4)\\\hline
\end{tabular} }\label{tab:random_weak_p100}
\end{table}

\subsection{Additional Results for Hub Adjacency Matrix Structures}
\begin{figure}[H]
\centering
\subfloat[Strong-signal (\(p=25\))]{ \label{fig:hub_strong_p25}
	\includegraphics[width=0.23\textwidth]{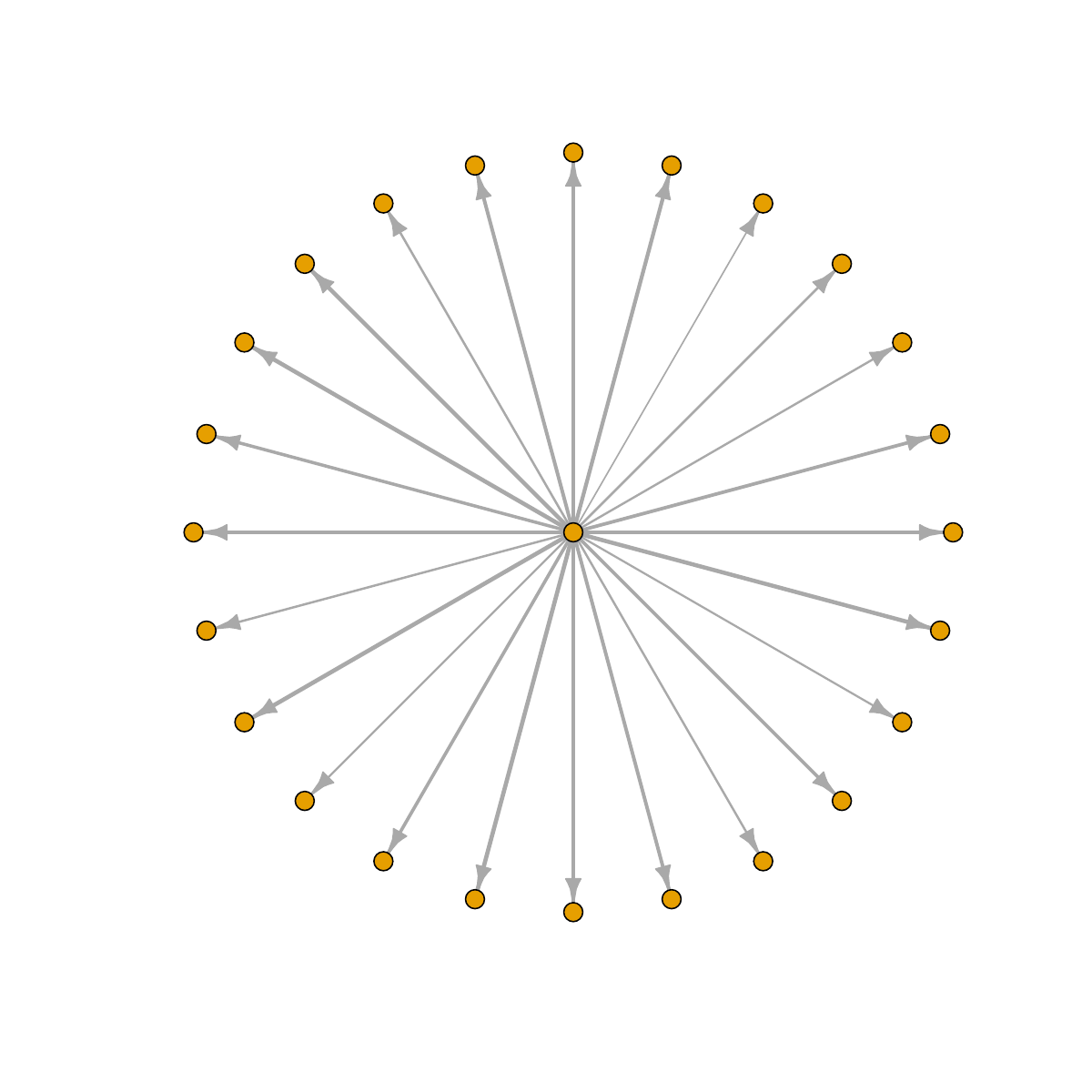}}\ \ 
\subfloat[Weak-signal (\(p=25\))]{ \label{fig:hub_weak_p25}
	\includegraphics[width=0.23\textwidth]{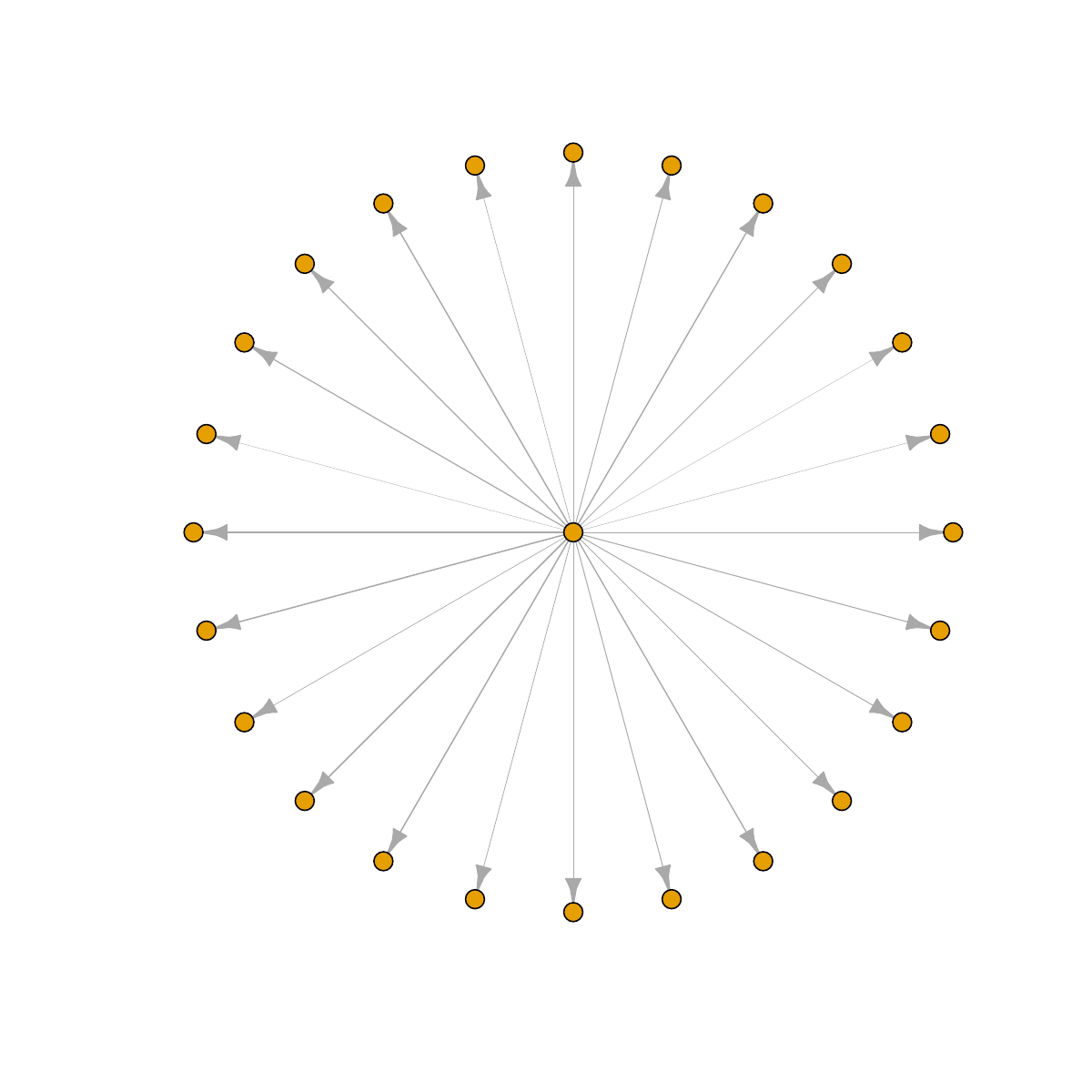}} \ \ 
\subfloat[Strong-signal (\(p=100\))]{ \label{fig:hub_strong_p100}
	\includegraphics[width=0.23\textwidth]{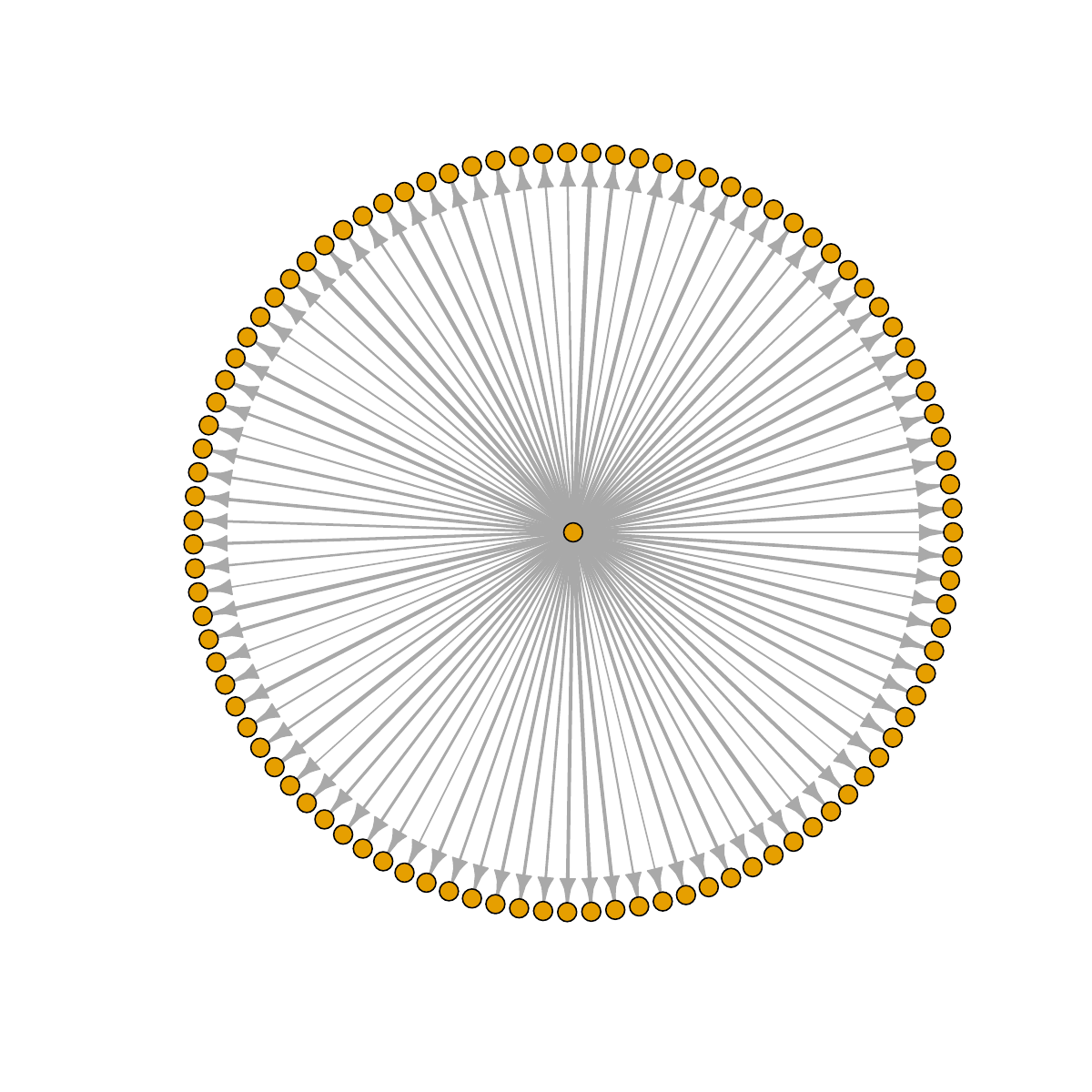}}\ \ 
\subfloat[Weak-signal (\(p=100\))]{ \label{fig:hub_weak_p100}
	\includegraphics[width=0.23\textwidth]{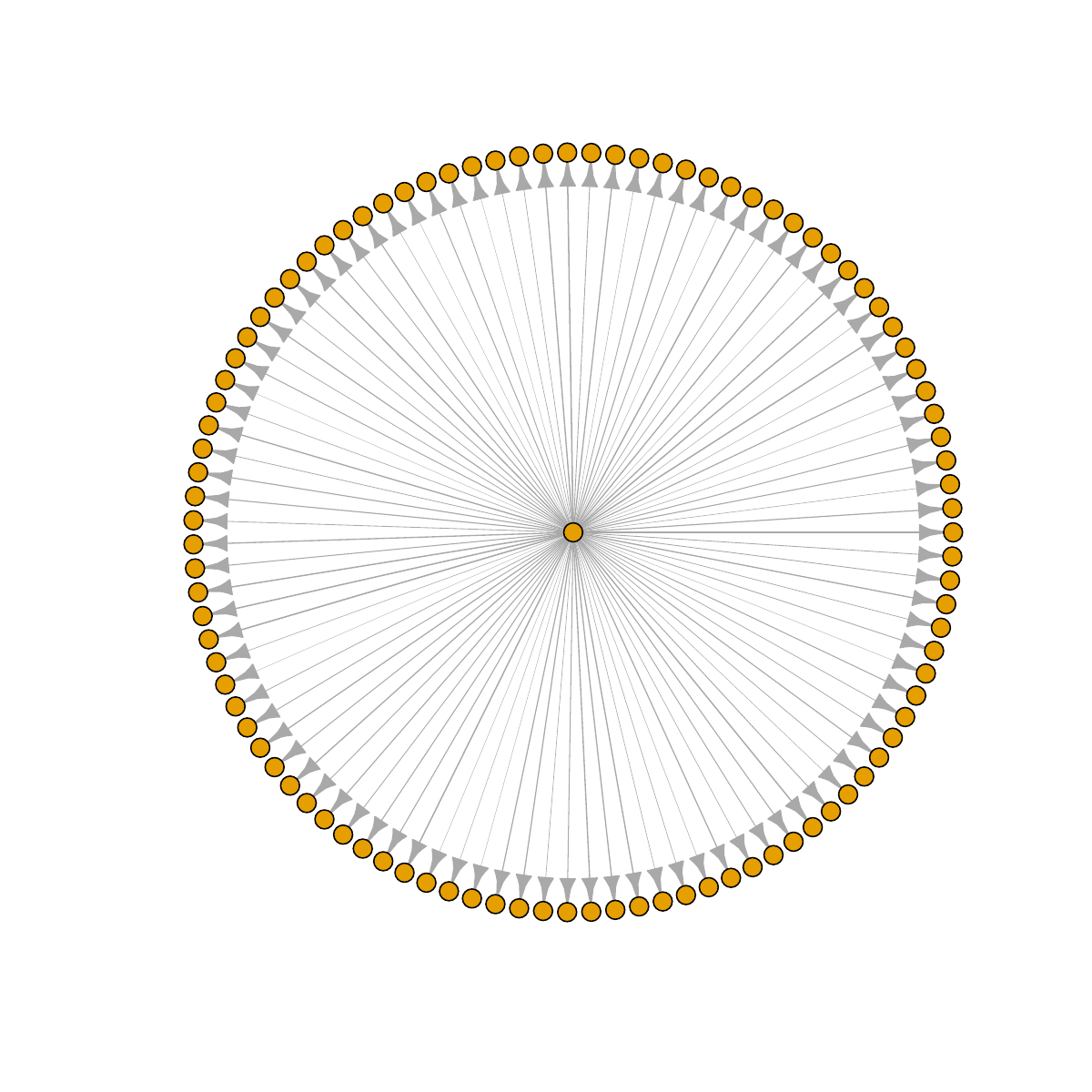}}
\caption{Hub adjacency matrix structures used in the Table~\ref{tab:hub_strong_p25}-Table~\ref{tab:hub_weak_p100}.}
\label{fig:hub_all}
\end{figure}

\begin{table}[H]
\centering
\caption{Performance comparison under the hub adjacency matrix structure in the strong-signal setting with $p=25$ and $n=100$}
\resizebox{1.05\textwidth}{!}{	\begin{tabular}{ccccccccc}
		\hline&  Average weights &Average squared error & FPR & FDR & FNR & MCC & KL & SHD\\\hline
		DAGgr-raw &  - &\textbf{0.5617} (0.0268) & 2.77\% (0.13\%) & 19.70\% (0.96\%) & \textbf{0.02\%} (0.02\%) & 0.53 (0.00) & 29.32 (0.15) & 14.3 (0.7)\\ 
		DAGgr-pruned (0.5) &  - &\textbf{0.5700} (0.0300) & 0.05\% (0.01\%) & 1.13\% (0.17\%) & \textbf{0.65\%} (0.15\%) & \textbf{0.93} (0.00) & 29.29 (0.15) & \textbf{0.5} (0.1)\\ 
		DAGgr-pruned (0.8) &  - &1.0227 (0.0971) & 0.01\% (0.00\%) & 0.20\% (0.06\%) & 4.04\% (0.47\%) & \textbf{0.93} (0.00) & 29.07 (0.17) & 1.0 (0.1)\\ 
		DAGgr-pruned (1-1/p) &  -&2.6224 (0.3667) & \textbf{0.00\%} (0.00\%) & \textbf{0.06\%} (0.03\%) & 11.00\% (1.22\%) & 0.85 (0.01) & \textbf{27.88} (0.38) & 2.7 (0.3)\\ 
		DAGBag &  - &- & 0.11\% (0.01\%) & 2.37\% (0.23\%) & 10.04\% (0.73\%) & 0.84 (0.01) &- & 3.0 (0.2)\\
		\hline\hline
		GES &  4.41\% &1.7833 (0.2313) & 2.19\% (0.06\%) & 25.77\% (0.67\%) & 2.02\% (0.70\%) & 0.67 (0.01) & 32.14 (0.17) & 12.9 (0.4)\\  
		PC-algorithm &  0\% &36.4323 (0.2565) & 1.60\% (0.04\%) & 44.00\% (1.12\%) & 82.35\% (0.45\%) & 0.17 (0.01) & \textbf{5.81} (0.25) & 27.2 (0.3)\\ 
		Hill climbing &  0.61\% &26.1085 (0.8788) & 4.42\% (0.11\%) & 45.50\% (1.10\%) & 79.12\% (2.79\%) & 0.11 (0.02) & 8.32 (0.90) & 36.5 (0.9)\\ 
		Max-min hill climbing &  0\% &31.0361 (0.1692) & 2.27\% (0.04\%) & 65.30\% (1.08\%) & 97.83\% (0.47\%) & 0.00 (0.01) & \textbf{1.98} (0.32) & 30.6 (0.2)\\ 
		Constrained-MLE &  39.36\% &\textbf{0.4168} (0.0190) & \textbf{0.03\%} (0.00\%) & \textbf{0.67\%} (0.11\%) & \textbf{0.75\%} (0.13\%) & \textbf{0.95} (0.00) & 31.48 (0.10) & \textbf{0.4} (0.0)\\ 
		clrdag &  16.3\% &2.5788 (0.1576) & \textbf{0.01\%} (0.00\%) & \textbf{0.12\%} (0.05\%) & 10.56\% (0.62\%) & 0.87 (0.01) & 28.75 (0.24) & 2.6 (0.1)\\ 
		NOTEARS-linear &  10.61\% &\textbf{1.0112} (0.0302) & \textbf{0.01\%} (0.00\%) & \textbf{0.26\%} (0.07\%) & \textbf{1.25\%} (0.14\%) & \textbf{0.97} (0.00) & 24.87 (0.11) & \textbf{0.4} (0.0)\\ 
		DAGMA-linear &  28.7\% &\textbf{0.5258} (0.0147) & 0.20\% (0.01\%) & 3.90\% (0.27\%) & \textbf{0.33\%} (0.08\%) & \textbf{0.90} (0.00) & 29.54 (0.10) & \textbf{1.2} (0.1)\\
		NOTIME &  0\% &33.8061 (0.2334) & 1.76\% (0.13\%) & 23.55\% (1.71\%) & 100.00\% (0.00\%) & -0.01 (0.00) & \textbf{2.55} (0.23) & 34.1 (0.7)\\
		\hline
\end{tabular} }\label{tab:hub_strong_p25}
\end{table}

\begin{table}[H]
\centering
\caption{Performance comparison under the hub adjacency matrix structure in the weak-signal setting with $p=25$ and $n=100$}
\resizebox{1.05\textwidth}{!}{	\begin{tabular}{ccccccccc} 
		\hline  &  Average weights &Average squared error & FPR & FDR & FNR& MCC & KL & SHD\\\hline 
		DAGgr-raw &  - &\textbf{1.1903} (0.0192) & 3.03\% (0.11\%) & 30.57\% (1.14\%) & \textbf{40.75\%} (0.81\%) & 0.36 (0.01) & 1.54 (0.03) & 24.8 (0.6)\\ 
		DAGgr-pruned (0.5) &  - &\textbf{1.3593} (0.0247) & 0.08\% (0.01\%) & 2.42\% (0.34\%) & \textbf{63.06\%} (0.82\%) & \textbf{0.42} (0.01) & 1.47 (0.04) & \textbf{15.6} (0.2)\\ 
		DAGgr-pruned (0.8) &  - &1.6688 (0.0315) & 0.02\% (0.01\%) & 0.78\% (0.16\%) & 74.25\% (0.91\%) & 0.29 (0.01) & 1.27 (0.04) & 18.0 (0.2)\\ 
		DAGgr-pruned (1-1/p) &  - &1.9814 (0.0346) & \textbf{0.01\%} (0.00\%) & \textbf{0.25\%} (0.11\%) & 83.40\% (0.89\%) & 0.21 (0.01) & \textbf{0.94} (0.05) & 20.0 (0.2)\\ 
		DAGBag &  - &- & 0.17\% (0.01\%) & 8.00\% (0.70\%) & 79.88\% (0.46\%) & 0.28 (0.01) &- & 19.8 (0.1)\\
		\hline\hline
		GES &  5.17\% &2.2455 (0.0509) & 2.45\% (0.05\%) & 33.55\% (0.75\%) & \textbf{54.19\%} (1.10\%) & 0.32 (0.01) & 3.43 (0.04) & 24.5 (0.4)\\ 
		PC-algorithm &  0.48\% &3.2992 (0.0392) & 1.58\% (0.04\%) & 45.55\% (1.06\%) & 85.60\% (0.37\%) & 0.14 (0.00) & 1.15 (0.03) & 28.0 (0.2)\\ 
		Hill climbing &  0.44\% &3.5296 (0.0492) & 3.61\% (0.07\%) & 54.70\% (1.06\%) & 91.40\% (1.29\%) & 0.04 (0.01) & 1.63 (0.07) & 35.7 (0.5)\\ 
		Max-min hill climbing &  0.62\% &3.2176 (0.0178) & 2.10\% (0.03\%) & 67.19\% (1.03\%) & 98.67\% (0.34\%) & -0.01 (0.00) & \textbf{0.75} (0.03) & 30.7 (0.2)\\ 
		Constrained-MLE &  20.16\% &\textbf{1.3679} (0.0245) & \textbf{0.03\%} (0.01\%) & \textbf{1.15\%} (0.21\%) & \textbf{59.79\%} (0.74\%) & \textbf{0.47} (0.01) & 1.95 (0.03) & \textbf{14.5} (0.2)\\ 
		clrdag &  6.57\% &2.2915 (0.0194) & \textbf{0.00\%} (0.00\%) & \textbf{0.00\%} (0.00\%) & 89.25\% (0.50\%) & 0.17 (0.01) & \textbf{0.93} (0.04) & 21.4 (0.1)\\ 
		NOTEARS-linear &  24.29\% &\textbf{1.5522} (0.0196) & \textbf{0.01\%} (0.00\%) & \textbf{0.54\%} (0.15\%) & 72.40\% (0.51\%) & \textbf{0.38} (0.01) & 1.19 (0.02) & \textbf{17.4} (0.1)\\ 
		DAGMA-linear &  40.46\% &\textbf{1.1987} (0.0177) & 0.17\% (0.01\%) & 5.50\% (0.41\%) & \textbf{56.98\%} (0.59\%) & \textbf{0.49} (0.01) & 2.08 (0.03) & \textbf{14.7} (0.2)\\ 
		NOTIME &  1.81\% &2.7654 (0.0159) & 0.27\% (0.02\%) & 17.00\% (1.09\%) & 99.23\% (0.29\%) & 0.01 (0.00) & \textbf{0.22} (0.02) & 25.3 (0.1)\\
		\hline
\end{tabular} }\label{tab:hub_weak_p25}
\end{table}

\begin{table}[H]
\centering
\caption{Performance comparison under the hub adjacency matrix structure in the strong-signal setting with $p=100$ and $n=250$}
\resizebox{1.05\textwidth}{!}{	\begin{tabular}{ccccccccc}
		\hline  
		&  Average weights &Average squared error & FPR & FDR & FNR & MCC & KL & SHD\\
		\hline
		DAGgr-raw &  - &\textbf{1.2289} (0.0652) & 0.30\% (0.01\%) & 18.98\% (0.64\%) & \textbf{0.04\%} (0.02\%) & 0.80 (0.00) & 105.55 (0.22) & 29.1 (1.0)\\ 
		DAGgr-pruned (0.5) &  - &\textbf{1.3667} (0.1220) & 0.02\% (0.00\%) & 1.16\% (0.36\%) & \textbf{0.81\%} (0.18\%) & 0.85 (0.00) & 105.45 (0.23) & 2.4 (0.5)\\ 
		DAGgr-pruned (0.8) &  - &1.8086 (0.1470) & \textbf{0.00\%} (0.00\%) & 0.13\% (0.13\%) & 2.11\% (0.23\%) & 0.84 (0.00) & 105.24 (0.24) & \textbf{2.3} (0.3)\\ 
		DAGgr-pruned (1-1/p) &  - &1.8698 (0.1463) & \textbf{0.00\%} (0.00\%) & \textbf{0.00\%} (0.00\%) & 2.30\% (0.23\%) & \textbf{0.98} (0.00) & \textbf{105.20} (0.23) & \textbf{2.3} (0.2)\\ 
		DAGBag &  - &- & 0.09\% (0.00\%) & 7.13\% (0.20\%) & 23.49\% (1.71\%) & 0.69 (0.02) & -& 31.5 (1.7)\\
		\hline\hline
		GES &  0\% &44.9597 (2.7757) & 1.30\% (0.01\%) & 48.67\% (0.36\%) & 34.65\% (2.34\%) & 0.39 (0.01) & 103.11 (1.12) & 156.8 (2.5)\\ 
		PC-algorithm &  0\% &146.3620 (0.5386) & 1.24\% (0.01\%) & 84.15\% (0.40\%) & 98.41\% (0.10\%) & 0.00 (0.00) & \textbf{36.57} (0.54) & 210.7 (0.6)\\ 
		Hill climbing &  0\% &63.6739 (3.0546) & 1.31\% (0.01\%) & 52.47\% (0.45\%) & 51.45\% (2.63\%) & 0.30 (0.02) & 85.84 (2.64) & 170.5 (2.9)\\ 
		Max-min hill climbing &  0\% &111.3782 (0.2152) & 1.02\% (0.01\%) & 80.53\% (0.50\%) & 96.19\% (0.26\%) & 0.02 (0.00) & \textbf{20.89} (1.24) & 188.1 (0.6)\\ 
		Constrained-MLE &  17.18\% &\textbf{2.0259} (0.0218) & 0.35\% (0.00\%) & 22.73\% (0.30\%) & \textbf{0.02\%} (0.01\%) & 0.80 (0.00) & 112.12 (0.13) & 34.8 (0.5)\\ 
		clrdag &  36.88\% &\textbf{2.2002} (0.1899) & \textbf{0.00\%} (0.00\%) & \textbf{0.00\%} (0.00\%) & \textbf{2.28\%} (0.23\%) & \textbf{0.98} (0.00) & 109.30 (0.24) & \textbf{2.3} (0.2)\\ 
		NOTEARS-linear &  0\% &6.0077 (0.1025) & \textbf{0.00\%} (0.00\%) & \textbf{0.00\%} (0.00\%) & 4.46\% (0.15\%) & \textbf{0.95} (0.00) & 81.95 (0.22) & \textbf{4.4} (0.1)\\ 
		DAGMA-linear &  45.94\% &\textbf{1.0938} (0.0183) & \textbf{0.00\%} (0.00\%) & \textbf{0.02\%} (0.01\%) & \textbf{0.21\%} (0.03\%) & \textbf{0.99} (0.00) & 101.20 (0.14) & \textbf{0.2} (0.0)\\ 
		NOTIME &  0\% &109.8656 (0.0009) & \textbf{0.00\%} (0.00\%) & 1.50\% (0.86\%) & 100.00\% (0.00\%) & 0.00 (0.00) & \textbf{0.00} (0.00) & 99.0 (0.0)\\\hline 
\end{tabular} }\label{tab:hub_strong_p100}
\end{table}

\begin{table}[H]
\centering
\caption{Performance comparison under the hub adjacency matrix structure in the weak-signal setting with $p=100$ and $n=250$}
\resizebox{1.05\textwidth}{!}{		\begin{tabular}{ccccccccc} 
		\hline  &  Average weights &Average squared error & FPR & FDR & FNR& MCC & KL & SHD\\\hline
		DAGgr-raw &  - &\textbf{3.1550} (0.0375) & 0.37\% (0.00\%) & 27.78\% (0.34\%) & \textbf{25.74\%} (0.32\%) & \textbf{0.64} (0.00) & 8.16 (0.11) & 61.6 (0.6)\\ 
		DAGgr-pruned (0.5) &  - &\textbf{3.5747} (0.0613) & 0.23\% (0.01\%) & 17.54\% (0.98\%) & \textbf{39.06\%} (1.30\%) & \textbf{0.53} (0.01) & 7.98 (0.12) & \textbf{61.3} (0.4)\\ 
		DAGgr-pruned (0.8) &  - &4.2422 (0.0578) & 0.08\% (0.01\%) & 6.00\% (0.84\%) & 54.66\% (1.13\%) & 0.40 (0.01) & 7.25 (0.13) & 61.7 (0.4)\\ 
		DAGgr-pruned (1-1/p) &  - &4.6691 (0.0860) & \textbf{0.03\%} (0.01\%) & \textbf{2.62\%} (0.61\%) & 60.91\% (0.96\%) & 0.34 (0.01) & \textbf{6.73} (0.13) & 63.6 (0.6)\\ 
		DAGBag &  - &- & 0.10\% (0.00\%) & 23.12\% (0.61\%) & 84.23\% (0.34\%) & 0.24 (0.01) & - & 92.1 (0.4)\\
		\hline\hline
		GES &  0\% &10.4352 (0.1559) & 1.33\% (0.01\%) & 63.11\% (0.40\%) & \textbf{63.70\%} (1.13\%) & 0.24 (0.01) & 12.64 (0.14) & 185.5 (1.4)\\ 
		PC-algorithm &  0\% &16.3944 (0.0604) & 1.25\% (0.01\%) & 84.81\% (0.39\%) & 98.38\% (0.09\%) & 0.00 (0.00) & 6.28 (0.06) & 211.2 (0.6)\\ 
		Hill climbing &  0\% &11.5562 (0.1491) & 1.33\% (0.01\%) & 66.11\% (0.52\%) & 74.08\% (1.18\%) & 0.17 (0.01) & 10.86 (0.24) & 193.2 (1.6)\\ 
		Max-min hill climbing &  0\% &12.9314 (0.0321) & 1.01\% (0.01\%) & 77.04\% (0.48\%) & 96.19\% (0.26\%) & 0.02 (0.00) & 4.76 (0.13) & 188.0 (0.6)\\ 
		Constrained-MLE &  60.11\% &\textbf{3.4934} (0.0348) & 0.37\% (0.00\%) & 28.00\% (0.34\%) & \textbf{25.79\%} (0.32\%) & \textbf{0.65} (0.00) & 10.83 (0.06) & \textbf{61.6} (0.6)\\ 
		clrdag &  0.65\% &8.8249 (0.0460) & \textbf{0.00\%} (0.00\%) & \textbf{0.00\%} (0.00\%) & 87.70\% (0.35\%) & 0.24 (0.01) & 3.82 (0.10) & 86.8 (0.3)\\ 
		NOTEARS-linear &  0.62\% &\textbf{7.0341} (0.0390) & \textbf{0.00\%} (0.00\%) & \textbf{0.00\%} (0.00\%) & 79.85\% (0.24\%) & \textbf{0.36} (0.00) & 3.20 (0.04) & \textbf{79.0} (0.2)\\ 
		DAGMA-linear &  38.52\% &\textbf{4.4217} (0.0369) & \textbf{0.00\%} (0.00\%) & \textbf{0.05\%} (0.02\%) & \textbf{62.32\%} (0.28\%) & \textbf{0.54} (0.00) & 6.11 (0.05) & \textbf{61.7} (0.3)\\ 
		NOTIME &  0.1\% &10.1950 (0.0014) & \textbf{0.00\%} (0.00\%) & 1.50\% (0.70\%) & 100.00\% (0.00\%) & 0.00 (0.00) & 0.00 (0.00) & 99.0 (0.0)\\ \hline
\end{tabular} }\label{tab:hub_weak_p100}
\end{table}

\subsection{Additional Results for Chain Adjacency Matrix Structures}

\begin{figure}[H]
\centering
\subfloat[Strong-signal (\(p=25\))]{ \label{fig:chain_strong_p25}
	\includegraphics[width=0.23\textwidth]{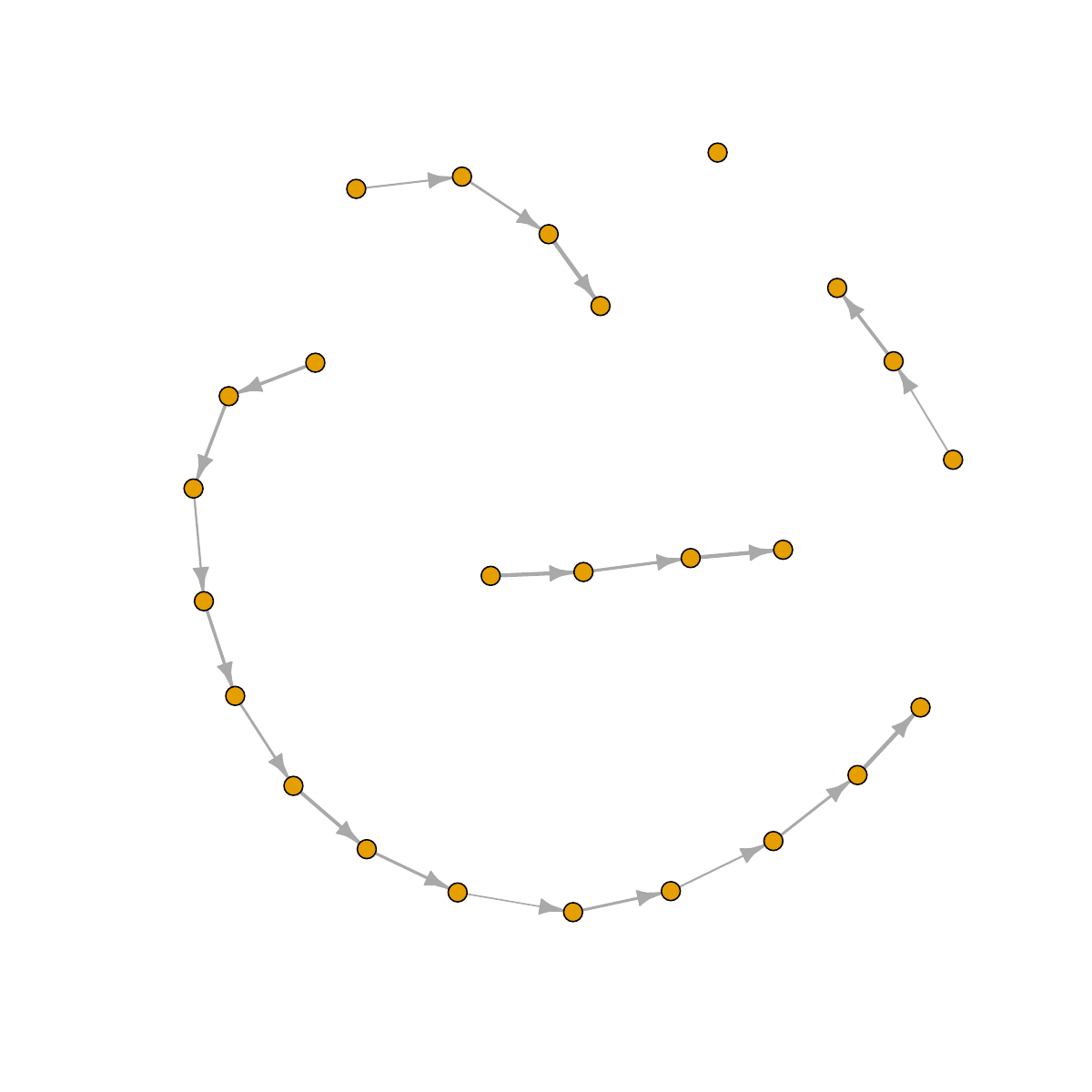}}\ \ 
\subfloat[Weak-signal (\(p=25\))]{ \label{fig:chain_weak_p25}
	\includegraphics[width=0.23\textwidth]{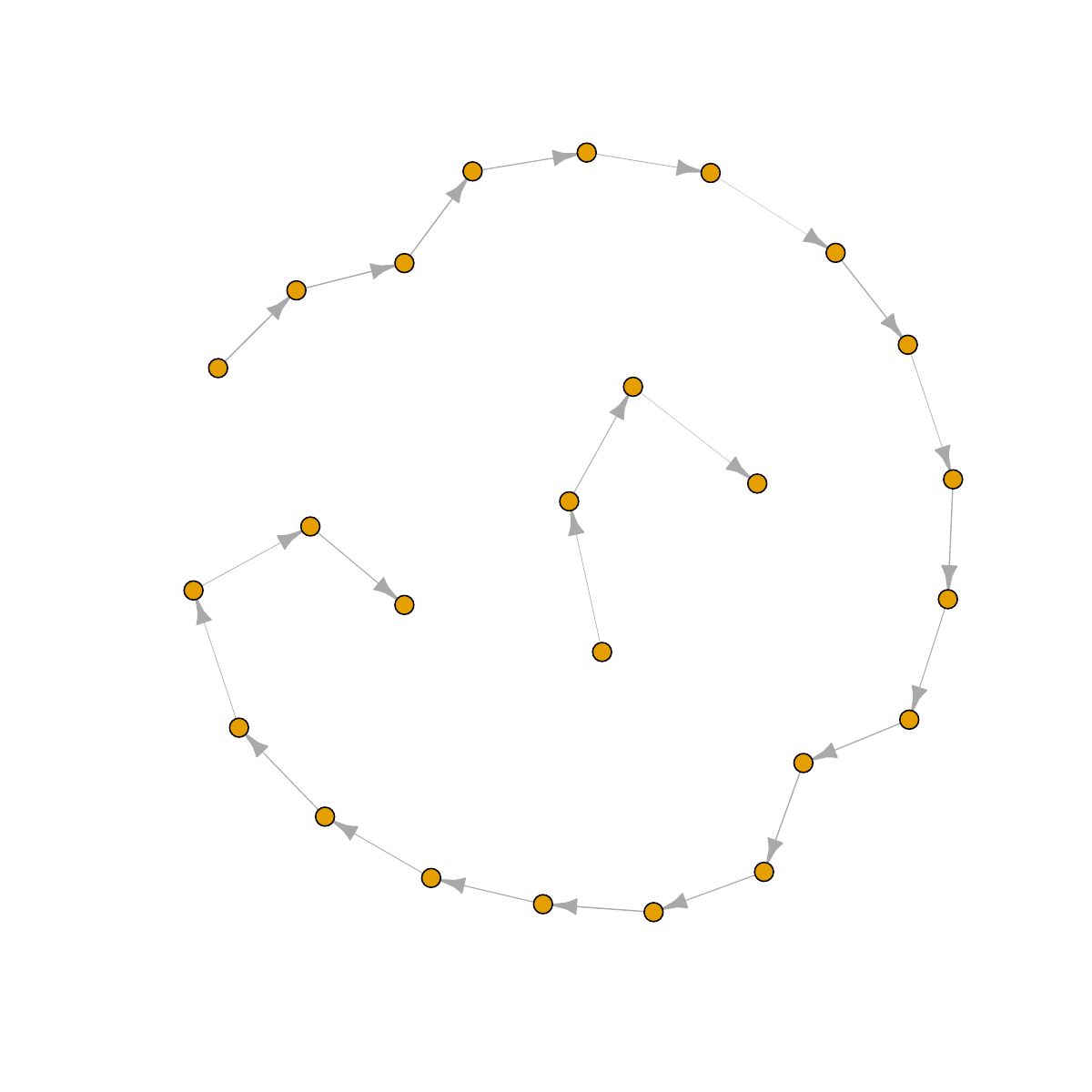}} \ \ 
\subfloat[Strong-signal (\(p=100\))]{ \label{fig:chain_strong_p100}
	\includegraphics[width=0.23\textwidth]{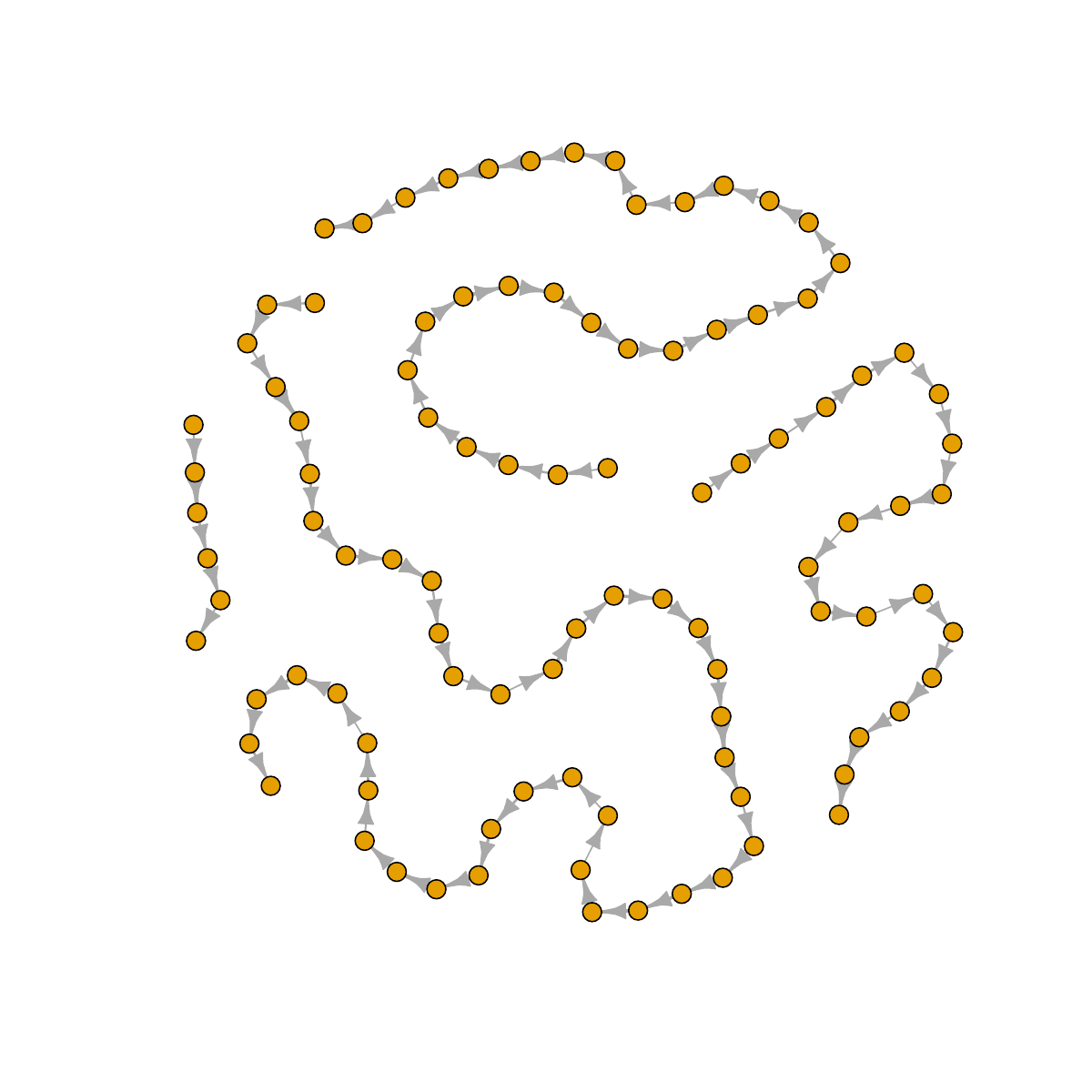}}\ \ 
\subfloat[Weak-signal (\(p=100\))]{ \label{fig:chain_weak_p100}
	\includegraphics[width=0.23\textwidth]{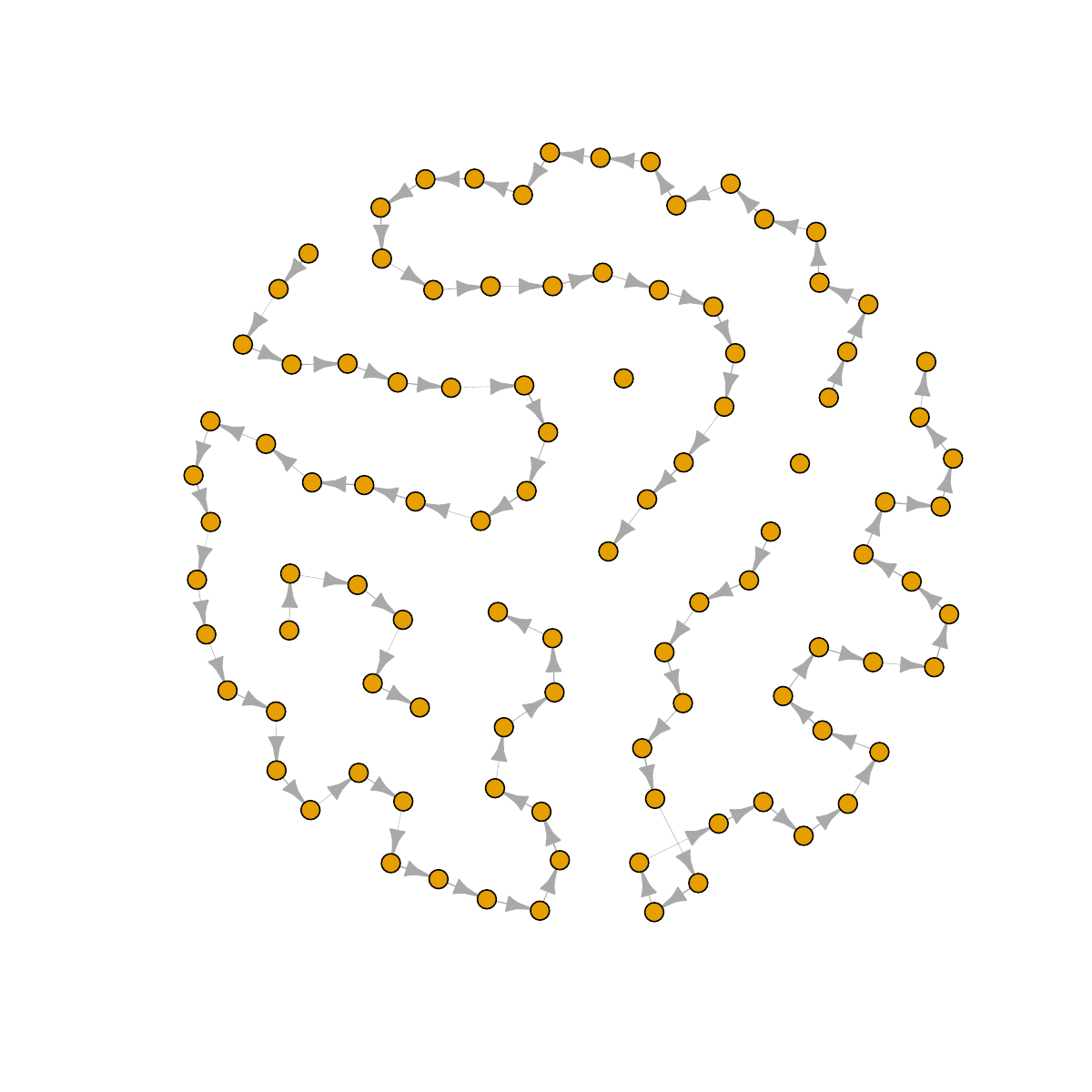}}
\caption{Chain adjacency matrix structures used in the Table~\ref{tab:chain_strong_p25}-Table~\ref{tab:chain_weak_p100}.}
\label{fig:chain_all}
\end{figure}

\begin{table}[H]
\centering
\caption{Performance comparison under the chain adjacency matrix structure in the strong-signal setting with $p=25$ and $n=100$}
\resizebox{1.05\textwidth}{!}{	\begin{tabular}{ccccccccc}
		\hline  &  Average weights &Average squared error & FPR & FDR & FNR& MCC & KL & SHD\\\hline  
		DAGgr-raw &  - &\textbf{0.2198} (0.0157) & 2.95\% (0.18\%) & 16.77\% (1.00\%) & \textbf{0.00\%} (0.00\%) & 0.43 (0.00) & 21.66 (0.07) & 16.8 (1.0)\\ 
		DAGgr-pruned (0.5) &  - &\textbf{0.2208} (0.0207) & 0.09\% (0.01\%) & 1.82\% (0.24\%) & \textbf{0.40\%} (0.10\%) & 0.81 (0.00) & 21.56 (0.07) & 0.6 (0.1)\\ 
		DAGgr-pruned (0.8) &  - &0.3231 (0.0345) & 0.04\% (0.01\%) & 1.14\% (0.17\%) & 1.47\% (0.24\%) & \textbf{0.93} (0.00) & 21.49 (0.08) & \textbf{0.5} (0.1)\\ 
		DAGgr-pruned (1-1/p) &  - &0.5299 (0.0512) & \textbf{0.02\%} (0.00\%) & \textbf{0.48\%} (0.10\%) & 3.45\% (0.40\%) & 0.91 (0.00) & \textbf{21.31} (0.08) & 0.8 (0.1)\\ 
		DAGBag &  - &- & 0.14\% (0.01\%) & 3.63\% (0.28\%) & 2.45\% (0.24\%) & 0.91 (0.00) & - & 0.9 (0.1)\\
		\hline\hline
		GES &  0\% &18.7715 (0.5082) & 3.57\% (0.06\%) & 54.50\% (0.91\%) & 54.08\% (1.53\%) & 0.31 (0.01) & \textbf{16.03} (0.16) & 20.7 (0.3)\\ 
		PC-algorithm &  0\% &8.8676 (0.4725) & \textbf{0.42\%} (0.05\%) & \textbf{10.67\% }(1.26\%) & 10.15\% (1.30\%) & \textbf{0.83} (0.01) & \textbf{12.01} (0.10) & \textbf{2.5} (0.3)\\ 
		Hill climbing &  8.09\% &\textbf{0.8056} (0.0535) & 1.55\% (0.03\%) & 24.32\% (0.54\%) & \textbf{0.67\%} (0.13\%) & 0.72 (0.00) & 22.77 (0.07) & 9.0 (0.2)\\ 
		Max-min hill climbing &  67.32\% &\textbf{0.2445} (0.0284) & \textbf{0.09\%} (0.01\%) & \textbf{2.24\%} (0.22\%) & \textbf{0.58\%} (0.11\%) & \textbf{0.93} (0.00) & 22.36 (0.06) & \textbf{0.6} (0.1)\\ 
		Constrained-MLE &  0.06\% &11.5174 (0.1998) & 12.28\% (0.09\%) & 67.82\% (0.47\%) & 31.65\% (0.39\%) & 0.26 (0.00) & 20.51 (0.09) & 71.3 (0.5)\\ 
		clrdag &  0.25\% &8.7559 (0.0952) & 2.30\% (0.03\%) & 41.64\% (0.48\%) & 29.38\% (0.24\%) & 0.54 (0.00) & 20.25 (0.09) & 13.6 (0.1)\\ 
		NOTEARS-linear &  22.01\% &\textbf{0.7701} (0.0688) & \textbf{0.15\%} (0.02\%) & \textbf{3.74\%} (0.40\%) & \textbf{2.50\%} (0.29\%) & \textbf{0.91} (0.00) & 19.48 (0.06) & \textbf{0.9} (0.1)\\ 
		DAGMA-linear &  2.26\% &1.3543 (0.0883) & 0.61\% (0.03\%) & 12.73\% (0.60\%) & 4.60\% (0.33\%) & 0.80 (0.00) & 21.42 (0.07) & 3.7 (0.2)\\ 
		NOTIME &  0\% &15.2177 (0.2388) & 2.05\% (0.03\%) & 41.00\% (0.62\%) & 42.50\% (0.64\%) & 0.46 (0.01) & \textbf{14.71} (0.14) & 12.6 (0.2)\\
		\hline 
\end{tabular} }\label{tab:chain_strong_p25}
\end{table}

\begin{table}[H]
\centering
\caption{Performance comparison under the chain adjacency matrix structure in the weak-signal setting with $p=25$ and $n=100$}
\resizebox{1.05\textwidth}{!}{		\begin{tabular}{ccccccccc} 
		\hline  &  Average weights &Average squared error & FPR & FDR & FNR& MCC & KL & SHD\\\hline 
		DAGgr-raw &  - &\textbf{1.3449} (0.0239) & 5.81\% (0.08\%) & 48.60\% (0.65\%) & \textbf{23.07\%} (0.50\%) & 0.41 (0.00) & 2.47 (0.03) & 34.2 (0.4)\\ 
		DAGgr-pruned (0.5) &  - &\textbf{1.5287} (0.0332) & 0.98\% (0.03\%) & 21.67\% (0.74\%) & \textbf{48.65\%} (1.00\%) & \textbf{0.47} (0.01) & 1.62 (0.03) & \textbf{13.6} (0.2)\\ 
		DAGgr-pruned (0.8) &  - &2.0718 (0.0290) & 0.22\% (0.02\%) & 6.35\% (0.53\%) & 75.04\% (0.95\%) & 0.26 (0.01) & 1.00 (0.03) & 17.8 (0.2)\\ 
		DAGgr-pruned (1-1/p) &  - &2.4135 (0.0217) & \textbf{0.05\%} (0.01\%) & \textbf{1.76\%} (0.31\%) & 88.13\% (0.70\%) & 0.13 (0.01) & \textbf{0.54} (0.03) & 20.4 (0.1)\\ 
		DAGBag &  - &- & 0.58\% (0.02\%) & 17.47\% (0.60\%) & 58.26\% (0.56\%) & 0.45 (0.01) & - & 13.9 (0.1)\\
		\hline\hline
		GES &  4.28\% &3.9028 (0.0596) & 3.85\% (0.05\%) & 56.99\% (0.72\%) & 63.26\% (0.83\%) & 0.25 (0.01) & 3.54 (0.03) & 26.9 (0.3)\\ 
		PC-algorithm &  26.56\% &2.2634 (0.0471) & 1.55\% (0.03\%) & 35.70\% (0.72\%) & \textbf{47.57\%} (0.83\%) & \textbf{0.48} (0.01) & 2.69 (0.02) & \textbf{14.8} (0.2)\\ 
		Hill climbing &  5.77\% &\textbf{1.9875} (0.0358) & 2.67\% (0.05\%) & 40.50\% (0.72\%) & \textbf{36.59\%} (0.58\%) & \textbf{0.47} (0.00) & 3.58 (0.03) & 20.1 (0.3)\\ 
		Max-min hill climbing &  27.84\% &\textbf{1.5493} (0.0325) & 1.22\% (0.03\%) & 27.12\% (0.66\%) & \textbf{39.00\%} (0.55\%) & \textbf{0.56} (0.01) & 3.02 (0.02) & \textbf{13.0} (0.2)\\ 
		Constrained-MLE &  6.36\% &2.7933 (0.0328) & 1.14\% (0.02\%) & 32.83\% (0.68\%) & 70.57\% (0.48\%) & 0.30 (0.01) & 2.39 (0.03) & 17.4 (0.1)\\ 
		clrdag &  0.23\% &2.8097 (0.0133) & \textbf{0.03\%} (0.01\%) & \textbf{6.67\%} (1.00\%) & 98.39\% (0.18\%) & 0.04 (0.00) & \textbf{0.22} (0.02) & 22.6 (0.0)\\ 
		NOTEARS-linear &  7.03\% &2.4174 (0.0351) & \textbf{0.56\%} (0.02\%) & \textbf{23.07\%} (0.85\%) & 77.59\% (0.55\%) & 0.28 (0.01) & \textbf{1.31} (0.02) & 17.9 (0.1)\\ 
		DAGMA-linear &  9.94\% &2.6212 (0.0399) & 1.01\% (0.02\%) & 27.71\% (0.64\%) & 68.17\% (0.60\%) & 0.32 (0.01) & \textbf{2.32} (0.03) & 16.7 (0.1)\\ 
		NOTIME &  11.98\% &\textbf{2.2061} (0.0404) & \textbf{0.83\%} (0.03\%) & \textbf{24.02\%} (0.77\%) & 62.50\% (0.63\%) & 0.39 (0.01) & 2.34 (0.02) & \textbf{15.4} (0.2)\\
		\hline
\end{tabular} }\label{tab:chain_weak_p25}
\end{table}

\begin{table}[H]
\centering
\caption{Performance comparison under the chain adjacency matrix structure in the strong-signal setting with $p=100$ and $n=250$}
\resizebox{1.05\textwidth}{!}{	\begin{tabular}{ccccccccc}
		\hline  &  Average weights &Average squared error & FPR & FDR & FNR& MCC & KL & SHD\\\hline
		DAGgr-raw &  - &0.1900 (0.0241) & 0.03\% (0.01\%) & 2.13\% (0.50\%) & \textbf{0.04\%} (0.01\%) & 0.77 (0.00) & 98.69 (0.05) & 3.4 (0.8)\\
		DAGgr-pruned (0.5) &  - &\textbf{0.1882} (0.0241) & \textbf{0.01\%} (0.00\%) & \textbf{1.15\%} (0.08\%) & \textbf{0.04\%} (0.01\%) & \textbf{0.98} (0.00) & 98.69 (0.05) & \textbf{1.1} (0.1)\\
		DAGgr-pruned (0.8) &  - &\textbf{0.1882} (0.0241) & \textbf{0.01\%} (0.00\%) & \textbf{1.15\%} (0.08\%) & \textbf{0.04\%} (0.01\%) & \textbf{0.98} (0.00) & 98.69 (0.05) & \textbf{1.1} (0.1)\\
		DAGgr-pruned (1-1/p) &  - &0.1984 (0.0266) & \textbf{0.01\%} (0.00\%) & \textbf{1.15\%} (0.08\%) & 0.05\% (0.02\%) & \textbf{0.98} (0.00) & \textbf{98.68} (0.05) & 1.2 (0.1)\\
		DAGBag &  - &- & 0.10\% (0.00\%) & 10.28\% (0.27\%) & 10.27\% (0.27\%) & 0.88 (0.00) & - & 10.3 (0.3)\\
		\hline\hline
		GES &  0\% &84.3782 (2.4629) & 1.01\% (0.01\%) & 57.24\% (0.72\%) & 47.69\% (1.29\%) & 0.38 (0.01) & 86.81 (0.26) & 99.0 (1.2)\\
		PC-algorithm &  0\% &61.8079 (2.3484) & \textbf{0.27\%} (0.02\%) & \textbf{26.03\%} (1.59\%) & 26.00\% (1.61\%) & 0.72 (0.02) & \textbf{71.65} (1.17) & \textbf{26.3} (1.6)\\
		Hill climbing &  0.7\% &\textbf{1.4332} (0.0791) & 0.56\% (0.01\%) & 30.87\% (0.31\%) & \textbf{0.27\%} (0.04\%) & \textbf{0.73} (0.00) & 99.50 (0.06) & 55.0 (0.6)\\
		Max-min hill climbing &  99.3\% &\textbf{0.1869} (0.0240) & \textbf{0.01\%} (0.00\%) & \textbf{1.15\%} (0.08\%) & \textbf{0.04\%} (0.01\%) & \textbf{0.98} (0.00) & 98.70 (0.05) & \textbf{1.1} (0.1)\\
		Constrained-MLE &  0\% &67.7678 (0.3708) & 7.64\% (0.02\%) & 85.66\% (0.18\%) & 38.67\% (0.17\%) & 0.18 (0.00) & 82.97 (0.17) & 748.8 (1.6)\\
		clrdag &  0\% &61.1686 (0.1008) & 0.95\% (0.00\%) & 58.70\% (0.16\%) & 42.72\% (0.06\%) & 0.44 (0.00) & 81.24 (0.09) & 93.4 (0.3)\\
		NOTEARS-linear &  0\% &\textbf{14.5058} (0.2740) & \textbf{0.18\%} (0.00\%) & \textbf{16.48\%} (0.31\%) & \textbf{12.92\%} (0.24\%) & \textbf{0.82} (0.00) & 76.42 (0.09) & \textbf{18.5} (0.4)\\
		DAGMA-linear &  0\% &25.0647 (0.2766) & 0.38\% (0.00\%) & 29.90\% (0.30\%) & 22.61\% (0.21\%) & 0.68 (0.00) & \textbf{74.84} (0.21) & 39.1 (0.4)\\
		NOTIME &  0\% &79.7030 (0.5253) & 0.65\% (0.00\%) & 51.83\% (0.33\%) & 48.13\% (0.27\%) & 0.45 (0.00) & \textbf{58.51} (0.23) & 69.0 (0.4)\\
		\hline
\end{tabular} }\label{tab:chain_strong_p100}
\end{table}

\begin{table}[H]
\centering
\caption{Performance comparison under the chain adjacency matrix structure in the weak-signal setting with $p=100$ and $n=250$}
\resizebox{1.05\textwidth}{!}{		\begin{tabular}{ccccccccc}
		\hline &  Average weights &Average squared error & FPR & FDR & FNR& MCC & KL & SHD\\\hline
		DAGgr-raw &  - &\textbf{4.9925} (0.0965) & 1.81\% (0.01\%) & 57.75\% (0.44\%) & \textbf{15.80\%} (0.31\%) & 0.46 (0.00) & 10.02 (0.08) & 159.3 (0.9)\\
		DAGgr-pruned (0.5) &  - &\textbf{5.7528} (0.1262) & 0.61\% (0.02\%) & 38.13\% (1.29\%) & \textbf{44.61\%} (0.98\%) & 0.42 (0.01) & 6.91 (0.10) & 81.1 (1.5)\\
		DAGgr-pruned (0.8) &  - &7.1119 (0.0861) & 0.16\% (0.01\%) & 10.99\% (0.75\%) & 69.13\% (0.90\%) & 0.25 (0.01) & 4.58 (0.13) & 75.4 (0.5)\\
		DAGgr-pruned (1-1/p) &  - &8.0791 (0.0752) & \textbf{0.07\%} (0.01\%) & \textbf{5.21\%} (0.63\%) & 80.54\% (0.78\%) & 0.16 (0.01) & \textbf{3.13} (0.12) & 80.5 (0.4)\\
		DAGBag &  - &- & 0.17\% (0.00\%) & 22.38\% (0.41\%) & 52.05\% (0.33\%) & \textbf{0.54} (0.00) &- & \textbf{55.5} (0.3)\\
		\hline\hline
		GES &  0\% &13.7918 (0.1251) & 1.66\% (0.01\%) & 70.20\% (0.38\%) & 55.97\% (0.50\%) & 0.27 (0.00) & 13.64 (0.04) & 173.4 (0.9)\\
		PC-algorithm &  16.43\% &8.9813 (0.2416) & 0.88\% (0.01\%) & 53.02\% (0.70\%) & \textbf{42.63\%} (1.21\%) & 0.43 (0.01) & 10.47 (0.03) & 98.1 (1.2)\\
		Hill climbing &  0\% &\textbf{6.5614} (0.0600) & 1.36\% (0.01\%) & 55.41\% (0.35\%) & \textbf{27.62\%} (0.26\%) & \textbf{0.45} (0.00) & 13.97 (0.04) & 144.3 (0.9)\\
		Max-min hill climbing &  29.99\% &\textbf{5.4346} (0.0631) & 0.69\% (0.01\%) & 42.88\% (0.33\%) & \textbf{30.73\%} (0.29\%) & \textbf{0.53} (0.00) & 11.91 (0.03) & 80.1 (0.5)\\
		Constrained-MLE &  34.5\% &10.6230 (0.0626) & 1.05\% (0.01\%) & 61.75\% (0.35\%) & 54.48\% (0.26\%) & 0.34 (0.00) & 11.28 (0.03) & 116.0 (0.6)\\
		clrdag &  0\% &10.3491 (0.0062) & \textbf{0.00\%} (0.00\%) & \textbf{0.75\%} (0.43\%) & 99.79\% (0.03\%) & 0.01 (0.00) & \textbf{0.08} (0.01) & 93.8 (0.0)\\
		NOTEARS-linear &  0\% &7.4157 (0.0563) & \textbf{0.07\%} (0.00\%) & \textbf{19.05\%} (0.51\%) & 76.82\% (0.25\%) & 0.36 (0.00) & \textbf{3.96} (0.03) & \textbf{72.2} (0.2)\\
		DAGMA-linear &  11.42\% &\textbf{7.0423} (0.0743) & \textbf{0.13\%} (0.00\%) & \textbf{25.43\%} (0.50\%) & 66.46\% (0.27\%) & \textbf{0.45} (0.00) & \textbf{7.15} (0.04) & \textbf{62.5} (0.3)\\
		NOTIME &  7.66\% &7.2459 (0.0725) & 0.15\% (0.00\%) & 26.88\% (0.43\%) & 65.04\% (0.29\%) & \textbf{0.45} (0.00) & 7.70 (0.04) & \textbf{61.3} (0.3)\\
		\hline
\end{tabular} }\label{tab:chain_weak_p100}
\end{table}

\end{document}